\documentclass{article}

 \usepackage[main, final]{neurips_2026}
\usepackage[utf8]{inputenc} % allow utf-8 input
\usepackage[T1]{fontenc}    % use 8-bit T1 fonts
\usepackage{url}            % simple URL typesetting
\usepackage{booktabs}       % professional-quality tables
\usepackage{amsfonts}       % blackboard math symbols
\usepackage{nicefrac}       % compact symbols for 1/2, etc.
\usepackage{microtype}      % microtypography
\usepackage{xcolor}         % colors

\usepackage{amsmath,amssymb,amsthm,mathtools,bm}
\usepackage{thmtools}
\usepackage{thm-restate}
\usepackage{longtable,array}
\usepackage{enumitem}
\usepackage{subcaption}
\usepackage{authblk}
\usepackage{soul}
\usepackage{makecell}
\usepackage{hyperref}

\theoremstyle{plain}
\newtheorem{theorem}{Theorem}[section]
\newtheorem{proposition}[theorem]{Proposition}
\newtheorem{lemma}[theorem]{Lemma}
\newtheorem{corollary}[theorem]{Corollary}

\newtheorem{remark}[theorem]{Remark}
\newtheorem{observation}[theorem]{Observation}
\newtheorem{definition}[theorem]{Definition}

\newcommand{\C}{\mathbb{C}}
\newcommand{\Enc}{\mathcal{E}}
\newcommand{\Dec}{\mathcal{D}}
\newcommand{\calE}{\mathcal{E}}
\newcommand{\calD}{\mathcal{D}}

\newcommand{\eps}{\varepsilon}
\newcommand{\ket}[1]{|#1\rangle}
\newcommand{\bra}[1]{\langle #1|}
\newcommand{\braket}[2]{\langle #1\,|\,#2\rangle}
\newcommand{\ketbra}[2]{|#1\rangle\!\langle #2|}
\newcommand{\proj}[1]{\ket{#1}\!\bra{#1}}
\newcommand{\norm}[1]{\left\lVert #1\right\rVert}
\newcommand{\abs}[1]{\left\lvert #1\right\rvert}
\newcommand{\krank}{\operatorname{rank}_{\mathrm{Kraus}}}

\newcommand{\tcb}[1]{\textcolor{black}{#1}}
\newcommand{\myBlue}{\color{black}}

\DeclareMathOperator{\Tr}{Tr}
\DeclareMathOperator{\rank}{rank}
\DeclareMathOperator{\ran}{ran}
\DeclareMathOperator{\im}{im}
\DeclareMathOperator{\supp}{supp}
\DeclareMathOperator{\Span}{span}

\DeclareMathOperator{\Ad}{Ad}

\AtBeginDocument{%
  \bibpunct{[}{]}{;}{a}{,}{,}%
}
\DeclareRobustCommand{\citeMy}[1]{\citep{#1}}
\title{Toward the Goldilocks Blind Compression\\of Quantum States}

\author{%
  \textbf{Hyunho Cha$^{1}$\qquad Chae-Yeun Park$^{2}$\qquad Jungwoo Lee$^{1}$}\thanks{Corresponding author: \texttt{junglee@snu.ac.kr}}\\
  $^{1}$NextQuantum and Department of Electrical and Computer Engineering, Seoul National University\\
  $^{2}$School of Integrated Technology and Department of Quantum Information, Yonsei University\\
}

\begin{document}

\maketitle

\begin{abstract}
Quantum autoencoders (QAEs) are learning architectures that compress quantum data into a low-dimensional latent state while preserving the information needed for reconstruction. We study blind single-copy compression of quantum states through a $k$-qubit bottleneck and investigate the minimal circuit width required to attain the information-theoretic optimum under average infidelity.
Between the conventional architecture, which is narrow but nonuniversal, and fully general \emph{completely positive and trace preserving} (CPTP) realizations, which are universal but overparameterized, we identify a balanced regime.
We prove that for every distribution of pure $n$-qubit states, there exists a QAE with $k$ encoder ancillas and $n$ decoder ancillas that achieves the optimal fidelity over all CPTP encoder--decoder pairs.
The encoder-side statement is sharp in that we construct source families for which every optimal scheme necessarily uses at least $k$ encoder ancillas, thereby determining the universal encoder threshold exactly.
On the decoder side, we show that isometric decoders are optimal for several analytically tractable source families, but we also exhibit an explicit counterexample demonstrating that decoder isometry is not universally sufficient.
Nevertheless, numerical experiments indicate that the performance gap is practically negligible.
\end{abstract}

\section{Introduction}

The rapid development of controllable quantum platforms has brought into focus a central challenge in quantum physics, namely how researchers can manipulate and preserve information encoded in exponentially large Hilbert spaces with finite and noisy resources. Manipulation of quantum states serves as the structural foundation for diverse fields, ranging from quantum simulation of many-body dynamics and chemistry to quantum computation, communication, and sensing \citeMy{lloyd1996universal, nielsen2010quantum, bernien2017probing, degen2017quantum, feynman2018simulating, preskill2018quantum, bharti2022noisy}.
Yet the same exponential growth of state space that fuels quantum advantage also makes quantum memory and coherent processing increasingly scarce commodities.
For near-term devices in particular, the available number of high-fidelity qubits and coherent operations remains a dominant bottleneck \citeMy{preskill2018quantum, bharti2022noisy}.
This necessitates architectures that use the quantum hardware at an appropriate scale by being neither underpowered for the task nor overprovisioned in ways that amplify noise.

Compression is the canonical response to resource scarcity in classical information theory \citeMy{shannon1948mathematical}, and quantum information theory has long established that quantum data can likewise be compressed. The quantum noiseless coding theorem shows that an i.i.d. source described by a density matrix $\rho$ can be asymptotically and faithfully compressed to $S(\rho):=-\Tr(\rho\log\rho)$ qubits per signal \citeMy{jozsa1994new, schumacher1995quantum, dur2001visible, koashi2001compressibility, koashi2001teleportation, winter2002compression, khanian2022general}, with refinements that address converses and general fidelity limits \citeMy{barnum1996general}, universality with minimal prior knowledge \citeMy{jozsa1998universal, hayashi2002quantum}, and extensions to mixed-state sources and different notions of side information or visibility of the signal identity \citeMy{horodecki1998limits, hayden2002trading}.
Experimental demonstrations of this compression have also been realized \citeMy{mitsumori2003experimental}. At the same time, many emerging use cases operate far from the asymptotic i.i.d. regime and often call for single-copy, distribution-dependent compression under experimentally meaningful distortion measures \citeMy{barnum2000quantum, datta2012quantum, datta2013one, wilde2013quantum, salek2018quantum, hadiashar2020entanglement}.

A parallel set of ideas has emerged from machine learning. Classical autoencoders learn low-dimensional latent representations by jointly training an encoder and decoder to minimize reconstruction error, enabling data compression, denoising, and feature extraction \citeMy{hinton2006reducing, vincent2008extracting, bengio2013representation}. Variational autoencoders extend this paradigm into probabilistic latent-variable models that support sampling and regularization of representations \citeMy{kingma2013auto, rezende2014stochastic, zhao2017infovae}. These concepts have begun to influence quantum information processing through variational quantum algorithms and quantum machine learning models that optimize physically motivated cost functions using parameterized quantum circuits \citeMy{schuld2015introduction, biamonte2017quantum, mcclean2018barren, benedetti2019parameterized, khoshaman2019quantum, sim2019expressibility, cerezo2021variational, holmes2022connecting}.

Within this landscape, quantum autoencoders (QAEs) were proposed as a route to compression of quantum data that may not admit efficient classical descriptions \citeMy{romero2017quantum}. In a QAE, a trainable encoder maps an $n$-qubit input state into a $k$-qubit latent register ($k<n$) plus discarded degrees of freedom. A decoder then aims to reconstruct the original state from the latent register.
QAEs have been studied theoretically and deployed in diverse physical contexts \citeMy{romero2017quantum, ding2019experimental, lamata2019quantum, pepper2019experimental, bondarenko2020quantum, huang2020realization, cao2021noise, ma2023compression}.
These developments demonstrate that QAE-style compression is becoming an architectural primitive for quantum technologies.

A major open issue, however, concerns resource optimality. Most QAE implementations adopt a unitary--trace architecture with a circuit acting on the system and a small number of ancillas (or auxiliary quantum systems), which can severely restrict the set of realizable encoder and decoder channels. From the perspective of quantum channel theory, the most general encoder and decoder are \emph{completely positive and trace preserving} (CPTP) maps \citeMy{stinespring1955positive, kraus1971general, jamiolkowski1972linear, jiang2013channel}.
Indeed, recent fully quantum variational autoencoder frameworks have proposed universal constructions capable of representing arbitrary CPTP maps by introducing sufficiently large ancillas \citeMy{wang2025quantum}. While attractive in principle, this ``maximally expressive'' design collides with noisy reality and high training costs \citeMy{mcclean2018barren, preskill2018quantum, bharti2022noisy}.
This motivates a fundamental question: for a given objective, what is the minimal circuit width required to achieve the information-theoretic optimum?

In this work, we investigate this question for blind, single-copy quantum state compression under the \emph{infidelity} loss \citeMy{mahler2013adaptive, huang2020realization, cerezo2021variational}.
In essence, for the infidelity objective there is a minimally universal architecture wide enough to reach the information-theoretic optimum, yet narrow enough to avoid the redundant ancilla overhead.

On the encoder side, we determine the ancilla cost sharply. Specifically, we prove matching sufficient and necessary conditions on the encoder ancilla size. In this sense, the encoder ancilla requirement is characterized exactly. This singles out a natural minimal encoder architecture that is expressive enough to attain the optimum while avoiding unnecessary overhead.

We further provide both analytical and numerical evidence for a conceptually appealing practical principle. Although our counterexample demonstrates that, in blind single-copy compression with infidelity loss, an isometric decoder need not always attain the theoretical optimum, numerical results suggest that restricting the decoder to be isometric nevertheless yields empirically near-optimal performance. We prove optimality of isometric decoders for several sources with closed-form distributions, and we find consistent numerical support on datasets constructed from MNIST-encoded quantum states \citeMy{lecun2002gradient} and additional physical benchmarks, where the observed performance gap is negligible. This makes isometric decoders particularly attractive in practice.
They offer near-optimal performance while reducing model size and computational overhead.

\section{Preliminaries and prior work}
For a positive integer $d$, let $H_d \cong \mathbb{C}^d$ be a $d$-dimensional Hilbert space and let $L(H_d)$ denote the linear operators on $H_d$. A qubit is a two-dimensional quantum system, so an $n$-qubit register has Hilbert space $H_{2^n} \cong (\mathbb{C}^2)^{\otimes n}$ and hence dimension $2^n$. For quantum registers $A,B,\ldots$, we write $H_A,H_B,\ldots$ for the associated Hilbert spaces and $H_{AB} := H_A \otimes H_B$ for the composite system. If the same physical space is viewed under two factorizations, we write for example $AB \equiv CD$. We use the computational basis $\{ |x\rangle : x \in \{0,1\}^n\}$ of an $n$-qubit system (equivalently $\{ |j\rangle : 0 \le j \le 2^n-1\}$ after fixing the binary-to-integer identification). In Dirac notation, $|\psi\rangle$ is a \emph{ket}, $\langle \psi| := |\psi\rangle^{\dagger}$ is the corresponding \emph{bra}, and $\langle \phi|\psi\rangle$ is the inner product. State vectors are always taken to have unit norm.

A pure state is represented by a unit vector $|\psi\rangle$, or equivalently by the rank-one \emph{density operator} $|\psi\rangle\langle \psi|$. More generally, a mixed state on a Hilbert space $H$ is a density operator $\rho \in L(H)$, meaning $\rho \ge 0$ and $\Tr(\rho)=1$. Tensor products are used for spaces, state vectors, and operators. For a bipartite state $\rho_{AB} \in L(H_{AB})$, the reduced state on $A$ is $\rho_A := \Tr_B(\rho_{AB})$, where $\Tr_B$ is the partial trace. Operationally, this is the operation of discarding subsystem $B$. Likewise $\rho_B := \Tr_A(\rho_{AB})$.

A \emph{superoperator} is a linear map between operator spaces. The physical maps considered here are quantum channels, namely CPTP maps $\Phi : L(H_A) \to L(H_B)$. Trace preservation means $\Tr[\Phi(X)] = \Tr[X]$, while complete positivity means that $\Phi \otimes \mathrm{id}_R$ sends positive operators to positive operators for every auxiliary register $R$. Closed-system evolution is unitary: if $U \in \mathrm{U}(d)$, the unitary group on $H_d$, then the induced channel is the conjugation map
\[
\Ad_U(\rho) := U\rho U^{\dagger}.
\]
More generally, every CPTP map admits a \emph{Kraus representation}
\[
\Phi(\rho) = \sum_{i=1}^r K_i \rho K_i^{\dagger},
\qquad
\sum_{i=1}^r K_i^{\dagger} K_i = I,
\]
and, equivalently, a Stinespring realization by adjoining an ancilla in a fixed state, applying a joint unitary, and partially tracing out an environment. The smallest possible number $r$ of Kraus operators is the \emph{Kraus rank} of $\Phi$.

A more detailed mathematical formulation can be found in Appendix~\ref{sec:appendix_prelim}.

\begin{figure}
    \centering
    \includegraphics[width=0.6\linewidth]{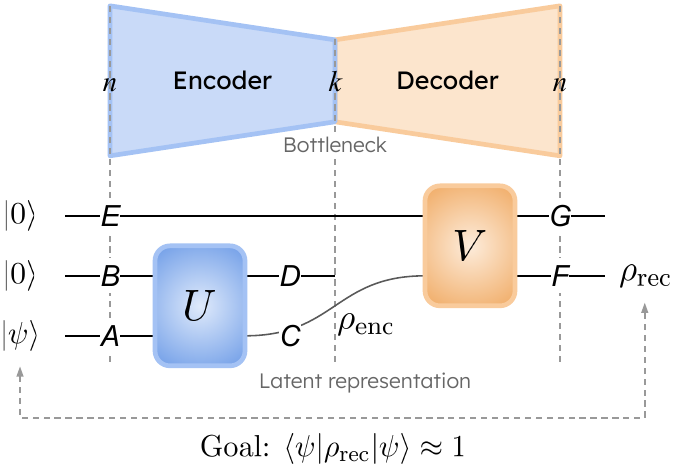}
    \caption{General structure of a QAE that compresses an $n$-qubit input state $|\psi\rangle$ into a $k$-qubit latent system and subsequently reconstructs the state using a decoder. \tcb{Labels $A$--$G$ identify the individual quantum registers.}}
    \label{fig:schematic}
\end{figure}

% \begin{wrapfigure}{r}{0.5\textwidth} % Use 'r' (lowercase) not 'R'
  % \vspace{-15pt} % Adjust this to move the image up/down if it looks weird
%   \centering
%   \includegraphics[width=0.48\textwidth]{figures/schematic/our_qae.pdf}
%   \caption{QAE Structure.}
%   \label{fig:schematic}
%   \vspace{-10pt} % Adjust this to tighten the space below the caption
% \end{wrapfigure}

Blind single-copy compression means that the encoder receives one copy of an unknown source state $|\psi\rangle$, drawn from a source distribution $\mu$, and must compress it to a smaller $k$-qubit register without being given a classical label for $|\psi\rangle$. The decoder then attempts to reconstruct the original state from that single compressed quantum register. When a pure state is said to be Haar-random, it is drawn from the unitarily invariant probability measure on the unit sphere of the relevant Hilbert space (or subspace) \citeMy{haar1933massbegriff}. A schematic depiction of the blind compression setting is provided in Figure~\ref{fig:schematic}.

\paragraph{Prior architectures.}
Among the many QAE constructions studied in the literature, two architectures are especially relevant for the resource question considered here. The \emph{conventional QAE} introduced in \citeMy{romero2017quantum, huang2020realization, bravo2021quantum, cao2021noise, ngairangbam2022anomaly, srikumar2022clustering, wu2024quantum} uses an $n$-qubit unitary on the input register, discards $n-k$ qubits to obtain a $k$-qubit latent register, and reconstructs by adjoining $n-k$ fresh ancillas followed by another $n$-qubit unitary. Anticipating the $(n,k,n_B,n_E)$ notation introduced in Section~\ref{sec:main_results}, this is the $(n,k,0,n-k)$ architecture. Its main virtue is its small circuit width, since both the encoder and decoder act on only $n$ qubits. However, it constrains the encoder to be a unitary-on-system followed by a partial trace with no additional encoder ancilla, and it constrains the decoder to be an isometric channel. Accordingly, it is not universal for blind single-copy compression through a $k$-qubit bottleneck.

At the opposite end of the expressivity spectrum, \citeMy{wang2025quantum} proposed the \emph{$\zeta$-QVAE}, a quantum variational autoencoder with mixed latent representations and an explicit latent-space regularization term. In that framework, both the encoder and decoder are modeled as general CPTP maps via \emph{Stinespring dilations} \citeMy{stinespring1955positive} with sufficiently large ancillas, so the construction is universal. Specialized to the compression setting studied here, one convenient realization corresponds to $n_B=n+2k$ encoder ancillas and $n_E=2n+k$ decoder ancillas, so that both associated unitaries act on $2(n+k)$ qubits. Thus the conventional QAE and the $\zeta$-QVAE represent the two extremes relevant to our analysis. The former is resource-efficient but nonuniversal, whereas the latter is universal but substantially wider than may be necessary if one is only interested in infidelity. This contrast motivates the main question of the present work, namely whether there exists a strictly smaller universal architecture for blind single-copy compression under the infidelity objective.

\section{Main results}
\label{sec:main_results}

\myBlue
\paragraph{Overview.} In this section, we establish four primary results. First, $k$ encoder ancilla qubits and $n$ decoder ancilla qubits always suffice to match the best possible CPTP encoder--decoder pair. Second, the encoder count $k$ is unavoidable in the worst case. Third, the smallest physical decoder, an isometric decoder with $n-k$ ancilla qubits, is not always optimal. Fourth, despite this counterexample, isometric decoders are near-optimal when the source is concentrated near a $2^k$-dimensional subspace.
\color{black}

For an $n$-qubit source distribution $\mu$ over pure states $\ket{\psi}\in H_{2^n}$ and integers $n>k\ge 1$, we evaluate an encoder $\Enc:L(H_{2^n})\to L(H_{2^k})$ and a decoder $\Dec:L(H_{2^k})\to L(H_{2^n})$ by the average fidelity
\begin{equation}\label{eq:objective}
F_\mu(\Enc,\Dec)
:=
\mathbb{E}_{\ket{\psi}\sim \mu}\bigl[\bra{\psi}(\Dec\circ \Enc)(\proj{\psi})\ket{\psi}\bigr].
\end{equation}
The corresponding average infidelity is $L_\mu(\Enc,\Dec):=1-F_\mu(\Enc,\Dec)$.

\begin{definition}[$(n,k,n_B,n_E)$-QAE]\label{def:qae}
Fix $n>k\ge 1$. An $(n,k,n_B,n_E)$-QAE is specified by two unitaries $(U,V)$ and acts as follows.
\begin{enumerate}[label=(\roman*),leftmargin=2.2em]
\item Prepare the input state $\ket{\psi}$ on an $n$-qubit register $A$ and $\ket{0}^{\otimes n_B}$ on an $n_B$-qubit ancilla register $B$. Apply an $(n+n_B)$-qubit unitary $U$ and reinterpret $AB\equiv CD$, where $C$ has $k$ qubits and $D$ has $n+n_B-k$ qubits. The encoded state is
\[
\rho_{\mathrm{enc}}:=\Tr_D\bigl(U(\proj{\psi}_A\otimes \proj{0^{n_B}}_B)U^\dagger\bigr)\in L(H_C).
\]
\item Prepare $\ket{0}^{\otimes n_E}$ on an $n_E$-qubit ancilla register $E$. Apply a $(k+n_E)$-qubit unitary $V$ to $CE$ and reinterpret $CE\equiv FG$, where $F$ has $n$ qubits and $G$ has $k+n_E-n$ qubits. The reconstructed state is
\[
\rho_{\mathrm{rec}}:=\Tr_G\bigl(V(\rho_{\mathrm{enc}}\otimes \proj{0^{n_E}}_E)V^\dagger\bigr)\in L(H_F).
\]
\end{enumerate}
\end{definition}

More generally, if the encoder ancilla $B$ and decoder ancilla $E$ are initialized to
fixed density operators $\sigma_B \in L(H_B)$ and $\tau_E \in L(H_E)$, respectively,
then the induced encoder and decoder channels are
\[
\Enc_{U,\sigma_B}(\rho_A)
:=
\operatorname{Tr}_D\!\left[
U(\rho_A \otimes \sigma_B)U^\dagger
\right],
\qquad
\Dec_{V,\tau_E}(\rho_C)
:=
\operatorname{Tr}_G\!\left[
V(\rho_C \otimes \tau_E)V^\dagger
\right].
\]

\begin{restatable}{lemma}{AncillaAlwaysPure}
\label{lem:ancilla-state-reduction}
Fix $n > k \ge 1$, ancilla counts $n_B \ge 0$ and $n_E \ge n-k$, and a source
distribution $\mu$ over pure $n$-qubit states. Then
\[
\max_{U,V,\sigma_B,\tau_E}
F_\mu\!\left(\Enc_{U,\sigma_B}, \Dec_{V,\tau_E}\right)
=
\max_{U,V}
F_\mu\!\left(
\Enc_{U,|0^{n_B}\rangle\langle 0^{n_B}|},
\Dec_{V,|0^{n_E}\rangle\langle 0^{n_E}|}
\right),
\]
where on the left the maximum is over all unitaries $U,V$ of the appropriate sizes
and all density operators $\sigma_B \in L(H_B)$ and $\tau_E \in L(H_E)$, and on the
right the maximum is over the ordinary $(n,k,n_B,n_E)$-QAEs of Definition~\ref{def:qae}.
\end{restatable}

\begin{proof}
See Appendix~\ref{app:proof_pure_ancillas_suffice}.
\end{proof}

Lemma~\ref{lem:ancilla-state-reduction} implies that allowing arbitrary fixed mixed ancilla states does not change the optimal average fidelity.
Thus, without loss of generality, we may fix the encoder and decoder ancilla states to all-zero states for the purpose of infidelity minimization. We maintain this convention throughout the remainder of this work.

\begin{remark}
\label{rem:no-contradiction-mixed-reference}
At first glance, \citeMy{ma2024quantum} may appear to contradict Lemma~\ref{lem:ancilla-state-reduction}, since they report an advantage
from using mixed reference states. There is, however, no contradiction, because the optimization problem studied in \citeMy{ma2024quantum} is different from the problem considered here.
Lemma~\ref{lem:ancilla-state-reduction} applies to our separately linear objective \(F_\mu(\Enc,\Dec)\), whereas \citeMy{ma2024quantum} optimizes a nonlinear cost function. Moreover, their protocol is not blind, because it uses input-dependent reference states, effectively employing state-dependent side information.
\end{remark}

Our first theorem identifies a universal resource-efficient architecture by showing that $k$ ancillas on the encoder side and $n$ ancillas on the decoder side always suffice.

\begin{restatable}[Universal sufficiency of $(k,n)$ ancillas]{theorem}{UniversalKNSufficient}\label{thm:sufficiency}
For every distribution $\mu$ of pure $n$-qubit states, there exists an $(n,k,k,n)$-QAE attaining
$
\max_{\Enc,\Dec} F_\mu(\Enc,\Dec),
$
where the maximum ranges over all CPTP encoders $\Enc:L(H_{2^n})\to L(H_{2^k})$ and decoders $\Dec:L(H_{2^k})\to L(H_{2^n})$.
\end{restatable}

\noindent
\emph{Proof sketch.}
By a Choi-matrix reformulation, the fidelity functional is separately linear in $\Enc$ and $\Dec$. Hence an optimum can be chosen with both channels extreme. Extreme encoders $L(H_{2^n})\to L(H_{2^k})$ have Choi rank at most $2^n$, and extreme decoders $L(H_{2^k})\to L(H_{2^n})$ have Choi rank at most $2^k$. Padded Kraus lists of these lengths can be embedded into isometries
\[
\widetilde U:H_{2^n}\to H_{2^k}\otimes H_{2^n},
\qquad
\widetilde V:H_{2^k}\to H_{2^n}\otimes H_{2^k},
\]
which extend to unitaries on $n+k$ qubits. The full proof is developed throughout Appendix~\ref{sec:choi-proof}, concluding in Appendix~\ref{sec:universal_kn_sufficient_proof}.

The universal sufficiency theorem is sharp on the encoder side. The source family below already forces $k$ encoder ancillas.

\myBlue
\begin{theorem}[Informal]\label{thm:encoder-lb-informal}
For every choice of input size $n$ and bottleneck size $k<n$, there is a source
distribution of pure $n$-qubit states for which fewer than $k$ encoder ancillas
are insufficient to reach the best possible reconstruction fidelity.
More precisely, for this source, no $(n,k,n_B,n_E)$-QAE with $n_B<k$ can attain
the optimum value of $F_\mu(\Enc,\Dec)$ over all CPTP encoders and decoders
through the $k$-qubit bottleneck. Hence every globally optimal QAE for this
source must satisfy
$
n_B \ge k.
$
In other words, the $k$ encoder ancillas used in the universal construction are necessary in the worst case.
\end{theorem}

\noindent
\emph{Proof sketch.}
The proof uses a one-parameter source family $\{\mu_{1,\eps}\}_{0<\eps<1}$, whose construction is given in Appendix~\ref{sec:mu1-lb}. Intuitively, the states are all close to one fixed reference state, but their small
differences point in many possible directions. A compression scheme that is optimal must therefore preserve the reference state perfectly while also
handling these hidden directions as well as the $k$-qubit bottleneck allows.
The optimal behavior has a simple operational form. The channel keeps the part
of the input lying in a subspace of the same size as the bottleneck and sends
the remaining part back to the reference state. In other words, the best channel
for this source is a ``reset'' channel.
The key point is that implementing this reset behavior through the bottleneck requires a sufficiently rich encoder. It can be shown that any architecture
with $n_B<k$ cannot be optimal for this source. A detailed restatement and full proof are given in Appendix~\ref{sec:mu1-lb}.
\color{black}

\begin{corollary}[Sharp universal encoder threshold]\label{cor:sharp-encoder}
For blind single-copy compression under the infidelity objective, $k$ encoder ancillas are universally sufficient and, in the worst case, necessary.
\end{corollary}

For the source \(\mu_{1,\eps}\) with sufficiently small $\eps$, as well as the sources discussed in Appendix~\ref{sec:haar-appendix}, the optimal fidelity is always achieved by an isometric decoder. Furthermore, as we will demonstrate in Section~\ref{sec:experiments}, isometric decoders also appear sufficient for the heuristic datasets constructed from real-world data (MNIST). In particular, there is no noticeable difference between the settings \((n_B,n_E)=(k,n-k)\) and \((n_B,n_E)=(k,n-k+1)\). Taken together, these observations might suggest that isometric decoders always suffice. A natural remaining question is whether the decoder ancilla count can \emph{always} be reduced to \(n-k\), equivalently, whether the decoder may always be restricted to an isometric channel. The answer is negative. In what follows, we establish this by constructing an explicit counterexample.

\begin{restatable}[Decoder isometry is not sufficient]{proposition}{DecoderCounterExample}\label{prop:counterexample}
Let $n=2$, $k=1$,
\begin{equation}
\label{eq:dec_counter_source}
\ket{\psi_\phi}=\frac{\ket{0}+e^{i\phi}\ket{1}}{\sqrt 2},
\qquad
\ket{\Psi_\phi}=\ket{\psi_\phi}^{\otimes 2},
\qquad
\phi\in[0,2\pi),
\end{equation}
and let $\mu_{\mathrm{ph}}$ be the uniform distribution on $\phi$. Then
\begin{equation}
\label{eq:dec_counter_optimum_cptp}
\sup_{\Enc,\,V^\dagger V=I_2} F_{\mu_{\mathrm{ph}}}(\Enc,\Dec_V)
<
\frac34
\le
\sup_{\Enc,\Dec} F_{\mu_{\mathrm{ph}}}(\Enc,\Dec),
\end{equation}
where $\Dec_V(\rho)=V\rho V^\dagger$.
\end{restatable}

\noindent
\emph{Proof sketch.}
A concrete encoder--decoder pair with a decoder of Kraus rank 2 attains fidelity $3/4$. For an isometric decoder, averaging over $\phi$ produces a rank-$2$ projector optimization problem whose value is at most $3/4$. Equality would force a CPTP map sending the input family to a family of pure projected states. A Gram-matrix criterion shows that no such pure-state transformation exists. The full argument is given in Appendix~\ref{sec:counterexample-appendix}.

Because the gap is visually negligible, we employ two distinct methods for greater precision. The first uses a training set of 2000 states and a test set of 1000 states. The second is a more computationally efficient method that replaces the continuous integral over the phase distribution with an exact evaluation at a finite set of points, justified by the following observation.

\myBlue
\begin{observation}[Informal]
\label{obs:finite_eval}
The average fidelity for the source $\mu_{\mathrm{ph}}$ can be computed from only five phases.
\end{observation}

A detailed restatement and full proof are given in Appendix~\ref{sec:proof_finite_eval}.
\color{black}

\begin{figure}
    \centering
    \begin{subfigure}[b]{0.15\linewidth}
        \centering
        \includegraphics[width=\linewidth]{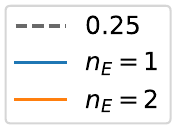}
    \end{subfigure}
    
    \begin{subfigure}[b]{0.48\linewidth}
        \centering
        \includegraphics[width=\linewidth]{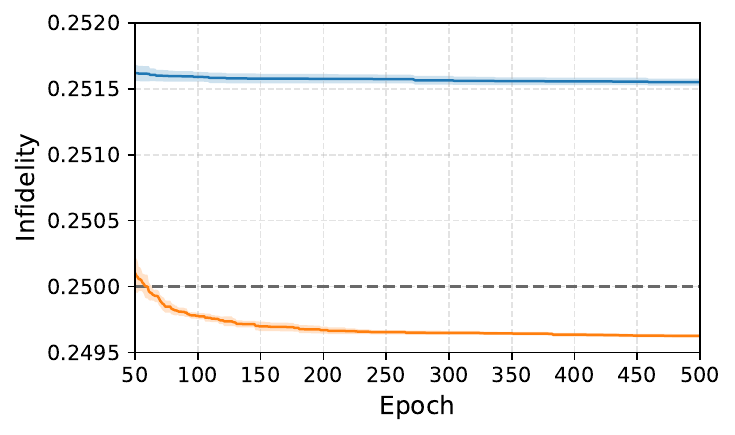}
        \caption{}
    \end{subfigure}
    \begin{subfigure}[b]{0.48\linewidth}
        \centering
        \includegraphics[width=\linewidth]{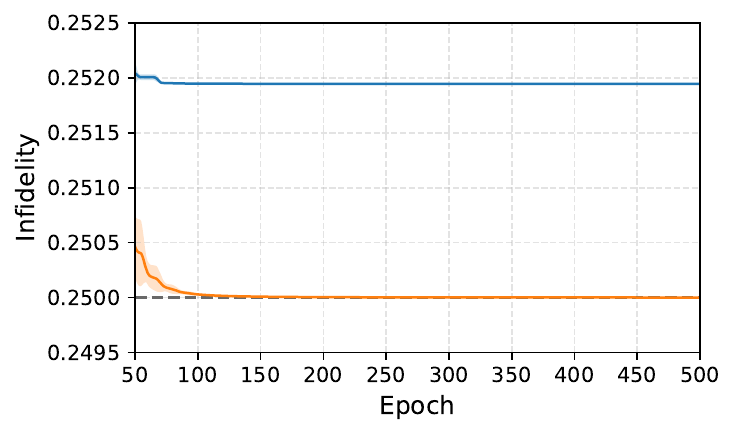}
        \caption{}
    \end{subfigure}
    \caption{Infidelity curves for an isometric decoder (\(n_E=1\)) and a non-isometric decoder (\(n_E=2\)) with $n_B=1$ encoder ancilla qubit, under the source distribution in Eq.~\eqref{eq:dec_counter_source} (with 2-sigma confidence intervals). The dashed line represents the theoretical upper bound of \(1/4\) for the minimum infidelity over all CPTP pairs \((\mathcal{E},\mathcal{D})\), implied by Eq.~\eqref{eq:dec_counter_optimum_cptp}. The displayed panels represent results obtained via (a) training with 2000 training states and 1000 test states, and (b) evaluation at discrete points as justified by Observation~\ref{obs:finite_eval}. While the converged losses are nearly indistinguishable on the original scale, the magnified view makes clear that the isometric decoder fails to reach the fidelity \(3/4\) achieved by a non-isometric decoder.}
    \label{fig:dec_counter}
\end{figure}

In all cases, we adopt the exponential parameterization for the trainable unitaries, $U = \exp(iH)$, where $H$ is a trainable Hermitian operator.
The distinction between isometric and non-isometric decoders is mathematically clear but quite subtle, as illustrated in Figure~\ref{fig:dec_counter}.

\begin{remark}
\label{rem:haar-remark}
The negative result of Proposition~\ref{prop:counterexample} is genuinely source-dependent. Appendix~\ref{sec:haar-appendix} shows that for the Haar prior on pure states, and more generally for certain mixtures with perfectly compressible priors, a rank-one decoder is already optimal.
\end{remark}

The universal sufficiency theorem is sharp on the decoder side when $k=1$.
Specifically, we show that the phase-family counterexample of Proposition~\ref{prop:counterexample} extends from
$(n,k) = (2,1)$ to every pair $(n,1)$ with $n \ge 2$.

For $n \ge 2$, let $R$ denote an $(n-2)$-qubit register, let
\[
|0_R\rangle := |0\rangle^{\otimes (n-2)},
\qquad
|\Omega_{\phi}\rangle := |\Psi_{\phi}\rangle \otimes |0_R\rangle,
\]
where $|\Psi_{\phi}\rangle$ is the two-qubit phase family from Eq.~\eqref{eq:dec_counter_source}, and
let $\mu_{\mathrm{ph}}^{(n)}$ be the distribution induced by the uniform choice of
$\phi \in [0,2\pi)$.

\begin{restatable}[Sharp decoder threshold for $k=1$]{theorem}{DecoderSharpKOne}
\label{thm:decoder-k1-threshold}
For every $n \ge 2$,
\[
\sup_{\Enc,\,V^{\dagger}V=I_2} F_{\mu_{\mathrm{ph}}^{(n)}}(\Enc,\Dec_V)
< \frac{3}{4}
\le
\sup_{\Enc,\Dec} F_{\mu_{\mathrm{ph}}^{(n)}}(\Enc,\Dec),
\]
where $\Enc : L(H_{2^n}) \to L(H_2)$, $\Dec : L(H_2) \to L(H_{2^n})$, and
$\Dec_V(\rho) = V\rho V^{\dagger}$. Consequently, every optimal
$(n,1,n_B,n_E)$-QAE for the source $\mu_{\mathrm{ph}}^{(n)}$ satisfies
$
n_E \ge n.
$
\end{restatable}

\begin{proof}
See Appendix~\ref{subsec:k1-extension}.
\end{proof}

\begin{corollary}[Sharp universal decoder threshold for $k=1$]
\label{cor:decoder-k1-threshold}
For blind single-copy compression through a one-qubit bottleneck under the
infidelity objective, $n$ decoder ancillas are universally sufficient and, in
the worst case, necessary.
\end{corollary}

\begin{table}[t]
\centering
\caption{Ancilla requirements and status of several architectures for blind single-copy compression under the infidelity objective. The row labeled ``reset channel'' refers to the explicit optimal construction for the family $\mu_{1,\eps}$ from Theorem~\ref{thm:encoder-lb-informal}. \tcb{Decoder ancillas $n-k$ correspond to an isometric decoder, because the decoder has just enough qubits to expand the $k$-qubit latent register to $n$ output qubits and has no leftover environment to discard.}}
\label{tab:resources}
\begin{tabular}{@{}lcccc@{}}
\toprule
Scheme & \makecell[c]{Encoder\\ancillas} & \makecell[c]{Decoder\\ancillas} & \makecell[c]{Unitary\\widths} & \makecell[c]{Status\\(infidelity)} \\
\midrule
Conventional QAE & $0$ & $n-k$ & $n$, $n$ & not universal \\
$\zeta$-QVAE \citeMy{wang2025quantum} & $n+2k$ & $2n+k$ & $2(n+k)$, $2(n+k)$ & universal \\
Theorem~\ref{thm:sufficiency} & $k$ & $n$ & $n+k$, $n+k$ & universal \\
Reset channel for $\mu_{1,\eps}$ & $k$ & $n-k$ & $n+k$, $n$ & optimal for small $\eps$ \\
Decoder isometry for $\mu_{\mathrm{ph}}$ & -- & $n-k$ & --, $n$ & suboptimal \\
\bottomrule
\end{tabular}
\end{table}

Table~\ref{tab:resources} summarizes the ancilla requirements and circuit dimensions of the compression architectures considered in this work.

The next theorem shows that, although isometric decoders need not be optimal in general, they remain near-optimal whenever the average source state is well concentrated on its top $2^k$ eigenspaces.

\begin{theorem}[Informal]
\label{thm:informal_multiplicative}
For blind single-copy compression, restricting the decoder to be an isometry is near-optimal whenever the source is close to an $m$-dimensional subspace, where $m=2^k$ is the latent dimension. More precisely, if $\eta_m$ denotes the weight of the average source state outside its top $m$ eigenspaces, then the best isometric decoder achieves at least a fraction
$
1-\bigl(1-\tfrac{1}{m}\bigr)\eta_m
$
of the unrestricted optimum average fidelity.
\end{theorem}

A detailed restatement and full proof are given in Appendix~\ref{subsec:isometric-multiplicative}.
% \color{black}

\section{Experiments}
\label{sec:experiments}

\subsection{Setup}
\label{sec:data-prep}

We perform experiments on \(\mu_{1,0.1}\) to examine whether \(k\) encoder ancillas are necessary.
The choice of $\mu_{1,0.1}$ is sufficient to illustrate Corollary~\ref{cor:sharp-encoder}.
To this end, we compare (i) \((n,k,0,n-k)\)-QAE, (ii) \((n,k,0,n)\)-QAE, and (iii) \((n,k,k,n)\)-QAE. If we observe a clear performance gain both from (i) to (ii) and from (ii) to (iii), this would suggest that when the number of encoder ancillas is smaller than \(k\), even enlarging the decoder ancilla count to \(n\) (which is already sufficient for the decoder itself under the infidelity loss) remains insufficient to attain the optimum.

We also introduce $\mu_2$, an empirical distribution induced by the MNIST dataset (version 3.0.1, CC BY-SA 3.0 license).
More precisely, $\mu_2$ is the distribution over quantum states obtained by encoding MNIST images into normalized state vectors in $\mathbb{C}^d$ via the data-preparation procedure described in Appendix~\ref{sec:appendix_mnist_data_prep}.
\tcb{This construction intentionally creates quantum states whose average source state has most of its mass in a low-dimensional subspace, matching the regime where Theorem~\ref{thm:informal_multiplicative} predicts that isometric decoders should be near-optimal.}

For the experiments reported in Sections~\ref{sec:results_mu_1} and \ref{sec:results_mu_2}, we generated 2000 training samples and 1000 test samples. The encoder and decoder unitaries were trained for 500 epochs using the Adam optimizer with a learning rate of $10^{-3}$ and a batch size of 64. All computations were performed on a 32-core Intel Xeon Processor (Skylake, 2.5 GHz) and an NVIDIA H100 NVL GPU running CUDA 12.9, with 62 GB of memory. Training time information is provided in Appendix~\ref{sec:training-time}. The experiments with uncertainty bands use ten independent training runs with random seeds $0,\ldots,9$.

% mu1 Results
\begin{figure}
    \centering
    \begin{subfigure}[b]{0.27\linewidth}
        \centering
        \includegraphics[width=\linewidth]{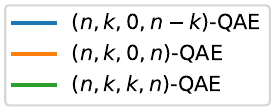}
    \end{subfigure}
    \\\vspace{1em}
    \begin{subfigure}[b]{0.32\linewidth}
        \centering
        \includegraphics[page=1, width=\linewidth]{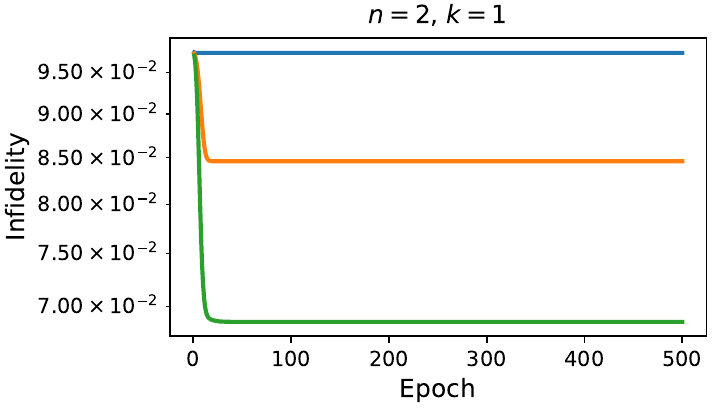}
    \end{subfigure}
    \hfill
    \begin{subfigure}[b]{0.32\linewidth}
        \centering
        \includegraphics[page=2, width=\linewidth]{figures/mu1/mu1_merged_tight.pdf}
    \end{subfigure}
    \hfill
    \begin{subfigure}[b]{0.32\linewidth}
        \centering
        \includegraphics[page=3, width=\linewidth]{figures/mu1/mu1_merged_tight.pdf}
    \end{subfigure}
    \\\vspace{1em}
    \begin{subfigure}[b]{0.32\linewidth}
        \centering
        \includegraphics[page=4, width=\linewidth]{figures/mu1/mu1_merged_tight.pdf}
    \end{subfigure}
    \hfill
    \begin{subfigure}[b]{0.32\linewidth}
        \centering
        \includegraphics[page=5, width=\linewidth]{figures/mu1/mu1_merged_tight.pdf}
    \end{subfigure}
    \hfill
    \begin{subfigure}[b]{0.32\linewidth}
        \centering
        \includegraphics[page=6, width=\linewidth]{figures/mu1/mu1_merged_tight.pdf}
    \end{subfigure}
    \caption{Test infidelity curves for source $\mu_{1,0.1}$, where $n \in \{2, 3, 4\}$ and $k \in \{1, \dots, n-1\}$ (with 2-sigma confidence intervals). Across all settings, the performance differences among $(n,k,0,n-k)$-QAE, $(n,k,0,n)$-QAE, and $(n,k,k,n)$-QAE are clearly visible. The converged values for $(2,1,0,1)$-QAE and $(2,1,1,2)$-QAE are consistent with the first-order bounds derived in Appendix~\ref{subsec:conventional-special-case}.}
    \label{fig:mu1_results}
\end{figure}

% MNIST Results
\begin{figure}
    \centering
    \begin{subfigure}[b]{0.3\linewidth}
        \centering
        \includegraphics[width=\linewidth]{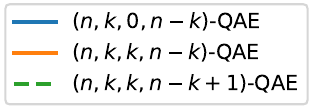}
    \end{subfigure}
    \\\vspace{1em}
    \begin{subfigure}[b]{0.32\linewidth}
        \centering
        \includegraphics[page=1, width=\linewidth]{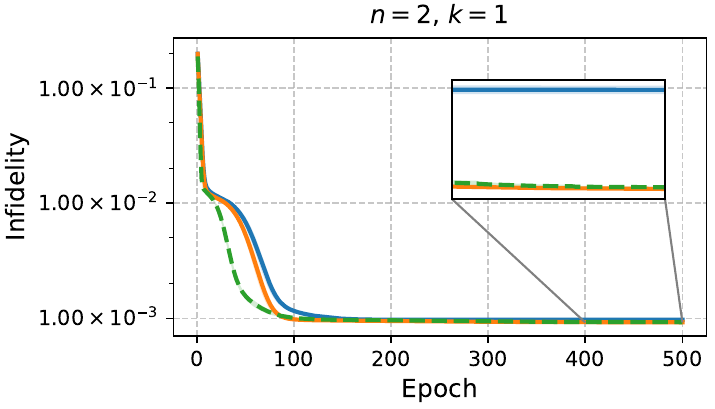}
    \end{subfigure}
    \hfill
    \begin{subfigure}[b]{0.32\linewidth}
        \centering
        \includegraphics[page=2, width=\linewidth]{figures/MNIST/mu2_merged_tight.pdf}
    \end{subfigure}
    \hfill
    \begin{subfigure}[b]{0.32\linewidth}
        \centering
        \includegraphics[page=3, width=\linewidth]{figures/MNIST/mu2_merged_tight.pdf}
    \end{subfigure}
    \\\vspace{1em}
    \begin{subfigure}[b]{0.32\linewidth}
        \centering
        \includegraphics[page=4, width=\linewidth]{figures/MNIST/mu2_merged_tight.pdf}
    \end{subfigure}
    \hfill
    \begin{subfigure}[b]{0.32\linewidth}
        \centering
        \includegraphics[page=5, width=\linewidth]{figures/MNIST/mu2_merged_tight.pdf}
    \end{subfigure}
    \hfill
    \begin{subfigure}[b]{0.32\linewidth}
        \centering
        \includegraphics[page=6, width=\linewidth]{figures/MNIST/mu2_merged_tight.pdf}
    \end{subfigure}
    \caption{Test infidelity curves for MNIST source $\mu_2$, where $n \in \{2, 3, 4\}$ and $k \in \{1, \dots, n-1\}$ (with 2-sigma confidence intervals). Across all settings, the performance difference between $(n,k,k,n-k)$-QAE and $(n,k,k,n-k+1)$-QAE is negligible relative to the gap between $(n,k,k,n-k)$-QAE and $(n,k,0,n-k)$-QAE.}
    \label{fig:mnist_results}
\end{figure}

\subsection{Numerical results for $\mu_{1,0.1}$}
\label{sec:results_mu_1}

We compare three architectures using the $\mu_{1,0.1}$ dataset: (i) the $(n,k,0,n-k)$-QAE (conventional), (ii) the $(n,k,0,n)$-QAE, and (iii) the $(n,k,k,n)$-QAE.
The results are shown in Figure~\ref{fig:mu1_results}.
We observe clear performance gaps among (i), (ii), and (iii).
This indicates that increasing the decoder ancilla count to $n$ is insufficient to attain the optimum over all CPTP encoder--decoder pairs when the number of encoder ancillas is smaller than $k$.

\subsection{Numerical results for $\mu_2$}
\label{sec:results_mu_2}

We compare three architectures using the $\mu_2$ dataset: (i) the $(n,k,0,n-k)$-QAE (conventional), (ii) the $(n,k,k,n-k)$-QAE, and (iii) the $(n,k,k,n-k+1)$-QAE.
The results are shown in Figure~\ref{fig:mnist_results}.
We observe a clear performance gap between (i) and (ii), whereas there is no noticeable difference between (ii) and (iii).
This indicates that increasing the decoder ancilla from $n-k$ to $n-k+1$ provides no practical advantage, at least within the scope of our experiments.
While this result may appear to contradict Proposition~\ref{prop:counterexample}, we note again that the performance gap between $n_E=1$ and $n_E=2$ described in the counterexample is analytically valid but numerically negligible.

\subsection{Additional distributions}
\label{sec:additional-experiments}

\paragraph{Increased-tail MNIST.}
To test sensitivity to the concentration of the original MNIST source, we consider a larger-tail preparation setting. Its definition and numerical infidelities are given in Appendix~\ref{app:larger-tail} and Table~\ref{tab:larger-tail}.

\paragraph{Coherently perturbed GHZ states.}
We consider pure states obtained by coherently perturbing the Greenberger--Horne--Zeilinger (GHZ) state \citeMy{greenberger1989going}. The source definition and infidelities appear in Appendix~\ref{app:noisy-ghz} and Table~\ref{tab:noisy-ghz}.

\paragraph{Output depolarization.}
To examine a departure from ideal reconstruction, we apply independent single-qubit depolarization to every reconstructed output qubit. Appendix~\ref{app:depolarization} specifies the channel and Table~\ref{tab:depolarization} reports infidelities.

\paragraph{Trace distance.}
To examine whether the empirical comparison depends on the distortion measure, we also report results based on trace distance in Appendix~\ref{app:trace-experiment}, Table~\ref{tab:trace-distance}.
\color{black}

\section{Discussion}

\paragraph{Natural settings with $k\ll n$.}
A sufficient structural condition is that the source (almost) lies in a common subspace of dimension $r\ll2^n$, which permits reliable compression to $\lceil\log_2 r\rceil$ qubits. Examples include:
\begin{itemize}%[label=(\roman*),leftmargin=*,nosep]
\item Identically prepared product states $|\phi\rangle^{\otimes n}$, which lie in an $(n+1)$-dimensional permutation-symmetric subspace \citeMy{plesch2010efficient, rozema2014quantum}.
\item A fixed-$r_{\mathrm{ex}}$ excitation sector of dimension $\binom{n}{r_{\mathrm{ex}}}$, requiring $k=\lceil\log_2\binom{n}{r_{\mathrm{ex}}}\rceil=O(r_{\mathrm{ex}}\log n)$ for fixed excitation number.
\item A code space for $n$ physical qubits that carries $q$ logical qubits, with $k=q\ll n$ when the code has large physical redundancy.
\item A low-energy source supported on $R$ relevant eigenstates, with $k=\lceil\log_2R\rceil$.
\item Nonorthogonal state families used in QAE-assisted group discrimination \citeMy{huang2020realization}, when their joint span is small.
\end{itemize}

For a fixed source, $k$ can be selected by a validation sweep with multiple initializations, choosing the smallest bottleneck meeting a target fidelity.

\textbf{Limitations.}
This work restricts attention to blind single-copy compression of pure states, with performance measured solely by reconstruction infidelity. Accordingly, our results should not be interpreted as statements about other objectives, such as latent-space regularization, representation learning quality, robustness, or trainability. Such objectives may be valuable in practice, but they lie beyond the scope of the present theory and may modify the optimal resource counts identified here. The conclusions should also be distinguished from results in the literature that optimize different loss functions or assume different informational access. Furthermore, our empirical validation uses an artificially engineered quantum dataset, which is useful for probing the low-dimensional concentration regime but does not adequately represent the entanglement structures of physical quantum data, such as many-body ground states.

\section{Conclusion}
Our main results are resource-theoretic statements about blind single-copy compression under the infidelity objective. We identify an architecture that is wide enough to attain the optimum over all CPTP encoder--decoder pairs, yet narrower than fully general models. In particular, $k$ encoder ancillas and $n$ decoder ancillas are universally sufficient, and the encoder threshold is sharp. The decoder side remains less completely characterized. Specifically, the tight universal decoder ancilla lower bound as a function of $n$ and $k$ remains open. Because infidelity is the primary objective in many compression tasks, the ancilla counts identified here provide a natural baseline. A promising direction is to study how additional objectives, including explicit regularization and other desiderata for representation learning, alter the architecture or lead to genuinely different optimal tradeoffs.

% \clearpage
\bibliographystyle{plainnat}
\bibliography{references}

% ---------------------------
\clearpage
\appendix

\section{List of symbols and logical framework}

\begingroup
\small
\renewcommand{\arraystretch}{1.12}
\setlength{\LTleft}{0pt}
\setlength{\LTright}{0pt}

\begin{longtable}{@{}>{\raggedright\arraybackslash}p{0.24\textwidth}%
                    >{\raggedright\arraybackslash}p{0.72\textwidth}@{}}
% \caption{}
\label{tab:symbols}\\
\toprule
\textbf{Symbol} & \textbf{Meaning} \\
\midrule
\endfirsthead

\multicolumn{2}{@{}l}{\textit{Table~\thetable\ continued.}}\\
\toprule
\textbf{Symbol} & \textbf{Meaning} \\
\midrule
\endhead

\bottomrule
\endfoot

\addlinespace[0.3em]
\multicolumn{2}{@{}l}{\textbf{Hilbert-space and channel notation}}\\
$H_d \cong \mathbb{C}^d$ & A $d$-dimensional Hilbert space. \\
$L(H)$ & Linear operators on the Hilbert space $H$. \\
$H_A,H_B,\ldots$ & Hilbert spaces associated with registers $A,B,\ldots$ \\
$H_{AB}:=H_A\otimes H_B$ & Composite Hilbert space of registers $A$ and $B$. \\
$AB \equiv CD$ & The same physical system viewed under two different factorizations. \\
$\{|x\rangle : x\in\{0,1\}^n\}$ & Computational basis of an $n$-qubit register. \\
$\{|j\rangle : 0\le j\le 2^n-1\}$ & Computational basis of an $n$-qubit register. \\
$|v\rangle^\perp$ & Subspace orthogonal to the vector $|v\rangle$. \\
$\operatorname{Tr}_B$ & Partial trace over subsystem $B$. \\
$\Phi : L(H_A)\to L(H_B)$ & Generic quantum channel. \\
$\operatorname{Ad}_U(\rho)=U\rho U^\dagger$ & Unitary channel induced by the unitary $U$. \\
$\{K_i\}$ & Kraus operators of a channel. \\
$\operatorname{rank}_{\mathrm{Kraus}}(\Phi)$ & Kraus rank (equivalently, Choi rank) of $\Phi$. \\
$\Lambda(\Phi)$ & Choi matrix of the linear map $\Phi$. \\
$\Lambda_1 \star \Lambda_2$ & Link product of two Choi matrices. \\
$\rho^T$, $T_B$ & Transpose of $\rho$ and partial transpose over subsystem $B$. \\

\addlinespace[0.4em]
\multicolumn{2}{@{}l}{\textbf{Compression setting and QAE architecture}}\\
$d:=2^n$ & Input/output Hilbert-space dimension. \\
$m:=2^k$ & Latent-space dimension (bottleneck dimension). \\
$F_\mu(\Enc,\Dec)$ & Average reconstruction fidelity for source distribution $\mu$. \\
$L_\mu(\Enc,\Dec):=1-F_\mu(\Enc,\Dec)$ & Average infidelity. \\
$(n,k,n_B,n_E)$-QAE & QAE architecture with $n_B$ encoder ancillas and $n_E$ decoder ancillas. \\
$\rho_{\mathrm{enc}}$ & Encoded state on the latent register $C$. \\
$\rho_{\mathrm{rec}}$ & Reconstructed state on the output register $F$. \\

\addlinespace[0.4em]
\multicolumn{2}{@{}l}{\textbf{Source families}}\\
$\mu_{\mathrm{Haar}}$ & Haar distribution on pure states in $\mathbb{C}^d$. \\
$H^\perp := |0\rangle^\perp \subset \mathbb{C}^d$ & Orthogonal complement of the reference basis vector $|0\rangle$. \\
$|\psi_\varepsilon(\varphi)\rangle$ & Perturbed source state $|\psi_\varepsilon(\varphi)\rangle
=
\sqrt{1-\varepsilon}\,|0\rangle+\sqrt{\varepsilon}\,|\varphi\rangle$, $|\varphi\rangle\in H^\perp$, $\|\varphi\|=1$. \\
$\mu_{1,\varepsilon}$ & Distribution induced by Haar-uniform $|\varphi\rangle$ on the states $|\psi_\varepsilon(\varphi)\rangle$. \\
$F_\varepsilon(E,D)$ & Average fidelity for the source family $\mu_{1,\varepsilon}$. \\
$\Phi_P(\rho)$ & Reset channel $\Phi_P(\rho)=P\rho P+\operatorname{Tr}\!\bigl[(I_d-P)\rho\bigr]\,|0\rangle\langle 0|$. \\

\addlinespace[0.4em]
\multicolumn{2}{@{}l}{\textbf{Decoder isometry counterexample}}\\
$|\psi_\phi\rangle$ & Single-qubit phase state $|\psi_\phi\rangle=(|0\rangle+e^{i\phi}|1\rangle) / \sqrt{2}$. \\
$|\Psi_\phi\rangle$ & Two-qubit source state $|\Psi_\phi\rangle=|\psi_\phi\rangle^{\otimes 2}$. \\
$\mu_{\mathrm{ph}}$ & Uniform distribution of $\phi$ on $[0,2\pi)$. \\
$\Dec_V(\rho)=V\rho V^\dagger$ & Isometric decoder generated by an isometry $V$. \\
$|\Sigma\rangle$ & Symmetric Bell state $|\Sigma\rangle=(|01\rangle+|10\rangle) / \sqrt{2}$. \\
$P_{t,\gamma}$ & Rank-$2$ projector $P_{t,\gamma}=|\Sigma\rangle\langle\Sigma|+|\eta_{t,\gamma}\rangle\langle\eta_{t,\gamma}|$. \\
$|\eta_{t,\gamma}\rangle$ & Vector $|\eta_{t,\gamma}\rangle=\cos t\,|00\rangle+e^{i\gamma}\sin t\,|11\rangle$. \\
$|\chi^{(t,\gamma)}_\phi\rangle$ & Normalized projection of $|\Psi_\phi\rangle$ onto $\operatorname{im}(P_{t,\gamma})$. \\

\addlinespace[0.4em]
\multicolumn{2}{@{}l}{\textbf{Source-dependent multiplicative guarantee}}\\
$\lambda_i(X)$
& Eigenvalues of a Hermitian operator $X$, listed in nonincreasing order.\\
$s_m := \sum_{i=1}^m \lambda_i(\bar{\rho})$
& Weight of the average source state on its top $m$ eigenvalues.\\
$\eta_m := 1-s_m$
& Tail weight of the average source state outside its top $m$ eigenvalues.\\
$p_P(\psi) := \langle\psi|P|\psi\rangle$
& Probability mass of the source state $|\psi\rangle$ inside $\operatorname{im}(P)$.\\
$\Enc_{P,\xi}$
& Encoder in the isometric lower bound construction.\\
$\Dec_{V_P}$
& Associated isometric decoder in the lower bound construction.\\

\end{longtable}
\endgroup
\setcounter{table}{1}

\clearpage
\begin{figure}
    \centering
    \includegraphics[width=\linewidth]{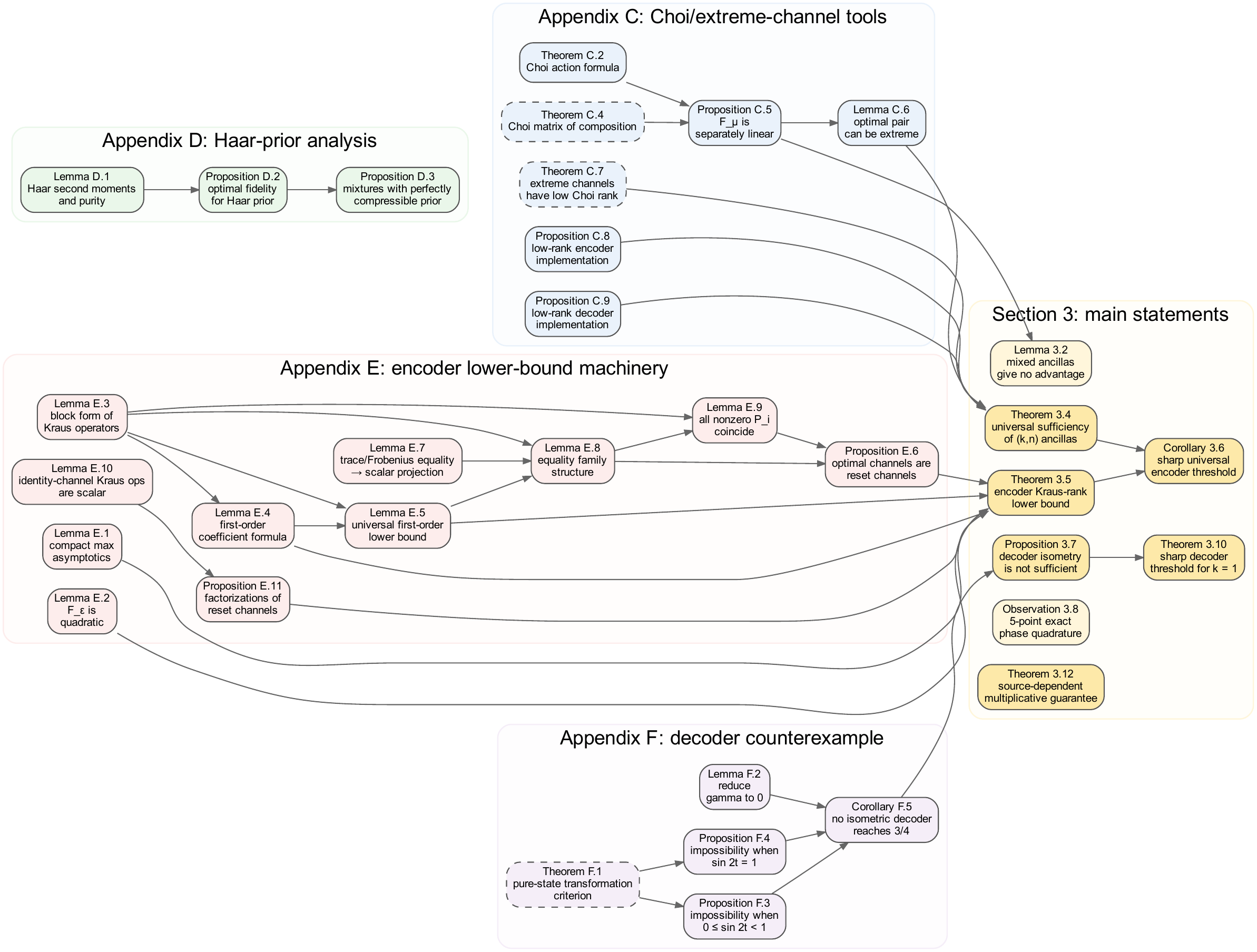}
    \caption{Logical dependency graph of the results. Dashed nodes are used as cited black boxes.}
    \label{fig:dependencies}
\end{figure}

\section{Basic concepts}
\label{sec:appendix_prelim}

\subsection{Hilbert spaces, qubits, registers, and Dirac notation}

A finite-dimensional Hilbert space is a finite-dimensional complex vector space equipped with an inner product. Throughout, $H_d \cong \mathbb{C}^d$ denotes a $d$-dimensional Hilbert space and $L(H)$ denotes the linear operators on $H$. In Dirac notation, a vector is written as a ket $|\psi\rangle$, its conjugate transpose is the bra $\langle \psi| := |\psi\rangle^{\dagger}$, and the inner product of $|\phi\rangle$ and $|\psi\rangle$ is $\langle \phi|\psi\rangle$. A vector represents a physical pure state only when it has unit norm, $\langle \psi|\psi\rangle = 1$.
For a vector $|v\rangle$, we write $|v\rangle^{\perp}$ for the subspace orthogonal to $|v\rangle$.

A qubit is a system with Hilbert space $H_2 \cong \mathbb{C}^2$ and computational basis $\{ |0\rangle, |1\rangle \}$. An $n$-qubit register has Hilbert space
\[
H_{2^n} \cong (\mathbb{C}^2)^{\otimes n},
\]
so an $n$-qubit system has dimension $2^n$. For quantum registers $A,B,\ldots$, we write $H_A,H_B,\ldots$ for the corresponding Hilbert spaces. Composite systems are modeled by tensor products, so $H_{AB} := H_A \otimes H_B$. The same physical space may be reinterpreted under a different factorization, written for example as $AB \equiv CD$.

Tensor products are used not only for Hilbert spaces but also for states and operators. Thus $|\psi\rangle_A \otimes |\phi\rangle_B$ is a joint pure state on $AB$, often abbreviated as $|\psi\rangle_A |\phi\rangle_B$, and $X_A \otimes Y_B$ is an operator on $H_{AB}$. The computational basis of an $n$-qubit system is $\{ |x_1\cdots x_n\rangle : x_i \in \{0,1\} \}$, equivalently $\{ |j\rangle : 0 \le j \le 2^n-1 \}$ after identifying a binary string with its integer value.

\subsection{Pure states, mixed states, density operators, and reduced states}

A pure state may be described by a unit vector $|\psi\rangle$, but the object acted on by channels and partial traces is the corresponding density operator $|\psi\rangle\langle \psi| \in L(H)$. These are different mathematical objects: $|\psi\rangle$ is a vector in $H$, whereas $|\psi\rangle\langle \psi|$ is an operator on $H$. In particular, $|\psi\rangle$ and $e^{i\theta}|\psi\rangle$ are different vectors but determine the same density operator.

A mixed state on $H$ is a density operator $\rho \in L(H)$, meaning a positive semidefinite operator with unit trace, $\rho \ge 0$ and $\Tr(\rho)=1$. Equivalently, a mixed state can be written as a convex combination of pure states,
\[
\rho = \sum_i p_i |\psi_i\rangle\langle \psi_i|,
\qquad
p_i \ge 0,
\qquad
\sum_i p_i = 1.
\]
The same density operator may arise from many different ensembles $\{p_i,|\psi_i\rangle\}$.

For a bipartite state $\rho_{AB}$, the reduced state (or marginal) on subsystem $A$ is defined by the partial trace over $B$:
\[
\rho_A := \Tr_B(\rho_{AB})
= \sum_j (I_A \otimes \langle j|_B)\, \rho_{AB}\, (I_A \otimes |j\rangle_B),
\]
where $\{ |j\rangle_B \}$ is any orthonormal basis of $H_B$. The result is basis independent. Likewise $\rho_B := \Tr_A(\rho_{AB})$. On simple tensors one has
\[
\Tr_B(X_A \otimes Y_B) = X_A \, \Tr(Y_B),
\]
and linearity extends this to all operators. Operationally, taking the partial trace is the same as discarding the corresponding subsystem. Thus $\rho_A$ is the state accessible to an observer who keeps only subsystem $A$.

\subsection{Quantum channels, unitary evolution, CPTP maps, and channel representations}

A superoperator is a linear map $\Phi : L(H_A) \to L(H_B)$. The physically allowed maps are quantum channels, namely CPTP maps. Complete positivity means that for every auxiliary Hilbert space $H_R$ and every positive semidefinite operator $X \in L(H_A \otimes H_R)$,
\[
(\Phi \otimes \mathrm{id}_R)(X) \ge 0.
\]
This is the correct positivity condition because a system may be entangled with external reference systems. Trace preservation means
\[
\Tr[\Phi(X)] = \Tr[X]
\]
for every operator $X$.

Closed-system evolution is unitary. The unitary group on $H_d$ is
\[
\mathrm{U}(d) := \{ U \in L(H_d) : U^{\dagger}U = UU^{\dagger} = I_d \}.
\]
If $U \in \mathrm{U}(d)$, the corresponding unitary channel is conjugation by $U$,
\[
\Ad_U(\rho) := U\rho U^{\dagger}.
\]
General open-system evolution is modeled by CPTP maps rather than only by unitary channels.

Every CPTP map $\Phi : L(H_A) \to L(H_B)$ admits a Kraus representation
\[
\Phi(\rho) = \sum_{i=1}^r K_i \rho K_i^{\dagger},
\qquad
K_i : H_A \to H_B,
\qquad
\sum_{i=1}^r K_i^{\dagger} K_i = I_A.
\]
The minimal number $r$ of nonzero Kraus operators in such a representation is the Kraus rank of $\Phi$.

Equivalently, Stinespring's theorem states that there exist an auxiliary Hilbert space $H_E$ and an isometry $V : H_A \to H_B \otimes H_E$ such that
\[
\Phi(\rho) = \Tr_E\!\bigl(V\rho V^{\dagger}\bigr).
\]
After adjoining an ancilla in a fixed state and extending the isometry to a unitary on a larger space, this becomes the familiar unitary-plus-ancilla-plus-partial-trace realization
\[
\Phi(\rho) = \Tr_E\!\left[ U\bigl(\rho \otimes |0\rangle\langle 0|_E\bigr) U^{\dagger} \right].
\]
Concretely, if we write $\Phi$ in Kraus form as
\[
\Phi(\rho)=\sum_{i=0}^{r-1} K_i \rho K_i^\dagger,
\]
choose an orthonormal family $\{\ket{i}_E\}_{i=0}^{r-1}$ in the auxiliary space $H_E$, and define
\[
V=\sum_{i=0}^{r-1} K_i \otimes \ket{i}_E,
\]
then $V:H_A\to H_B\otimes H_E$ is an isometry (because $\sum_i K_i^\dagger K_i=I_{H_A}$), and
\begin{align*}
V\rho V^\dagger
&=
\left(\sum_{i=0}^{r-1} K_i\otimes \ket{i}_E\right)\rho
\left(\sum_{j=0}^{r-1} K_j^\dagger\otimes \bra{j}_E\right) \\
&=
\sum_{i,j=0}^{r-1} K_i \rho K_j^\dagger \otimes \ket{i}\!\bra{j}_E.
\end{align*}
Taking the partial trace over $H_E$ gives
\begin{align*}
\Tr_E\!\bigl(V\rho V^\dagger\bigr)
&=
\sum_{k=0}^{r-1}
(I_{H_B}\otimes \bra{k}_E)
\Bigl(\sum_{i,j=0}^{r-1} K_i \rho K_j^\dagger \otimes \ket{i}\!\bra{j}_E\Bigr)
(I_{H_B}\otimes \ket{k}_E) \\
&=
\sum_{i,j,k=0}^{r-1}
K_i \rho K_j^\dagger \,
\langle k | i\rangle \langle j | k\rangle \\
&=
\sum_{i=0}^{r-1} K_i \rho K_i^\dagger \\
&= \Phi(\rho).
\end{align*}

\subsection{Blind single-copy compression and Haar-random pure states}

In blind compression, the encoder is not told which source state was supplied. It must apply the same encoding channel to every input state. In single-copy compression, only one copy of the source state is available. Thus blind single-copy compression from $n$ to $k<n$ qubits asks for channels
\[
\mathcal{E} : L(H_{2^n}) \to L(H_{2^k}),
\qquad
\mathcal{D} : L(H_{2^k}) \to L(H_{2^n}),
\]
that reconstruct an unknown pure $n$-qubit state as well as possible after passing through a $k$-qubit bottleneck.

Haar measure on $\mathrm{U}(d)$ is the unique probability measure invariant under left and right multiplication by fixed unitaries. The induced measure on pure states is obtained by applying a Haar-random unitary to a fixed reference vector. For example, if $U \sim \mathrm{Haar}$ on $\mathrm{U}(d)$, then $U|0\rangle$ is Haar-random on the unit sphere of $H_d$. More generally, if $S \subseteq H$ is a subspace, a Haar-random pure state in $S$ is one drawn from the unique unitarily invariant probability measure on the unit sphere of $S$.

\section{Choi calculus, extreme channels, and the proof of Theorem~\ref{thm:sufficiency}}\label{sec:choi-proof}

\subsection{Choi representation and link product}\label{subsec:choi}

\begin{definition}[Choi representation]\label{def:choi}
Let $\Phi:L(H_{d_1})\to L(H_{d_2})$ be linear. Its \emph{Choi matrix} is
\begin{equation}\label{eq:choi}
\Lambda(\Phi)
:=
\sum_{i,j=0}^{d_1-1} \ketbra{i}{j}\otimes \Phi(\ketbra{i}{j})
\in L(H_{d_1}\otimes H_{d_2}).
\end{equation}
The \emph{Choi rank} (or \emph{Kraus rank}) of a linear map $\Phi$ is the rank of its corresponding Choi matrix.
\end{definition}

\begin{theorem}\label{thm:choi-action}
Let $\Phi:L(H_A)\to L(H_B)$ be linear. Then for every $\rho\in L(H_A)$,
\begin{equation}\label{eq:choi-action}
\Phi(\rho)
=
\Tr_A\bigl(\Lambda(\Phi)(\rho^T\otimes I_B)\bigr),
\end{equation}
where the transpose is taken in the computational basis.
\end{theorem}

\begin{definition}[Link product]\label{def:link}
For $\Lambda_1\in L(H_B\otimes H_C)$ and $\Lambda_2\in L(H_A\otimes H_B)$, define
\begin{equation}\label{eq:link}
\Lambda_1\star \Lambda_2
:=
\Tr_B\Bigl((I_A\otimes \Lambda_1)(\Lambda_2^{T_B}\otimes I_C)\Bigr),
\end{equation}
where $T_B$ denotes partial transpose over $H_B$.
\end{definition}

\begin{theorem}[\citeMy{chiribella2009theoretical}, Theorem~1]\label{thm:link-compose}
If $\Phi_1:L(H_A)\to L(H_B)$ and $\Phi_2:L(H_B)\to L(H_C)$ are linear maps, then
\begin{equation}\label{eq:link-compose}
\Lambda(\Phi_2\circ \Phi_1)=\Lambda(\Phi_2)\star \Lambda(\Phi_1).
\end{equation}
\end{theorem}

\begin{proposition}\label{prop:separate-linear}
For any source distribution $\mu$, the functional $F_\mu(\Enc,\Dec)$ from Eq.~\eqref{eq:objective} extends to a separately linear continuous functional on the vector spaces of linear maps $\Enc$ and $\Dec$.
\end{proposition}

\begin{proof}
For a pure state $\ket{\psi}$,
\begin{align}
(\Dec\circ \Enc)(\proj{\psi})
&=
\Tr_A\Bigl((\Lambda(\Dec)\star \Lambda(\Enc))(\proj{\psi}^{T}\otimes I)\Bigr) \notag\\
&=
\Tr_A\Bigl(\Tr_C\bigl[(I_A\otimes \Lambda(\Dec))(\Lambda(\Enc)^{T_C}\otimes I)\bigr](\proj{\psi}^{T}\otimes I)\Bigr).
\label{eq:expanded-composition}
\end{align}
The right-hand side is bilinear in $\Lambda(\Enc)$ and $\Lambda(\Dec)$, hence separately linear in $\Enc$ and $\Dec$. Averaging over $\mu$ preserves separate linearity and continuity.
\end{proof}

\subsection{Proof of Lemma~\ref{lem:ancilla-state-reduction}}
\label{app:proof_pure_ancillas_suffice}

\myBlue
We recall and prove Lemma~\ref{lem:ancilla-state-reduction}.
\AncillaAlwaysPure*
\color{black}

\begin{proof}
Let
\[
M_{\mathrm{mix}}
:=
\max_{U,V,\sigma_B,\tau_E}
F_\mu\!\left(\Enc_{U,\sigma_B}, \Dec_{V,\tau_E}\right),
\qquad
M_{0}
:=
\max_{U,V}
F_\mu\!\left(
\Enc_{U,|0^{n_B}\rangle\langle 0^{n_B}|},
\Dec_{V,|0^{n_E}\rangle\langle 0^{n_E}|}
\right).
\]
Both maxima exist because the relevant unitary groups and state spaces are compact,
and the objective is continuous in all parameters.

Since $|0^{n_B}\rangle\langle 0^{n_B}|$ and $|0^{n_E}\rangle\langle 0^{n_E}|$ are
particular density operators, we immediately have
\[
M_0 \le M_{\mathrm{mix}}.
\]
It remains to prove the reverse inequality.

Fix arbitrary unitaries $U,V$ and arbitrary density operators $\sigma_B,\tau_E$.
Choose spectral decompositions
\[
\sigma_B = \sum_r p_r |\beta_r\rangle\langle \beta_r|,
\qquad
\tau_E = \sum_s q_s |\eta_s\rangle\langle \eta_s|,
\]
with $p_r,q_s \ge 0$ and $\sum_r p_r = \sum_s q_s = 1$. Define the corresponding
pure-ancilla encoder and decoder channels by
\[
\Enc_r := \Enc_{U,|\beta_r\rangle\langle \beta_r|},
\qquad
\Dec_s := \Dec_{V,|\eta_s\rangle\langle \eta_s|}.
\]
By linearity,
\[
\Enc_{U,\sigma_B} = \sum_r p_r \Enc_r,
\qquad
\Dec_{V,\tau_E} = \sum_s q_s \Dec_s.
\]
Applying separate linearity twice gives
\[
F_\mu\!\left(\Enc_{U,\sigma_B}, \Dec_{V,\tau_E}\right)
=
\sum_{r,s} p_r q_s\, F_\mu(\Enc_r,\Dec_s).
\]
Hence there exist indices $r_\ast,s_\ast$ such that
\[
F_\mu\!\left(\Enc_{U,\sigma_B}, \Dec_{V,\tau_E}\right)
\le
F_\mu(\Enc_{r_\ast},\Dec_{s_\ast}).
\]

Now choose unitaries $W_B$ on $H_B$ and $W_E$ on $H_E$ satisfying
\[
W_B|0^{n_B}\rangle = |\beta_{r_\ast}\rangle,
\qquad
W_E|0^{n_E}\rangle = |\eta_{s_\ast}\rangle,
\]
and define
\[
\widetilde U := U (I_A \otimes W_B),
\qquad
\widetilde V := V (I_C \otimes W_E).
\]
Then, for every input state $\rho_A$,
\begin{align*}
\Enc_{\widetilde U,|0^{n_B}\rangle\langle 0^{n_B}|}(\rho_A)
&=
\operatorname{Tr}_D\!\left[
\widetilde U
(\rho_A \otimes |0^{n_B}\rangle\langle 0^{n_B}|)
\widetilde U^\dagger
\right] \\
&=
\operatorname{Tr}_D\!\left[
U
(\rho_A \otimes W_B|0^{n_B}\rangle\langle 0^{n_B}|W_B^\dagger)
U^\dagger
\right] \\
&=
\operatorname{Tr}_D\!\left[
U
(\rho_A \otimes |\beta_{r_\ast}\rangle\langle \beta_{r_\ast}|)
U^\dagger
\right] \\
&= \Enc_{r_\ast}(\rho_A).
\end{align*}
Likewise, for every input state $\rho_C$,
\[
\Dec_{\widetilde V,|0^{n_E}\rangle\langle 0^{n_E}|}(\rho_C)
=
\Dec_{s_\ast}(\rho_C).
\]
Therefore
\[
F_\mu\!\left(\Enc_{U,\sigma_B}, \Dec_{V,\tau_E}\right)
\le
F_\mu\!\left(
\Enc_{\widetilde U,|0^{n_B}\rangle\langle 0^{n_B}|},
\Dec_{\widetilde V,|0^{n_E}\rangle\langle 0^{n_E}|}
\right)
\le M_0.
\]
Since the mixed-ancilla implementation $(U,V,\sigma_B,\tau_E)$ was arbitrary, it follows that
\[
M_{\mathrm{mix}} \le M_0.
\]
Together with $M_0 \le M_{\mathrm{mix}}$, this proves $M_{\mathrm{mix}} = M_0$.
\end{proof}

\subsection{Extreme points and low-rank implementations}

\begin{lemma}\label{lem:extreme-optimum}
There exists an optimal pair $(\widehat{\Enc},\widehat{\Dec})$ for Eq.~\eqref{eq:objective} such that both channels are extreme points of their respective CPTP sets.
\end{lemma}

\begin{proof}
The set of CPTP maps between fixed finite-dimensional spaces is convex and compact. By continuity, the maximum of $F_\mu$ is attained. Fix an optimal decoder $\Dec_0$. Because $\Enc\mapsto F_\mu(\Enc,\Dec_0)$ is continuous and linear by Proposition~\ref{prop:separate-linear}, Bauer's maximum principle implies that it attains its maximum at an extreme encoder $\widehat{\Enc}$. Now fix this encoder. The same argument applied to $\Dec\mapsto F_\mu(\widehat{\Enc},\Dec)$ yields an extreme decoder $\widehat{\Dec}$, and $(\widehat{\Enc},\widehat{\Dec})$ is still optimal.
\end{proof}

\begin{theorem}[\citeMy{ruskai2007some}, Theorem~1]\label{thm:extreme-choi-rank}
Let $\Phi:L(H_{d_1})\to L(H_{d_2})$ be an extreme CPTP channel. Then its Choi rank is at most $d_1$. More generally, the closure of the set of extreme CPTP channels equals the set of CPTP channels with Choi rank at most $d_1$ \citeMy{choi1975completely, ruskai2007some, friedland2016extreme, memarzadeh2022group}.
\end{theorem}

\begin{proposition}[Implementing a low-rank encoder]\label{prop:enc-implementation}
Let $\Phi:L(H_{2^n})\to L(H_{2^k})$ be a CPTP channel of Kraus rank at most $2^n$. Then there exists an $(n+k)$-qubit unitary $U$ such that, under a relabeling $AB\equiv CD$ with $|A|=n$, $|B|=k$, $|C|=k$, and $|D|=n$,
\begin{equation}\label{eq:enc-implementation}
\Phi(\rho)
=
\Tr_D\Bigl(U(\rho_A\otimes \proj{0^k}_B)U^\dagger\Bigr).
\end{equation}
\end{proposition}

\begin{proof}
Choose Kraus operators $K_0,\dots,K_{2^n-1}:H_{2^n}\to H_{2^k}$ for $\Phi$, padding with zero operators if necessary. Define
\begin{equation}\label{eq:enc-isometry}
\widetilde U:=\sum_{i=0}^{2^n-1} K_i\otimes \ket{i}_D : H_{2^n}\to H_{2^k}\otimes H_{2^n}.
\end{equation}
Then
\[
\widetilde U^\dagger \widetilde U = \sum_{i=0}^{2^n-1} K_i^\dagger K_i = I_{2^n},
\]
so $\widetilde U$ is an isometry. Identify the domain of $\widetilde U$ with the $2^n$-dimensional subspace $H_A\otimes \Span\{\ket{0^k}_B\}$ of $H_A\otimes H_B$ and extend $\widetilde U$ to a unitary $U$ on $H_A\otimes H_B\cong H_C\otimes H_D$. Then
\[
\Tr_D\bigl(\widetilde U\rho\widetilde U^\dagger\bigr)
=
\sum_{i=0}^{2^n-1} K_i\rho K_i^\dagger
=
\Phi(\rho),
\]
which is Eq.~\eqref{eq:enc-implementation}.
\end{proof}

\begin{proposition}[Implementing a low-rank decoder]\label{prop:dec-implementation}
Let $\Phi:L(H_{2^k})\to L(H_{2^n})$ be a CPTP channel of Kraus rank at most $2^k$. Then there exists an $(n+k)$-qubit unitary $V$ such that, under a relabeling $CE\equiv FG$ with $|C|=k$, $|E|=n$, $|F|=n$, and $|G|=k$,
\begin{equation}\label{eq:dec-implementation}
\Phi(\rho)
=
\Tr_G\Bigl(V(\rho_C\otimes \proj{0^n}_E)V^\dagger\Bigr).
\end{equation}
\end{proposition}

\begin{proof}
Choose Kraus operators $L_0,\dots,L_{2^k-1}:H_{2^k}\to H_{2^n}$ for $\Phi$, padding with zeros if necessary, and define
\begin{equation}\label{eq:dec-isometry}
\widetilde V:=\sum_{i=0}^{2^k-1} L_i\otimes \ket{i}_G:H_{2^k}\to H_{2^n}\otimes H_{2^k}.
\end{equation}
Again $\widetilde V^\dagger \widetilde V=I_{2^k}$, so $\widetilde V$ is an isometry. Identify $H_{2^k}$ with the subspace $H_C\otimes \Span\{\ket{0^n}_E\}$ of $H_C\otimes H_E$ and extend $\widetilde V$ to a unitary $V$ on $H_C\otimes H_E\cong H_F\otimes H_G$. Then
\[
\Tr_G\bigl(\widetilde V\rho\widetilde V^\dagger\bigr)=\sum_{i=0}^{2^k-1} L_i\rho L_i^\dagger=\Phi(\rho),
\]
which is Eq.~\eqref{eq:dec-implementation}.
\end{proof}

\subsection{Proof of Theorem~\ref{thm:sufficiency}}
\label{sec:universal_kn_sufficient_proof}

\myBlue
We recall and prove Theorem~\ref{thm:sufficiency}.
\UniversalKNSufficient*
\color{black}

\begin{proof}
Let $(\widehat{\Enc},\widehat{\Dec})$ be an optimal pair as in Lemma~\ref{lem:extreme-optimum}. By Theorem~\ref{thm:extreme-choi-rank},
\[
\krank(\widehat{\Enc})\le 2^n,
\qquad
\krank(\widehat{\Dec})\le 2^k.
\]
Therefore Propositions~\ref{prop:enc-implementation} and \ref{prop:dec-implementation} realize $\widehat{\Enc}$ and $\widehat{\Dec}$ by the encoder and decoder parts of an $(n,k,k,n)$-QAE. The resulting QAE attains
\[
F_\mu(\widehat{\Enc},\widehat{\Dec})=\max_{\Enc,\Dec} F_\mu(\Enc,\Dec),
\]
as claimed.
\end{proof}

\section{Haar analysis and mixtures with perfectly compressible priors}\label{sec:haar-appendix}

We write $d:=2^n$ and $m:=2^k$ for the source and bottleneck dimensions.

\begin{lemma}\label{lem:haar-identities}
Let $\ket{\psi}$ be Haar-random on $\C^d$.
\begin{enumerate}[label=(\roman*),leftmargin=2.2em]
\item For all operators $X,Y\in L(\C^d)$,
\begin{equation}\label{eq:haar-second-moment}
\mathbb{E}_{\psi\sim\mathrm{Haar}}\bigl[\bra{\psi}X\ket{\psi}\,\bra{\psi}Y\ket{\psi}\bigr]
=
\frac{\Tr(XY)+\Tr(X)\Tr(Y)}{d(d+1)}.
\end{equation}
\item If $\ket{\psi}$ is Haar-random on $H_A\otimes H_B$ and $\rho_A:=\Tr_B(\proj{\psi})$, then
\begin{equation}\label{eq:haar-purity}
\mathbb{E}_{\psi\sim\mathrm{Haar}}\Tr(\rho_A^2)
=
\frac{\dim H_A+\dim H_B}{1+\dim H_A\dim H_B}.
\end{equation}
\end{enumerate}
\end{lemma}

\begin{proposition}[Optimal fidelity for the Haar prior]\label{prop:haar-opt}
For the Haar prior on pure states in $\C^d$, one has
\begin{equation}\label{eq:haar-opt}
\sup_{\Enc,\Dec} F_{\mu_{\mathrm{Haar}}}(\Enc,\Dec)
=
\frac{d+m^2}{d(d+1)}.
\end{equation}
An optimal encoder--decoder pair is obtained by identifying $\C^d\cong \C^m\otimes \C^{d/m}$ and setting
\begin{equation}\label{eq:haar-opt-pair}
\Enc(\rho)=\Tr_B(\rho),
\qquad
\Dec(\sigma)=\sigma\otimes \frac{I_B}{d/m}.
\end{equation}
By Haar symmetry one may replace $I_B/(d/m)$ in Eq.~\eqref{eq:haar-opt-pair} by any pure state on $H_B$. In particular, a rank-one decoder is optimal.
\end{proposition}

\begin{proof}
Let $\Phi:=\Dec\circ\Enc$. Choose Kraus operators $\{E_j\}_j$ for $\Enc$ and $\{D_i\}_i$ for $\Dec$, and write $K_{ij}:=D_iE_j$. For a pure input $\ket{\psi}$,
\[
\bra{\psi}\Phi(\proj{\psi})\ket{\psi} = \sum_{i,j} \abs{\bra{\psi}K_{ij}\ket{\psi}}^2.
\]
Applying Eq.~\eqref{eq:haar-second-moment} with $X=K_{ij}$ and $Y=K_{ij}^\dagger$ gives
\begin{align}
\mathbb{E}_{\psi\sim\mathrm{Haar}}\bigl[\bra{\psi}\Phi(\proj{\psi})\ket{\psi}\bigr]
&=
\frac1{d(d+1)}\sum_{i,j}\Bigl(\Tr(K_{ij}K_{ij}^\dagger)+\abs{\Tr K_{ij}}^2\Bigr) \notag\\
&=
\frac1{d+1}+\frac1{d(d+1)}\sum_{i,j}\abs{\Tr(D_iE_j)}^2.
\label{eq:haar-upper-start}
\end{align}
Now $E_jD_i$ is an $m\times m$ matrix, so
\[
\abs{\Tr(D_iE_j)}^2 = \abs{\Tr(E_jD_i)}^2 \le m\,\Tr(D_i^\dagger E_j^\dagger E_j D_i).
\]
Summing over $i,j$ yields
\[
\sum_{i,j}\abs{\Tr(D_iE_j)}^2
\le
m\sum_i \Tr\!\left(D_i^\dagger\Bigl(\sum_j E_j^\dagger E_j\Bigr)D_i\right)
=
m\sum_i \Tr(D_i^\dagger D_i)
=
m^2.
\]
Substituting into Eq.~\eqref{eq:haar-upper-start} gives the upper bound in Eq.~\eqref{eq:haar-opt}.

For achievability, identify $\C^d\cong \C^m\otimes \C^{d/m}=H_A\otimes H_B$, take $\Enc(\rho)=\Tr_B(\rho)$, and $\Dec(\sigma)=\sigma\otimes I_B/(d/m)$. If $\rho_A:=\Tr_B(\proj{\psi})$, then
\[
\bra{\psi}\Dec(\Enc(\proj{\psi}))\ket{\psi}
=
\Tr\!\left(\proj{\psi}\Bigl(\rho_A\otimes \frac{I_B}{d/m}\Bigr)\right)
=
\frac{1}{d/m}\Tr(\rho_A^2).
\]
Averaging and using Eq.~\eqref{eq:haar-purity} with $\dim H_A=m$ and $\dim H_B=d/m$ gives
\[
\mathbb{E}_{\psi\sim\mathrm{Haar}}\bigl[\bra{\psi}\Dec(\Enc(\proj{\psi}))\ket{\psi}\bigr]
=
\frac{1}{d/m}\cdot \frac{m+d/m}{d+1}
=
\frac{d+m^2}{d(d+1)}.
\]
This proves Eq.~\eqref{eq:haar-opt}. Finally, because the Haar measure is invariant under unitary rotations on $H_B$, replacing $I_B/(d/m)$ by any pure state on $H_B$ leaves the average fidelity unchanged.
\end{proof}

\begin{proposition}[Mixtures with a perfectly compressible prior]\label{prop:haar-mixture}
Let $\mu'=\mu_{\mathrm{pc}}$ be a prior on pure states in $\C^d$ such that there exist a unitary $U\in \mathrm{U}(d)$, a fixed reference state $\ket{\varphi_{\mathrm{ref}}}$, and an $m$-dimensional subsystem $C$ with
\[
U\ket{\psi}_A = \ket{\psi'}_C\otimes \ket{\varphi_{\mathrm{ref}}}_D
\qquad
\textnormal{for all }\ket{\psi}\in\supp(\mu').
\]
Let $\mu = c\,\mu_{\mathrm{Haar}}+(1-c)\mu'$ with $0<c<1$. Then the maps
\[
\Enc(\rho):=\Tr_D(U\rho U^\dagger),
\qquad
\Dec(\sigma):=U^\dagger\bigl(\sigma\otimes \proj{\varphi_{\mathrm{ref}}}\bigr)U
\]
simultaneously maximize the fidelity for the Haar part and reconstruct every state in $\supp(\mu')$ perfectly. Consequently,
\begin{equation}\label{eq:haar-mixture}
F_\mu(\Enc,\Dec)
=
1-c\,\frac{(d+m)(d-m)}{d(d+1)}.
\end{equation}
\end{proposition}

\begin{proof}
The chosen pair reconstructs $\mu'$ perfectly by construction, so its contribution to the average fidelity is $1-c$. On the Haar component the same pair is of the form considered in Proposition~\ref{prop:haar-opt}, hence contributes $c(d+m^2)/(d(d+1))$. Therefore
\[
F_\mu(\Enc,\Dec)=(1-c)+c\,\frac{d+m^2}{d(d+1)}=1-c\,\frac{d^2-m^2}{d(d+1)},
\]
which is Eq.~\eqref{eq:haar-mixture}.
\end{proof}

\section{The source $\mu_{1,\eps}$ and the encoder ancilla lower bound}\label{sec:mu1-lb}

Throughout this appendix we set
\[
d:=2^n,
\qquad
m:=2^k,
\qquad
q:=d-1,
\qquad
H_\perp:=\ket{0}^{\perp}\subset \C^d,
\qquad
\dim H_\perp=q.
\]
For $\eps\in(0,1)$ and a unit vector $\ket{\varphi}\in H_\perp$, write
\[
\ket{\psi_\eps(\varphi)}:=\sqrt{1-\eps}\ket{0}+\sqrt{\eps}\ket{\varphi},
\]
and define
\begin{equation}\label{eq:F-eps}
F_\eps(\Enc,\Dec)
:=
\mathbb{E}_{\varphi}\bigl[\bra{\psi_\eps(\varphi)}(\Dec\circ\Enc)(\proj{\psi_\eps(\varphi)})\ket{\psi_\eps(\varphi)}\bigr],
\end{equation}
where $\ket{\varphi}$ is Haar-uniform on the unit sphere of $H_\perp$.

\subsection{Two preliminary lemmas}\label{subsec:mu1-prelims}

\begin{lemma}\label{lem:compact-first-order}
Let $K\neq\varnothing$ be compact, and let $A,B,C:K\to\mathbb{R}$ be bounded and continuous. Define
\[
f(\eps):=\max_{x\in K}\bigl(A(x)+\eps B(x)+\eps^2 C(x)\bigr),
\qquad
\eps\ge 0.
\]
If
\[
A_{\max}:=\max_{x\in K} A(x),
\qquad
S:=\{x\in K:A(x)=A_{\max}\},
\]
then
\begin{equation}\label{eq:compact-first-order}
\limsup_{\eps\to 0^+}\frac{f(\eps)-A_{\max}}{\eps}
\le
\sup_{x\in S} B(x).
\end{equation}
\end{lemma}

\begin{proof}
For $\eps>0$, choose $x_\eps\in K$ such that
\[
f(\eps)=A(x_\eps)+\eps B(x_\eps)+\eps^2 C(x_\eps).
\]
Since $A(x_\eps)\le A_{\max}$,
\begin{equation}\label{eq:compact-first-order-basic}
\frac{f(\eps)-A_{\max}}{\eps}
\le
B(x_\eps)+\eps C(x_\eps).
\end{equation}
Choose a sequence $\eps_r\to 0^+$ along which the left-hand side tends to the limsup. By compactness of $K$, some subsequence of $x_{\eps_r}$ converges to a point $x_*\in K$. Passing to that subsequence and using continuity of $A,B,C$ gives
\[
A(x_*)=A_{\max},
\qquad
x_*\in S.
\]
Taking the limit in Eq.~\eqref{eq:compact-first-order-basic} along the same subsequence yields
\[
\limsup_{\eps\to 0^+}\frac{f(\eps)-A_{\max}}{\eps}
\le
B(x_*)
\le
\sup_{x\in S} B(x),
\]
which proves Eq.~\eqref{eq:compact-first-order}.
\end{proof}

\begin{lemma}\label{lem:quadratic-eps}
For every fixed CPTP map $\Phi:L(\C^d)\to L(\C^d)$, the function
\[
\eps\longmapsto F_\eps(\Phi)
:=
\mathbb{E}_{\varphi}\bigl[\bra{\psi_\eps(\varphi)}\Phi(\proj{\psi_\eps(\varphi)})\ket{\psi_\eps(\varphi)}\bigr]
\]
is a polynomial of degree at most $2$ in $\eps$. More precisely, there exist real numbers $A(\Phi),B(\Phi),C(\Phi)$ such that
\[
F_\eps(\Phi)=A(\Phi)+\eps B(\Phi)+\eps^2 C(\Phi)
\qquad
\textnormal{for all }\eps\in[0,1],
\]
and
\[
A(\Phi)=\bra{0}\Phi(\proj{0})\ket{0}.
\]
\end{lemma}

\begin{proof}
Write
\[
A_0:=\proj{0},
\qquad
A_1:=\proj{\varphi},
\qquad
C_+:=\ketbra{0}{\varphi},
\qquad
C_-:=\ketbra{\varphi}{0}.
\]
Then
\[
\proj{\psi_\eps(\varphi)}=(1-\eps)A_0+\eps A_1+\sqrt{\eps(1-\eps)}\,(C_++C_-).
\]
Because $(X,Y)\mapsto \Tr(\Phi(X)Y)$ is bilinear, the quantity
\[
\Tr\bigl(\Phi(\proj{\psi_\eps(\varphi)})\proj{\psi_\eps(\varphi)}\bigr)
\]
is a finite linear combination of monomials in $(1-\eps)$, $\eps$, and $\sqrt{\eps(1-\eps)}$. Any term containing unequal numbers of $\ket{\varphi}$ and $\bra{\varphi}$ picks up a nontrivial phase under $\ket{\varphi}\mapsto e^{i\theta}\ket{\varphi}$ and hence has zero Haar average. The surviving terms have equal numbers of $\ket{\varphi}$ and $\bra{\varphi}$, so they contribute only powers $\eps^0,\eps^1,\eps^2$. At $\eps=0$ the state is $\ket{0}$, giving the formula for $A(\Phi)$.
\end{proof}

\subsection{Asymptotic separation for the conventional architecture when $(n,k)=(2,1)$}\label{subsec:conventional-special-case}

We first analyze the conventional architecture in the minimal nontrivial case. Here
\[
\ket{\psi}=\sqrt{1-\eps}\ket{00}+\sqrt{\eps}\ket{\varphi},
\qquad
\ket{\varphi}\in\Span\{\ket{01},\ket{10},\ket{11}\},
\qquad
\norm{\varphi}=1,
\]
and $\ket{\varphi}$ is Haar-random in the three-dimensional orthogonal complement of $\ket{00}$.

Let $U,V\in \mathrm{U}(4)$ be the encoder and decoder unitaries of a conventional $(2,1,0,1)$-QAE. Define the decoder isometry
\[
W:=V(I\otimes \ket{0}):H_2\to H_4,
\qquad
W^\dagger W=I_2,
\]
and for $b\in\{0,1\}$ set
\[
\widetilde K_b:=(I\otimes \bra{b})U:H_4\to H_2,
\qquad
K_b:=W\widetilde K_b:H_4\to H_4.
\]
Then
\begin{equation}\label{eq:PhiUV}
\Phi_{U,V}(\rho)=\sum_{b=0}^1 K_b\rho K_b^\dagger.
\end{equation}
Because $U$ is unitary,
\[
\widetilde K_b\widetilde K_{b'}^\dagger=(I\otimes \bra b)UU^\dagger(I\otimes \ket{b'})=\delta_{bb'}I_2,
\]
so
\begin{equation}\label{eq:Kb-left-support}
K_bK_{b'}^\dagger=\delta_{bb'}P,
\qquad
P:=WW^\dagger.
\end{equation}
Thus $P$ is the rank-$2$ projector onto $\ran(W)$. If we set
\[
Q_b:=K_b^\dagger K_b,
\]
then each $Q_b$ is a rank-$2$ projection, $Q_0Q_1=0$, and trace preservation implies
\begin{equation}\label{eq:Q-projections}
Q_0+Q_1=I_4.
\end{equation}
For a pure input $\ket{\psi}$,
\begin{equation}\label{eq:conventional-fidelity}
\bra{\psi}\Phi_{U,V}(\proj{\psi})\ket{\psi}
=
\sum_{b=0}^1 \abs{\bra{\psi}K_b\ket{\psi}}^2.
\end{equation}

For each fixed $(U,V)$, define
\[
f_{U,V}(\eps):=\mathbb{E}_{\varphi}\bigl[\bra{\psi}\Phi_{U,V}(\proj{\psi})\ket{\psi}\bigr].
\]
The Kraus pair $(K_0,K_1)$ may be mixed by an arbitrary unitary $S\in \mathrm{U}(2)$ without changing the channel, and under this mixing the vector
\[
v(\psi):=(\bra{\psi}K_0\ket{\psi},\bra{\psi}K_1\ket{\psi})\in\C^2
\]
transforms as $v(\psi)\mapsto Sv(\psi)$, leaving Eq.~\eqref{eq:conventional-fidelity} unchanged. We may therefore choose the Kraus gauge so that
\begin{equation}\label{eq:gauge-choice}
\bra{00}K_0\ket{00}=\sqrt{f_0}\ge 0,
\qquad
\bra{00}K_1\ket{00}=0,
\qquad
f_0:=\sum_{b=0}^1 \abs{\bra{00}K_b\ket{00}}^2\in[0,1].
\end{equation}
If $f_0=0$, then the constant term of $f_{U,V}(\eps)$ is $0$, so the pair is clearly suboptimal. Henceforth assume $f_0>0$.

For $b\in\{0,1\}$ define
\[
f_b^{(1)}(\varphi):=\bra{\psi}K_b\ket{\psi}.
\]
Expanding,
\[
f_b^{(1)}(\varphi)
=
(1-\eps)\bra{00}K_b\ket{00}
+\sqrt{\eps(1-\eps)}\bigl(\bra{00}K_b\ket{\varphi}+\bra{\varphi}K_b\ket{00}\bigr)
+\eps\bra{\varphi}K_b\ket{\varphi}.
\]
Since $\sqrt{\eps(1-\eps)}=\sqrt{\eps}+o(\eps)$, we may write
\[
f_b^{(1)}(\varphi)
=
f_b^{(2)}+\sqrt{\eps}\,f_b^{(3)}(\varphi)+\eps f_b^{(4)}(\varphi)+o(\eps),
\]
where
\[
f_b^{(2)}:=\bra{00}K_b\ket{00},
\qquad
f_b^{(3)}(\varphi):=\bra{00}K_b\ket{\varphi}+\bra{\varphi}K_b\ket{00},
\qquad
f_b^{(4)}(\varphi):=\bra{\varphi}K_b\ket{\varphi}-f_b^{(2)}.
\]
Hence
\begin{align*}
\abs{f_b^{(1)}(\varphi)}^2
&=
\abs{f_b^{(2)}}^2
+2\sqrt{\eps}\,\Re\!\bigl((f_b^{(2)})^*f_b^{(3)}(\varphi)\bigr)\\
&\qquad
+\eps\Bigl(\abs{f_b^{(3)}(\varphi)}^2+2\Re\!\bigl((f_b^{(2)})^*f_b^{(4)}(\varphi)\bigr)\Bigr)
+o(\eps).
\end{align*}
Taking the Haar expectation removes the $\sqrt{\eps}$ term because $\mathbb{E}[\ket{\varphi}]=0$:
\[
\mathbb{E}\abs{f_b^{(1)}(\varphi)}^2
=
\abs{f_b^{(2)}}^2
+\eps\left(\mathbb{E}\abs{f_b^{(3)}(\varphi)}^2+2\Re\!\bigl((f_b^{(2)})^*\mathbb{E}[f_b^{(4)}(\varphi)]\bigr)\right)+o(\eps).
\]
Summing over $b$ and using Eq.~\eqref{eq:gauge-choice},
\begin{equation}\label{eq:conventional-expansion-start}
f_{U,V}(\eps)
=
f_0
+\eps\sum_{b=0}^1\left(\mathbb{E}\abs{f_b^{(3)}(\varphi)}^2+2\Re\!\bigl((f_b^{(2)})^*\mathbb{E}[f_b^{(4)}(\varphi)]\bigr)\right)
+o(\eps).
\end{equation}

Let
\[
\Pi:=I_4-\proj{00},
\]
the projector onto $\Span\{\ket{01},\ket{10},\ket{11}\}$. For Haar-random $\ket{\varphi}$ in this three-dimensional subspace,
\begin{equation}\label{eq:three-dim-moments}
\mathbb{E}[\proj{\varphi}]=\frac{\Pi}{3},
\qquad
\mathbb{E}[\bra{\varphi}A\ket{\varphi}]=\frac{\Tr(\Pi A)}{3}.
\end{equation}
Moreover,
\[
\mathbb{E}\bigl[\braket{u}{\varphi}\braket{v}{\varphi}\bigr]=0
\qquad
\text{for all }u,v,
\]
by invariance under the global phase $\ket{\varphi}\mapsto e^{i\theta}\ket{\varphi}$. Therefore
\begin{equation}\label{eq:f3-moment}
\mathbb{E}\abs{f_b^{(3)}(\varphi)}^2
=
\frac{\norm{\Pi K_b^\dagger\ket{00}}^2+\norm{\Pi K_b\ket{00}}^2}{3}.
\end{equation}
Set
\[
p_0:=\bra{00}P\ket{00},
\qquad
q_b:=\bra{00}Q_b\ket{00},
\qquad
q_0+q_1=1.
\]
Then by Eqs.~\eqref{eq:Kb-left-support} and \eqref{eq:Q-projections},
\[
\norm{K_b^\dagger\ket{00}}^2=\bra{00}P\ket{00}=p_0,
\qquad
\norm{K_b\ket{00}}^2=\bra{00}Q_b\ket{00}=q_b.
\]
Since $\norm{\Pi v}^2=\norm{v}^2-\abs{\braket{00}{v}}^2$,
\[
\norm{\Pi K_b^\dagger\ket{00}}^2=p_0-\abs{f_b^{(2)}}^2,
\qquad
\norm{\Pi K_b\ket{00}}^2=q_b-\abs{f_b^{(2)}}^2.
\]
Substituting into Eq.~\eqref{eq:f3-moment} and summing over $b$ gives
\begin{equation}\label{eq:f3-sum}
\sum_{b=0}^1 \mathbb{E}\abs{f_b^{(3)}(\varphi)}^2
=
\frac{2p_0+1-2f_0}{3}.
\end{equation}
Next define
\[
\tau_b:=\Tr(\Pi K_b\Pi).
\]
By Eq.~\eqref{eq:three-dim-moments},
\[
\mathbb{E}[f_b^{(4)}(\varphi)]
=
\frac{\tau_b}{3}-f_b^{(2)}.
\]
Hence
\begin{align}\label{eq:f4-sum}
\sum_{b=0}^1 2\Re\!\bigl((f_b^{(2)})^*\mathbb{E}[f_b^{(4)}(\varphi)]\bigr)
&=
\frac23\Re\!\left(\sum_{b=0}^1 (f_b^{(2)})^*\tau_b\right)-2f_0 \notag\\
&=
\frac23\sqrt{f_0}\,\Re(\tau_0)-2f_0,
\end{align}
where the last equality uses the gauge choice in Eq.~\eqref{eq:gauge-choice}.

The set $\mathrm{U}(4)\times \mathrm{U}(4)$ is compact, and by Lemmas~\ref{lem:compact-first-order} and \ref{lem:quadratic-eps} the constant term of $f_{U,V}(\eps)$ is maximized at pairs with $f_0=1$ (for instance $U=V=I_4$). Therefore
\begin{equation}\label{eq:conventional-limsup}
\limsup_{\eps\to 0^+}\frac{\max_{U,V}f_{U,V}(\eps)-1}{\eps}
\le
\sup_{(U,V):f_0=1}
\left(
\frac{2p_0+1-2f_0}{3}
+\frac23\sqrt{f_0}\,\Re(\tau_0)
-2f_0
\right).
\end{equation}
If $f_0=1$, then $K_0\ket{00}=\ket{00}$ and $K_1\ket{00}=0$. Because $K_0$ has rank $2$, there exists a unit vector $\ket{u}\in\Span\{\ket{01},\ket{10},\ket{11}\}$ such that
\[
\ran(K_0)=\Span\{\ket{00},\ket{u}\}.
\]
Consequently $\Pi K_0 \Pi$ has rank at most $1$. Also $K_0^\dagger K_0=Q_0\le I_4$, so $\norm{K_0}_{\infty}\le 1$, and hence
\[
\norm{\Pi K_0\Pi}_{\infty}\le 1,
\qquad
\norm{\Pi K_0\Pi}_\mathrm{F}\le 1.
\]
Therefore
\begin{equation}\label{eq:tau-bound}
\abs{\tau_0}
=
\abs{\Tr(\Pi K_0\Pi)}
\le
\sqrt{\rank(\Pi K_0\Pi)}\;\norm{\Pi K_0\Pi}_\mathrm{F}
\le 1.
\end{equation}
Combining $p_0\le 1$ with Eq.~\eqref{eq:tau-bound}, the expression on the right-hand side of Eq.~\eqref{eq:conventional-limsup} is at most
\[
\frac{2p_0+2\Re(\tau_0)-7}{3}\le -1.
\]
Thus
\begin{equation}\label{eq:conventional-bound-final}
\liminf_{\eps\to 0^+}\frac{1-\max_{U,V}f_{U,V}(\eps)}{\eps}
\ge 1.
\end{equation}
The identity choice $U=V=I_4$ attains equality in Eq.~\eqref{eq:conventional-bound-final}.

For unrestricted CPTP encoder--decoder pairs, consider the projector
\[
\Pi':=\proj{00}+\proj{01}
\]
and the channel
\[
\Phi(\rho)=\Pi'\rho\Pi'+\Tr\bigl((I-\Pi')\rho\bigr)\proj{00}.
\]
It admits an encoder with Kraus operators
\[
E_0=\ketbra{0}{00}+\ketbra{1}{01},
\qquad
E_1=\ketbra{0}{10},
\qquad
E_2=\ketbra{0}{11},
\]
and a decoder with a single Kraus operator
\[
D_0=\ketbra{00}{0}+\ketbra{01}{1}.
\]
Write
\[
\ket{\varphi}=\varphi_1\ket{01}+\sqrt{1-\abs{\varphi_1}^2}\,\ket{\chi},
\qquad
\ket{\chi}\in\Span\{\ket{10},\ket{11}\},
\qquad
\norm{\chi}=1.
\]
Then
\[
\Pi'\ket{\psi}=\sqrt{1-\eps}\ket{00}+\sqrt{\eps}\,\varphi_1\ket{01},
\qquad
(I-\Pi')\ket{\psi}=\sqrt{\eps(1-\abs{\varphi_1}^2)}\;\ket{\chi},
\]
and therefore
\begin{align*}
\bra{\psi}\Phi(\proj{\psi})\ket{\psi}
&=
\norm{\Pi'\ket{\psi}}^4 + (1-\eps)\norm{(I-\Pi')\ket{\psi}}^2 \\
&=
1-\eps\bigl(1-\abs{\varphi_1}^2\bigr)+o(\eps).
\end{align*}
Since $\mathbb{E}\abs{\varphi_1}^2=1/3$,
\begin{equation}\label{eq:unrestricted-2over3}
\mathbb{E}\bigl[\bra{\psi}\Phi(\proj{\psi})\ket{\psi}\bigr]
=
1-\frac23\eps+o(\eps),
\end{equation}
which is strictly better than the conventional-architecture coefficient in Eq.~\eqref{eq:conventional-bound-final}.

\subsection{A universal first-order lower bound for all factorized channels}\label{subsec:general-first-order}

In the remainder of this appendix, let $\Phi=\Dec\circ\Enc$ be an arbitrary factorized channel with
\[
\Enc:L(\C^d)\to L(\C^m),
\qquad
\Dec:L(\C^m)\to L(\C^d),
\qquad
\Phi(\proj{0})=\proj{0}.
\]
Choose Kraus operators $\{E_\ell\}_\ell$ for $\Enc$ and $\{D_j\}_j$ for $\Dec$, and relabel the Kraus family $\{D_jE_\ell\}_{j,\ell}$ of $\Phi$ as $\{A_i\}_{i=1}^r$. Each $A_i$ has rank at most $m$.

\begin{lemma}[Block form]\label{lem:block-form}
For every $i$ there exists $\alpha_i\in\C$ such that
\[
A_i\ket{0}=\alpha_i\ket{0}.
\]
With respect to the decomposition $\C^d=\Span\{\ket{0}\}\oplus H_\perp$, we may write
\[
A_i=
\begin{pmatrix}
\alpha_i & v_i^\dagger \\
0 & B_i
\end{pmatrix},
\qquad
v_i\in H_\perp,
\qquad
B_i:H_\perp\to H_\perp,
\]
and the trace preserving condition becomes
\begin{equation}\label{eq:block-identities}
\sum_{i=1}^r \abs{\alpha_i}^2=1,
\qquad
\sum_{i=1}^r \alpha_i v_i=0,
\qquad
\sum_{i=1}^r\bigl(v_i v_i^\dagger + B_i^\dagger B_i\bigr)=I_{H_\perp}.
\end{equation}
In particular,
\begin{equation}\label{eq:trace-block}
\sum_{i=1}^r \norm{v_i}^2 + \sum_{i=1}^r \norm{B_i}_\mathrm{F}^2 = d-1.
\end{equation}
\end{lemma}

\begin{proof}
Because
\[
\proj{0}=\Phi(\proj{0})=\sum_{i=1}^r A_i\proj{0}A_i^\dagger,
\]
each positive semidefinite summand on the right has support contained in $\Span\{\ket{0}\}$. Hence $A_i\ket{0}\in \Span\{\ket{0}\}$, proving the first claim. Writing the block decomposition and comparing the matrix blocks in $\sum_i A_i^\dagger A_i=I_d$ gives Eq.~\eqref{eq:block-identities}. Taking traces yields Eq.~\eqref{eq:trace-block}.
\end{proof}

\begin{lemma}\label{lem:first-order-expansion}
Define
\[
S:=\sum_{i=1}^r \norm{v_i}^2,
\qquad
T:=\sum_{i=1}^r \norm{B_i}_\mathrm{F}^2.
\]
Then
\begin{equation}\label{eq:first-order-expansion}
F_\eps(\Phi)
=
1-2\eps+
\frac{\eps}{d-1}
\left(
S+2\Re\sum_{i=1}^r \alpha_i^*\Tr(B_i)
\right)
+o(\eps).
\end{equation}
Equivalently,
\begin{equation}\label{eq:cPhi}
c(\Phi):=\lim_{\eps\to 0^+}\frac{1-F_\eps(\Phi)}{\eps}
=
2-
\frac{1}{d-1}
\left(
S+2\Re\sum_{i=1}^r \alpha_i^*\Tr(B_i)
\right).
\end{equation}
\end{lemma}

\begin{proof}
Fix a unit vector $\ket{\varphi}\in H_\perp$ and set
\[
a:=\sqrt{1-\eps},
\qquad
b:=\sqrt{\eps},
\qquad
x_i(\varphi):=\braket{v_i}{\varphi},
\qquad
y_i(\varphi):=\bra{\varphi}B_i\ket{\varphi}.
\]
Then
\[
\bra{\psi_\eps(\varphi)}A_i\ket{\psi_\eps(\varphi)}
=
a^2\alpha_i+ab\,x_i(\varphi)+b^2 y_i(\varphi).
\]
Since
\[
\bra{\psi}\Phi(\proj{\psi})\ket{\psi}=
\sum_{i=1}^r \abs{\bra{\psi}A_i\ket{\psi}}^2,
\]
we obtain
\begin{align*}
F_\eps(\Phi)
&=
\sum_{i=1}^r \mathbb{E}_{\varphi}\Bigl[\abs{a^2\alpha_i+ab\,x_i(\varphi)+b^2 y_i(\varphi)}^2\Bigr] \\
&=
\sum_{i=1}^r\left(a^4\abs{\alpha_i}^2 + a^2b^2\,\mathbb{E}\abs{x_i(\varphi)}^2 + 2a^2b^2\Re\bigl(\alpha_i^*\mathbb{E}[y_i(\varphi)]\bigr)\right)+o(\eps).
\end{align*}
The terms of order $a^3b$ and $ab^3$ vanish after Haar averaging: $\mathbb{E}[x_i(\varphi)]=0$ by symmetry, and $\mathbb{E}[x_i(\varphi)^*y_i(\varphi)]=0$ because $x_i(\varphi)^*y_i(\varphi)$ acquires a nontrivial phase under $\ket{\varphi}\mapsto e^{i\theta}\ket{\varphi}$. By the first moments of Haar measure on the unit sphere of $H_\perp$,
\[
\mathbb{E}[x_i(\varphi)]=0,
\qquad
\mathbb{E}\abs{x_i(\varphi)}^2=\frac{\norm{v_i}^2}{d-1},
\qquad
\mathbb{E}[y_i(\varphi)]=\frac{\Tr(B_i)}{d-1}.
\]
Using $a^4=(1-\eps)^2=1-2\eps+O(\eps^2)$, $a^2b^2=\eps(1-\eps)=\eps+O(\eps^2)$, and $\sum_i \abs{\alpha_i}^2=1$, we obtain Eq.~\eqref{eq:first-order-expansion}. Equation~\eqref{eq:cPhi} is immediate.
\end{proof}

\begin{lemma}[Universal first-order lower bound]\label{lem:universal-first-order}
For every factorized channel $\Phi=\Dec\circ\Enc$ as above,
\begin{equation}\label{eq:universal-first-order}
c(\Phi)\ge \frac{d-m}{d-1}.
\end{equation}
\end{lemma}

\begin{proof}
Let
\[
I_+:=\{i:\alpha_i\neq 0\}.
\]
If $i\in I_+$, then $\ket{0}\in \im(A_i)$ and $\rank(A_i)\le m$, so $\rank(B_i)\le m-1$. Hence
\[
\abs{\Tr(B_i)}^2\le \rank(B_i)\,\norm{B_i}_\mathrm{F}^2\le (m-1)\norm{B_i}_\mathrm{F}^2
\qquad
(i\in I_+).
\]
Therefore,
\begin{align*}
\abs{\sum_{i=1}^r \alpha_i^*\Tr(B_i)}
&\le \sum_{i\in I_+} \abs{\alpha_i}\,\abs{\Tr(B_i)} \\
&\le \left(\sum_{i\in I_+}\abs{\alpha_i}^2\right)^{1/2}
\left(\sum_{i\in I_+}\abs{\Tr(B_i)}^2\right)^{1/2} \\
&\le \sqrt{(m-1)\sum_{i\in I_+}\norm{B_i}_\mathrm{F}^2}
\le \sqrt{(m-1)T}.
\end{align*}
Using $S+T=d-1$ from Eq.~\eqref{eq:trace-block},
\[
S+2\Re\sum_{i=1}^r \alpha_i^*\Tr(B_i)
\le d-1-T+2\sqrt{(m-1)T}.
\]
The right-hand side is maximized at $T=m-1$, where it equals $d+m-2$. Substituting into Eq.~\eqref{eq:cPhi} gives
\[
c(\Phi)
\ge
2-\frac{d+m-2}{d-1}
=
\frac{d-m}{d-1},
\]
which proves Eq.~\eqref{eq:universal-first-order}.
\end{proof}

Define the optimal first-order coefficient
\begin{equation}\label{eq:cstar}
c_*:=\frac{d-m}{d-1}.
\end{equation}
The reset construction below attains this bound.

\subsection{Classification of first-order-optimal channels}\label{subsec:classification}

\begin{proposition}\label{prop:classification-reset}
Let $\Phi=\Dec\circ\Enc$ factor through $\C^m$ and satisfy $\Phi(\proj{0})=\proj{0}$. If $c(\Phi)=c_*$, then
\begin{equation}\label{eq:reset-channel}
\Phi(\rho)=\Phi_P(\rho):=P\rho P+\Tr\bigl((I_d-P)\rho\bigr)\proj{0}
\end{equation}
for some rank-$m$ orthogonal projection $P$ with $\ket{0}\in\im(P)$.
\end{proposition}

\begin{lemma}\label{lem:trace-equality}
Let $B\neq 0$ be a linear operator on a finite-dimensional Hilbert space. If
\[
\abs{\Tr(B)}^2=\rank(B)\,\norm{B}_\mathrm{F}^2,
\]
then there exist a scalar $\beta\neq 0$ and an orthogonal projection $Q$ of rank $\rank(B)$ such that
\[
B=\beta Q.
\]
\end{lemma}

\begin{proof}
Let $r:=\rank(B)$, and let $Q$ be the orthogonal projection onto $(\ker B)^\perp$. Then
\[
\rank(Q)=\dim\bigl((\ker B)^\perp\bigr)=\rank(B)=r.
\]
Moreover, $B=BQ$.
Now equip the operator space with the Frobenius inner product
\[
\langle X,Y\rangle_{\mathrm F}:=\Tr(X^\dagger Y).
\]
Since $Q=Q^\dagger$, and since the trace is cyclic in finite dimensions, we get
\[
\Tr(B)=\Tr(BQ)=\Tr(QB)=\langle Q,B\rangle_{\mathrm F}.
\]
Hence, by the Cauchy--Schwarz inequality,
\[
\abs{\Tr(B)}^2
=\abs{\langle Q,B\rangle_{\mathrm F}}^2
\le \norm{Q}_{\mathrm F}^2\,\norm{B}_{\mathrm F}^2.
\]
Because $Q$ is an orthogonal projection,
\[
\norm{Q}_{\mathrm F}^2=\Tr(Q^\dagger Q)=\Tr(Q^2)=\Tr(Q)=\rank(Q)=r.
\]
Therefore
\[
\abs{\Tr(B)}^2\le r\,\norm{B}_{\mathrm F}^2.
\]
By hypothesis, equality holds. So equality holds in the Cauchy--Schwarz inequality, and therefore $B$ is a scalar multiple of $Q$:
\[
B=\beta Q
\]
for some $\beta\in\C$. Since $B\neq 0$, we must have $\beta\neq 0$.
\end{proof}

\begin{lemma}[The equality family]\label{lem:equality-family}
Assume $c(\Phi)=c_*$. After multiplying each Kraus operator $A_i$ by a phase of modulus $1$, we may assume:
\begin{enumerate}[label=(\roman*),leftmargin=2.2em]
\item each $\alpha_i$ is real and nonnegative.
\item if $\alpha_i=0$, then $B_i=0$.
\item if $\alpha_i>0$, then there exists a rank-$(m-1)$ orthogonal projection $P_i$ on $H_\perp$ such that
\[
B_i=\alpha_i P_i.
\]
\end{enumerate}
Moreover,
\begin{equation}\label{eq:equality-family-identities}
\sum_i \alpha_i^2=1,
\qquad
\sum_i \alpha_i v_i=0,
\qquad
\sum_i \bigl(v_i v_i^\dagger+\alpha_i^2 P_i\bigr)=I_{H_\perp}.
\end{equation}
\end{lemma}

\begin{proof}
Because $c(\Phi)=c_*$, every inequality in the proof of Lemma~\ref{lem:universal-first-order} is saturated. In particular,
\[
d-1-T+2\sqrt{(m-1)T}
\]
must attain its maximum at the actual value of $T$, so $T=m-1$.

Equality in
\[
\sum_{i\in I_+}\abs{\Tr(B_i)}^2
\le
(m-1)\sum_{i\in I_+}\norm{B_i}_\mathrm{F}^2
\le
(m-1)T
\]
forces $B_i=0$ whenever $\alpha_i=0$, and for every $i\in I_+$,
\[
\abs{\Tr(B_i)}^2=(m-1)\norm{B_i}_\mathrm{F}^2,
\qquad
\rank(B_i)=m-1.
\]
By Lemma~\ref{lem:trace-equality}, for each such $i$ there exist $\beta_i\neq 0$ and a rank-$(m-1)$ projection $P_i$ such that $B_i=\beta_i P_i$.

Equality in the Cauchy--Schwarz step
\[
\sum_i \abs{\alpha_i}\,\abs{\Tr(B_i)}
\le
\left(\sum_i \abs{\alpha_i}^2\right)^{1/2}
\left(\sum_i \abs{\Tr(B_i)}^2\right)^{1/2}
\]
shows that there exists $c\ge 0$ with $\abs{\Tr(B_i)}=c\abs{\alpha_i}$ for every $i$. Using $\sum_i\abs{\alpha_i}^2=1$ and $\sum_i \abs{\Tr(B_i)}^2=(m-1)^2$, we find $c=m-1$. Since $\abs{\Tr(B_i)}=(m-1)\abs{\beta_i}$, this implies $\abs{\beta_i}=\abs{\alpha_i}$ for every $i\in I_+$.

Finally, equality throughout
\[
\Re\sum_i \alpha_i^*\Tr(B_i)
\le
\abs{\sum_i \alpha_i^*\Tr(B_i)}
\le
\sum_i \abs{\alpha_i}\,\abs{\Tr(B_i)}
\]
implies that $\sum_i \alpha_i^*\Tr(B_i)$ is real and nonnegative, and that all complex numbers $\alpha_i^*\Tr(B_i)$ have a common phase. Multiplying each $A_i$ by an individual unit-modulus scalar does not change the channel, so we may choose those phases to make every $\alpha_i$ real and nonnegative. Then $\Tr(B_i)$ has the same phase, hence is also real and nonnegative, which forces $\beta_i=\alpha_i$. This proves the lemma.
\end{proof}

\begin{lemma}\label{lem:all-projections-coincide}
Under the hypotheses of Lemma~\ref{lem:equality-family}, all nonzero projections $P_i$ are equal.
\end{lemma}

\begin{proof}
Set $P_i:=0$ when $\alpha_i=0$, and define
\[
U_i:=\im(P_i)\subset H_\perp,
\qquad
N:=\sum_i \alpha_i^2 P_i,
\qquad
W:=\im(N)=\Span\Bigl(\bigcup_i U_i\Bigr).
\]
If the nonzero projections are not all equal, then at least two distinct $(m-1)$-dimensional subspaces occur among the $U_i$, and therefore
\begin{equation}\label{eq:dimW-ge-m}
\dim W\ge m.
\end{equation}
We will show that this contradicts the bottleneck dimension.

\smallskip
\paragraph{Step 1: the image of a factorized channel has dimension at most $m^2$.}
Since $\Phi=\Dec\circ\Enc$, the image of $\Phi$ as a linear map on the operator space is contained in the linear subspace $\Dec(L(\C^m))$. Hence
\begin{equation}\label{eq:image-dim-upper}
\dim \im(\Phi)\le \dim L(\C^m)=m^2.
\end{equation}

\smallskip
\paragraph{Step 2: a block formula for $\Phi$.}
Write an arbitrary $\rho\in L(\C^d)$ in block form as
\[
\rho=
\begin{pmatrix}
 a & x^\dagger \\
 y & Z
\end{pmatrix},
\qquad
 a\in\C,
\quad
 x,y\in H_\perp,
\quad
 Z\in L(H_\perp).
\]
Using the block decomposition from Lemma~\ref{lem:block-form} together with $B_i=\alpha_i P_i$, a direct multiplication yields
\[
A_i\rho A_i^\dagger
=
\begin{pmatrix}
\alpha_i^2 a + \alpha_i v_i^\dagger y + \alpha_i x^\dagger v_i + v_i^\dagger Z v_i & \alpha_i^2 x^\dagger P_i + \alpha_i v_i^\dagger ZP_i \\
\alpha_i^2 P_i y + \alpha_i P_i Z v_i & \alpha_i^2 P_i ZP_i
\end{pmatrix}.
\]
Summing over $i$, using $\sum_i \alpha_i v_i=0$ and
\[
\sum_i v_i v_i^\dagger = I_{H_\perp}-N,
\]
we obtain
\begin{equation}\label{eq:block-formula-Phi}
\Phi(\rho)=
\begin{pmatrix}
 a+\Tr\bigl((I_{H_\perp}-N)Z\bigr) & x^\dagger N + \sum_i \alpha_i v_i^\dagger ZP_i \\
 Ny + \sum_i \alpha_i P_i Zv_i & T(Z)
\end{pmatrix},
\end{equation}
where
\begin{equation}\label{eq:T-def}
T(Z):=\sum_i \alpha_i^2 P_i ZP_i.
\end{equation}

\smallskip
\paragraph{Step 3: a lower bound on $\dim\im(\Phi)$.}
From Eq.~\eqref{eq:block-formula-Phi} we introduce the four linear subspaces
\[
V_{00}:=\Span\{\proj{0}\},
\]
\[
V_{01}:=\left\{\begin{pmatrix}0 & x^\dagger N \\ 0 & 0\end{pmatrix}:x\in H_\perp\right\}
=\{\ketbra{0}{w}:w\in W\},
\]
\[
V_{10}:=\left\{\begin{pmatrix}0 & 0 \\ Ny & 0\end{pmatrix}:y\in H_\perp\right\}
=\{\ketbra{w}{0}:w\in W\},
\]
\[
V_{11}:=\left\{\begin{pmatrix}0 & 0 \\ 0 & T(Z)\end{pmatrix}:Z\in L(H_\perp)\right\}.
\]
These subspaces live in disjoint matrix blocks, so their sum is direct. In particular,
\begin{equation}\label{eq:image-dim-lower-raw}
\dim\im(\Phi)\ge 1+\dim V_{01}+\dim V_{10}+\dim V_{11}
\end{equation}
once we know that all four subspaces are contained in $\im(\Phi)$.

The inclusions $V_{00},V_{01},V_{10}\subseteq \im(\Phi)$ are immediate from Eq.~\eqref{eq:block-formula-Phi}. The inclusion $V_{11}\subseteq \im(\Phi)$ is slightly less immediate. Fix $Z\in L(H_\perp)$ and set
\[
\rho_Z:=\begin{pmatrix}0&0\\0&Z\end{pmatrix}.
\]
Then
\[
\Phi(\rho_Z)=
\begin{pmatrix}
\Tr\bigl((I_{H_\perp}-N)Z\bigr) & \sum_i \alpha_i v_i^\dagger ZP_i \\
\sum_i \alpha_i P_i Zv_i & T(Z)
\end{pmatrix}.
\]
Define
\[
c_Z:=\sum_i \alpha_i P_i Zv_i.
\]
For each $i$, $P_i Zv_i\in \im(P_i)\subseteq W$, hence $c_Z\in W$. Since $N$ is positive semidefinite and $\im(N)=W$, the restriction $N|_W:W\to W$ is invertible. Therefore there exists $y_Z\in W$ with $Ny_Z=c_Z$, and so
\[
\begin{pmatrix}0&0\\ c_Z&0\end{pmatrix}
=
\begin{pmatrix}0&0\\ Ny_Z&0\end{pmatrix}
\in V_{10}.
\]
Next define
\[
r_Z^\dagger:=\sum_i \alpha_i v_i^\dagger ZP_i.
\]
If $u\in W^\perp$, then $P_i u=0$ for every $i$, hence $r_Z^\dagger u=0$. Equivalently, there exists $u_Z\in W$ such that
\[
r_Z^\dagger = u_Z^\dagger \Pi_W,
\]
where $\Pi_W$ is the orthogonal projection onto $W$. Again using invertibility of $N|_W$, choose $x_Z\in W$ with $Nx_Z=u_Z$. Since $N=N\Pi_W=\Pi_W N$ and $N$ is self-adjoint,
\[
x_Z^\dagger N = u_Z^\dagger \Pi_W = r_Z^\dagger.
\]
Therefore
\[
\begin{pmatrix}0&r_Z^\dagger\\0&0\end{pmatrix}
=
\begin{pmatrix}0&x_Z^\dagger N\\0&0\end{pmatrix}
\in V_{01}.
\]
Also,
\[
\begin{pmatrix}\Tr((I_{H_\perp}-N)Z)&0\\0&0\end{pmatrix}\in V_{00}.
\]
Subtracting these three elements from $\Phi(\rho_Z)\in\im(\Phi)$ shows that
\[
\begin{pmatrix}0&0\\0&T(Z)\end{pmatrix}\in \im(\Phi),
\]
and therefore $V_{11}\subseteq\im(\Phi)$.

Now $\dim V_{01}=\dim V_{10}=\dim W$. Choose an index $i_0$ with $\alpha_{i_0}>0$. For every $X\in L(U_{i_0})$,
\[
\langle X,T(X)\rangle_{\mathrm{HS}}
=
\sum_i \alpha_i^2 \norm{P_iXP_i}_\mathrm{F}^2
\ge \alpha_{i_0}^2 \norm{P_{i_0}XP_{i_0}}_\mathrm{F}^2
=\alpha_{i_0}^2\norm{X}_\mathrm{F}^2.
\]
Hence $T$ is injective on $L(U_{i_0})$, so
\[
\dim V_{11}=\rank(T)\ge \dim L(U_{i_0})=(m-1)^2.
\]
Substituting into Eq.~\eqref{eq:image-dim-lower-raw},
\[
\dim\im(\Phi)
\ge
1+2\dim W+(m-1)^2.
\]
If the $P_i$ are not all equal, then by Eq.~\eqref{eq:dimW-ge-m}
\[
\dim\im(\Phi)
\ge 1+2m+(m-1)^2
= m^2+2,
\]
contradicting Eq.~\eqref{eq:image-dim-upper}. Therefore all nonzero $P_i$ coincide.
\end{proof}

\begin{proof}[Proof of Proposition~\ref{prop:classification-reset}]
By Lemmas~\ref{lem:equality-family} and \ref{lem:all-projections-coincide}, there exists a single rank-$(m-1)$ orthogonal projection $P_\perp$ on $H_\perp$ such that
\[
B_i=\alpha_i P_\perp
\qquad
\text{for all }i.
\]
Set
\[
P:=\proj{0}+P_\perp.
\]
Then $P$ is a rank-$m$ orthogonal projection with $\ket{0}\in\im(P)$.

Choose a unitary matrix $U=(u_{ji})$ on the Kraus index space such that the new coefficients
\[
\alpha'_j:=\sum_i u_{ji}\alpha_i
\]
satisfy $(\alpha'_1,\alpha'_2,\dots,\alpha'_r)=(1,0,\dots,0)$. Define the mixed Kraus family
\[
A'_j:=\sum_i u_{ji}A_i.
\]
The channel is unchanged. Because $B_i=\alpha_i P_\perp$, the lower-right blocks transform as
\[
B'_j=\sum_i u_{ji}B_i=
\left(\sum_i u_{ji}\alpha_i\right)P_\perp
=\alpha'_j P_\perp.
\]
Hence
\[
B'_1=P_\perp,
\qquad
B'_j=0\quad (j\ge 2).
\]
Also,
\[
v'_1=\sum_i u_{1i}v_i=\sum_i \alpha_i v_i=0
\]
by Eq.~\eqref{eq:equality-family-identities}. Therefore
\[
A'_1=P,
\qquad
A'_j=\ketbra{0}{v'_j}\quad (j\ge 2).
\]
Using again Eq.~\eqref{eq:equality-family-identities},
\[
\sum_{j\ge 2} v'_j {v'_j}^\dagger = I_{H_\perp}-P_\perp = I_d-P
\]
on $H_\perp$. Thus
\begin{align*}
\Phi(\rho)
&=
P\rho P + \sum_{j\ge 2} \ketbra{0}{v'_j}\,\rho\,\ketbra{v'_j}{0} \\
&=
P\rho P + \Tr\bigl((I_d-P)\rho\bigr)\proj{0},
\end{align*}
which is Eq.~\eqref{eq:reset-channel}.
\end{proof}

\subsection{Factorizations of reset channels}\label{subsec:reset-factorizations}

\begin{lemma}[Kraus operators of the identity channel]\label{lem:identity-kraus}
Let $C$ be finite dimensional, and let $\{K_t\}_t$ be Kraus operators of the identity channel on $L(C)$:
\[
\sum_t K_t X K_t^\dagger = X
\qquad
\textnormal{for all }X\in L(C).
\]
Then every $K_t$ is a scalar multiple of the identity on $C$.
\end{lemma}

\begin{proof}
Fix a unit vector $\ket{\psi}\in C$ and apply the channel identity to $X=\proj{\psi}$:
\[
\proj{\psi}=
\sum_t K_t\proj{\psi}K_t^\dagger.
\]
The right-hand side is a sum of positive semidefinite rank-one operators, while the left-hand side has one-dimensional support $\Span\{\ket{\psi}\}$. Therefore every vector $K_t\ket{\psi}$ belongs to $\Span\{\ket{\psi}\}$. As this holds for every unit vector, each $K_t$ is a scalar multiple of the identity.
\end{proof}

\begin{proposition}[Factorizations of a reset channel]\label{prop:reset-factorization}
Let $P$ be a rank-$m$ orthogonal projection on $\C^d$ with $\ket{0}\in\im(P)$, and define
\[
\Phi_P(\rho):=P\rho P+\Tr\bigl((I_d-P)\rho\bigr)\proj{0}.
\]
Assume
\[
\Phi_P=\Dec\circ\Enc,
\qquad
\Enc:L(\C^d)\to L(\C^m),
\qquad
\Dec:L(\C^m)\to L(\C^d)
\]
are CPTP. Then:
\begin{enumerate}[label=(\roman*),leftmargin=2.2em]
\item there exists a unitary $U:\im(P)\to \C^m$ such that
\[
\Enc(X)=UXU^\dagger
\qquad
\textnormal{for all }X\in L(\im(P)).
\]
\item the decoder is the inverse isometry channel,
\[
\Dec(Y)=U^\dagger YU
\qquad
\textnormal{for all }Y\in L(\C^m).
\]
\item the encoder has Kraus rank
\[
\krank(\Enc)=d-m+1.
\]
\end{enumerate}
\end{proposition}

\begin{proof}
Let $C:=\im(P)\subset \C^d$. For every $X\in L(C)$ we have $PXP=X$ and $(I_d-P)X=0$, so $\Phi_P(X)=X$. Thus $\Dec\circ\Enc$ acts as the identity channel on the full matrix algebra $L(C)$.

Choose Kraus operators $\{F_a\}_a$ for $\Enc$ and $\{G_\mu\}_\mu$ for $\Dec$. Then $\{G_\mu F_a P\}_{\mu,a}$ is a Kraus family for the identity channel on $L(C)$. By Lemma~\ref{lem:identity-kraus}, for every pair $(\mu,a)$ there is a scalar $c_{\mu a}$ such that
\begin{equation}\label{eq:scalar-kraus}
G_\mu F_a P = c_{\mu a} P.
\end{equation}
Using $\sum_\mu G_\mu^\dagger G_\mu = I_m$, we obtain
\begin{equation}\label{eq:lambda-matrix}
P F_a^\dagger F_b P
=
\sum_\mu P F_a^\dagger G_\mu^\dagger G_\mu F_b P
=
\left(\sum_\mu c_{\mu a}^* c_{\mu b}\right)P.
\end{equation}
So there is a positive semidefinite matrix $\Lambda=(\lambda_{ab})$ such that $P F_a^\dagger F_b P = \lambda_{ab}P$. Diagonalize $\Lambda$ by a unitary change of Kraus basis. After this change we may assume
\[
P\widetilde F_j^\dagger \widetilde F_\ell P = \lambda_j\delta_{j\ell} P,
\qquad
\lambda_j\ge 0,
\qquad
\sum_j \lambda_j=1.
\]
If $\lambda_j>0$, define
\[
U_j:=\lambda_j^{-1/2}\widetilde F_j P:C\to \C^m.
\]
Then $U_j^\dagger U_j=P$, so $U_j$ is an isometry from the $m$-dimensional space $C$ into the $m$-dimensional space $\C^m$, hence a unitary. If two coefficients $\lambda_j$ and $\lambda_\ell$ were positive with $j\neq \ell$, then
\[
U_j^\dagger U_\ell = (\lambda_j\lambda_\ell)^{-1/2}P\widetilde F_j^\dagger \widetilde F_\ell P = 0,
\]
which is impossible because the product of two unitaries cannot be the zero operator. Therefore exactly one $\lambda_j$ is nonzero. Since their sum is $1$, it must equal $1$. Renaming that Kraus operator as $F_1$, we conclude
\[
F_1P=U,
\qquad
F_jP=0\quad (j\ge 2),
\]
for some unitary $U:C\to \C^m$. This proves part (i).

Now let $Y\in L(\C^m)$ be arbitrary. Because $U$ is unitary, there exists $X\in L(C)$ with $Y=UXU^\dagger = \Enc(X)$. Hence
\[
\Dec(Y)=\Dec(\Enc(X))=\Phi_P(X)=X=U^\dagger YU,
\]
which proves part (ii).

Finally, part (ii) implies that for every $\rho\in L(\C^d)$,
\[
\Phi_P(\rho)=\Dec(\Enc(\rho))=U^\dagger \Enc(\rho)U,
\qquad
\text{so}
\qquad
\Enc(\rho)=U\Phi_P(\rho)U^\dagger.
\]
Thus $\Enc$ and $\Phi_P$ have the same Kraus rank. To compute the latter, choose an orthonormal basis $\{e_1,\dots,e_{d-m}\}$ of $(\im P)^\perp$. Then
\[
\Phi_P(\rho)=P\rho P+\sum_{j=1}^{d-m} \ketbra{0}{e_j}\,\rho\,\ketbra{e_j}{0}.
\]
So $\Phi_P$ has the Kraus family
\[
P,
\qquad
\ketbra{0}{e_1},\dots,\ketbra{0}{e_{d-m}},
\]
which is Hilbert--Schmidt orthogonal and hence linearly independent. Therefore $\krank(\Phi_P)=d-m+1$, and the same is true for $\Enc$.
\end{proof}

\subsection{Alternative proof that the decoder in Proposition~\ref{prop:reset-factorization} is isometric}\label{subsec:alternative-isometric-proof}

We provide an alternative proof of the isometric-decoder part of Proposition~\ref{prop:reset-factorization}. This proof is not needed elsewhere, but it is conceptually useful because it isolates the decoder from the Kraus-rank computation.

Let $\ket{0_m}$ denote the first standard basis vector of $\C^m$. Because $\ket{0}\in\im(P)$ and $\rank(P)=m$, choose an orthonormal basis $\ket{u_1},\dots,\ket{u_m}$ of $\im(P)$ with $\ket{u_1}=\ket{0}$, and define an isometry $V:\C^m\to \C^d$ by
\[
V\ket{j-1}=\ket{u_j}
\qquad
(1\le j\le m).
\]
Then $VV^\dagger=P$, $V^\dagger V=I_m$, and $V\ket{0_m}=\ket{0}$. Define the embedding channel
\[
\mathcal V(X):=VXV^\dagger,
\qquad
X\in L(\C^m),
\]
and the channel $R:L(\C^d)\to L(\C^m)$ by
\[
R(Y):=V^\dagger YV+\Tr\bigl((I_d-P)Y\bigr)\proj{0_m}.
\]
Choose an orthonormal basis $\ket{f_1},\dots,\ket{f_{d-m}}$ of $\im(I_d-P)$. Then $R$ has Kraus operators
\[
K_0:=V^\dagger,
\qquad
K_j:=\ketbra{0_m}{f_j}
\quad (1\le j\le d-m),
\]
so $R$ is CPTP. Moreover,
\[
(R\circ \mathcal V)(X)=X
\qquad
\text{for all }X\in L(\C^m).
\]
Now set
\[
\widetilde{\Enc}:=\Enc\circ \mathcal V:L(\C^m)\to L(\C^m),
\qquad
\widetilde{\Dec}:=R\circ \Dec:L(\C^m)\to L(\C^m).
\]
Both maps are CPTP, and
\[
\widetilde{\Dec}\circ \widetilde{\Enc}
=
R\circ \Dec\circ \Enc\circ \mathcal V
=
R\circ \Phi_P\circ \mathcal V
=
R\circ \mathcal V
=
\mathrm{id}_{L(\C^m)}.
\]
Thus $\widetilde{\Enc}$ is a CPTP channel with a CPTP left inverse. By the standard reversibility characterization of left-invertible channels (\cite[Theorem~18 and Remark~19]{puzzuoli2017ancilla}, \cite[Theorem~2.1]{nayak2006invertible}, and \cite[Theorem~10.1]{nielsen2010quantum}), such a channel must be of the form
\[
\widetilde{\Enc}(X)=W XW^\dagger
\]
for some unitary $W\in \mathrm{U}(m)$. On the other hand,
\[
\Dec\circ \widetilde{\Enc}
=
\Dec\circ \Enc\circ \mathcal V
=
\Phi_P\circ \mathcal V
=
\mathcal V.
\]
Therefore, for every $X\in L(\C^m)$,
\[
\Dec(X)
=
\Dec\bigl(\widetilde{\Enc}(W^\dagger XW)\bigr)
=
\mathcal V(W^\dagger XW)
=
VW^\dagger XWV^\dagger.
\]
Setting $A:=VW^\dagger$ gives an isometry $A^\dagger A=I_m$ with $AA^\dagger=P$ and
\[
\Dec(X)=AXA^\dagger.
\]
This is the isometric form asserted in Proposition~\ref{prop:reset-factorization}.

\subsection{From asymptotics to optimality for small $\eps$}\label{subsec:compactness-upgrade}

\myBlue
We now give a detailed restatement and proof of Theorem~\ref{thm:encoder-lb-informal}.

\begin{theorem}[Encoder Kraus-rank lower bound]\label{thm:encoder-lb}
Set $d:=2^n$, $m:=2^k$, and $H_\perp:=\ket{0}^{\perp}\subset \C^d$. For $\eps\in(0,1)$ define
\[
\ket{\psi_{\eps}(\varphi)}:=\sqrt{1-\eps}\ket{0}+\sqrt{\eps}\ket{\varphi},
\qquad
\norm{\varphi}=1,
\quad
\ket{\varphi}\in H_\perp,
\]
and let $\mu_{1,\eps}$ be the distribution induced by Haar-uniform $\ket{\varphi}$ on the unit sphere of $H_\perp$. Then for every $R<d-m+1$ there exists $\eps_0(R)>0$ such that, for every $0<\eps<\eps_0(R)$, no optimal pair $(\Enc,\Dec)$ maximizing $F_{\mu_{1,\eps}}(\Enc,\Dec)$ can have encoder Kraus rank at most $R$.

In particular, for all sufficiently small $\eps>0$, every optimal $(n,k,n_B,n_E)$-QAE satisfies
\[
 n_B\ge k.
\]
\end{theorem}
\color{black}

\begin{proof}%[Proof of Theorem~\ref{thm:encoder-lb}]
Fix $R<d-m+1$, and let $K_R$ be the set of all pairs $(\Enc,\Dec)$ such that
\begin{itemize}[leftmargin=1.8em]
\item $\Enc:L(\C^d)\to L(\C^m)$ is CPTP and $\krank(\Enc)\le R$.
\item $\Dec:L(\C^m)\to L(\C^d)$ is CPTP.
\end{itemize}
This set is compact: finite-dimensional CPTP sets are compact, and the rank condition is closed because it is the condition that the Choi matrix have rank at most $R$.

For $(\Enc,\Dec)\in K_R$, write $\Phi=\Dec\circ\Enc$ and view $F_\eps(\Enc,\Dec):=F_\eps(\Phi)$. By Lemma~\ref{lem:quadratic-eps}, there are continuous coefficient functions $A_R,B_R,C_R$ on $K_R$ such that
\[
F_\eps(\Enc,\Dec)=A_R(\Enc,\Dec)+\eps B_R(\Enc,\Dec)+\eps^2 C_R(\Enc,\Dec).
\]
Moreover,
\[
A_R(\Enc,\Dec)=\bra{0}\Phi(\proj{0})\ket{0}\le 1.
\]
If $K_R=\varnothing$, there is nothing to prove. Otherwise, pairing any
admissible encoder with the replacement decoder
$\Dec_{\mathrm{rep}}(Y):=\Tr(Y)\proj{0}$ gives a pair in $K_R$ with
$A_R=1$. Hence
\[
A_{\max}:=
\max_{(\Enc,\Dec)\in K_R}A_R(\Enc,\Dec)
=1,
\]
and the maximizing set
\[
S_R:=\{(\Enc,\Dec)\in K_R: A_R(\Enc,\Dec)=1\}
\]
is nonempty and compact.

Now fix $(\Enc,\Dec)\in S_R$. Then $\Phi(\proj{0})=\proj{0}$, so Lemma~\ref{lem:universal-first-order} applies and gives
\[
-B_R(\Enc,\Dec)=c(\Phi)\ge c_*.
\]
Suppose, for contradiction, that
\[
\sup_{(\Enc,\Dec)\in S_R} B_R(\Enc,\Dec) = -c_*.
\]
Because $S_R$ is compact and $B_R$ is continuous, some pair $(\Enc_*,\Dec_*)\in S_R$ attains this supremum. Then $c(\Dec_*\circ \Enc_*)=c_*$. By Proposition~\ref{prop:classification-reset}, the overall channel must be a reset channel, and then Proposition~\ref{prop:reset-factorization} implies $\krank(\Enc_*)=d-m+1$, contradicting the assumption $\krank(\Enc_*)\le R<d-m+1$.

Therefore there exists $\delta_R>0$ such that
\begin{equation}\label{eq:deltaR}
\sup_{(\Enc,\Dec)\in S_R} B_R(\Enc,\Dec)\le -c_* - \delta_R.
\end{equation}
Applying Lemma~\ref{lem:compact-first-order} to the compact set $K_R$ and the functions $A_R,B_R,C_R$ gives
\[
\limsup_{\eps\to 0^+} \frac{\max_{(\Enc,\Dec)\in K_R} F_\eps(\Enc,\Dec)-1}{\eps}
\le -c_* - \delta_R.
\]
Consequently, there exists $\eps_1>0$ such that for all $0<\eps<\eps_1$,
\begin{equation}\label{eq:KR-upper}
\max_{(\Enc,\Dec)\in K_R} F_\eps(\Enc,\Dec)
\le 1-\left(c_*+\frac{\delta_R}{2}\right)\eps.
\end{equation}

On the other hand, fix any rank-$m$ projection $P$ with $\ket{0}\in\im(P)$, choose an isometry $V:\C^m\to \C^d$ with $VV^\dagger=P$ and $V\ket{0_m}=\ket{0}$, and define
\[
\Enc_P(\rho):=V^\dagger P\rho P V + \Tr\bigl((I_d-P)\rho\bigr)\proj{0_m},
\qquad
\Dec_P(Y):=VYV^\dagger.
\]
Then $\Dec_P\circ \Enc_P = \Phi_P$, and one checks directly from Eqs.~\eqref{eq:first-order-expansion} and \eqref{eq:cstar} that $c(\Phi_P)=c_*$. Hence
\[
F_\eps(\Phi_P)=1-c_*\eps + O(\eps^2),
\]
so there exists $\eps_2>0$ such that for all $0<\eps<\eps_2$,
\begin{equation}\label{eq:PhiP-lower}
F_\eps(\Phi_P)
\ge 1-\left(c_*+\frac{\delta_R}{4}\right)\eps.
\end{equation}
Comparing Eqs.~\eqref{eq:KR-upper} and \eqref{eq:PhiP-lower}, we conclude that for all
\[
0<\eps<\eps_0(R):=\min\{\eps_1,\eps_2\},
\]
no pair in $K_R$ can be globally optimal. This proves the Kraus-rank statement.

Finally, if an $(n,k,n_B,n_E)$-QAE has encoder ancilla count $n_B$, then its encoder is implemented by a unitary on $n+n_B$ qubits followed by tracing out $n+n_B-k$ qubits, so its encoder Kraus rank is at most $2^{n+n_B-k}$. If $n_B<k$, then
\[
2^{n+n_B-k}\le 2^{n-1}<2^n-2^k+1=d-m+1,
\]
so such a QAE cannot be optimal for $\mu_{1,\eps}$ once $\eps$ is sufficiently small. Therefore every optimal $(n,k,n_B,n_E)$-QAE must satisfy $n_B\ge k$.
\end{proof}

\section{A counterexample to sufficiency of isometric decoders}\label{sec:counterexample-appendix}

\myBlue
We recall and prove Proposition~\ref{prop:counterexample}.
\DecoderCounterExample*
\color{black}

We prove Proposition~\ref{prop:counterexample}. The source ensemble is the one-parameter family $\{\ket{\Psi_\phi}\}_\phi$ with $\phi\in[0,2\pi)$ uniformly distributed. We consider blind single-copy compression from $2$ qubits to $1$ qubit and back, with arbitrary CPTP encoder
\[
\Enc:L((\C^2)^{\otimes 2})\to L(\C^2)
\]
and arbitrary CPTP decoder
\[
\Dec:L(\C^2)\to L((\C^2)^{\otimes 2}).
\]
The performance functional is
\begin{equation}\label{eq:counterexample-objective}
F(\Enc,\Dec)
:=
\frac{1}{2\pi}\int_0^{2\pi} \bra{\Psi_\phi}\Dec(\Enc(\proj{\Psi_\phi}))\ket{\Psi_\phi}\,d\phi.
\end{equation}

\subsection{A Kraus-rank-$2$ decoder attaining $3/4$}\label{subsec:rank2-achiever}

Define the symmetric Bell state
\[
\ket{\Sigma}:=\frac{\ket{01}+\ket{10}}{\sqrt 2}.
\]
Then
\begin{equation}\label{eq:Psi-expansion}
\ket{\Psi_\phi} = \frac12\ket{00}+\frac{e^{i\phi}}{\sqrt 2}\ket{\Sigma}+\frac{e^{2i\phi}}{2}\ket{11}.
\end{equation}
Take the encoder to be the partial trace over the second qubit,
\[
\Enc(\rho)=\Tr_2(\rho),
\]
so that $\Enc(\proj{\Psi_\phi})=\proj{\psi_\phi}$. Define a decoder with two Kraus operators $A_0,A_1:\C^2\to (\C^2)^{\otimes 2}$,
\begin{equation}\label{eq:rank2-kraus}
A_0 = \frac{1}{\sqrt 3}\ketbra{00}{0}+\sqrt{\frac23}\ketbra{\Sigma}{1},
\qquad
A_1 = \sqrt{\frac23}\ketbra{\Sigma}{0}+\frac{1}{\sqrt 3}\ketbra{11}{1}.
\end{equation}
Since $A_0^\dagger A_0 + A_1^\dagger A_1 = I_2$, this is a CPTP decoder.

For every $\phi$,
\[
A_0\ket{\psi_\phi} = \frac{1}{\sqrt 2}\left(\frac{1}{\sqrt 3}\ket{00}+e^{i\phi}\sqrt{\frac23}\ket{\Sigma}\right),
\qquad
A_1\ket{\psi_\phi} = \frac{1}{\sqrt 2}\left(\sqrt{\frac23}\ket{\Sigma}+e^{i\phi}\frac{1}{\sqrt 3}\ket{11}\right).
\]
Using Eq.~\eqref{eq:Psi-expansion}, a direct calculation gives
\[
\abs{\braket{\Psi_\phi}{A_0\psi_\phi}}^2=\frac38,
\qquad
\abs{\braket{\Psi_\phi}{A_1\psi_\phi}}^2=\frac38.
\]
Therefore
\begin{equation}\label{eq:rank2-fidelity-3over4}
\bra{\Psi_\phi}\Dec(\Enc(\proj{\Psi_\phi}))\ket{\Psi_\phi}=\frac34
\qquad
\text{for every }\phi.
\end{equation}
In particular,
\begin{equation}\label{eq:general-sup-lower-bound}
\sup_{\Enc,\Dec} F(\Enc,\Dec)\ge \frac34.
\end{equation}

\subsection{A universal upper bound for isometric decoders}\label{subsec:isometric-upper-bound}

Now restrict the decoder to Kraus rank $1$, so
\[
\Dec_V(\rho)=V\rho V^\dagger,
\qquad
V^\dagger V=I_2,
\]
for some isometry $V:\C^2\to (\C^2)^{\otimes 2}$. Let
\[
P:=VV^\dagger,
\]
so $P$ is a rank-$2$ orthogonal projector onto the image of $V$. Fix an arbitrary encoder $\Enc$ and set
\[
\rho_\phi:=V\Enc(\proj{\Psi_\phi})V^\dagger.
\]
Because $\Enc(\proj{\Psi_\phi})$ is a density matrix on a two-dimensional space,
\[
0\le \Enc(\proj{\Psi_\phi})\le I_2.
\]
Multiplying by $V$ and $V^\dagger$ gives
\begin{equation}\label{eq:rho-phi-between-0-and-P}
0\le \rho_\phi\le P.
\end{equation}
Hence
\[
\bra{\Psi_\phi}\rho_\phi\ket{\Psi_\phi}\le \bra{\Psi_\phi}P\ket{\Psi_\phi}.
\]
Averaging over $\phi$ yields
\begin{equation}\label{eq:isometric-upper-start}
F(\Enc,\Dec_V)\le \Tr(P\rho),
\qquad
\rho:=\frac{1}{2\pi}\int_0^{2\pi} \proj{\Psi_\phi}\,d\phi.
\end{equation}
From Eq.~\eqref{eq:Psi-expansion}, averaging kills the oscillatory cross-terms and gives
\begin{equation}\label{eq:rho-average}
\rho = \frac14\proj{00}+\frac12\proj{\Sigma}+\frac14\proj{11}.
\end{equation}
Thus $\rho$ has eigenvalues $1/2,1/4,1/4,0$, with eigenspaces spanned by $\ket{\Sigma}$, by $\Span\{\ket{00},\ket{11}\}$, and by the antisymmetric Bell state
\[
\ket{A}:=\frac{\ket{01}-\ket{10}}{\sqrt 2}.
\]
Therefore the maximal value of $\Tr(P\rho)$ over rank-$2$ projectors $P$ is the sum of the two largest eigenvalues:
\begin{equation}\label{eq:isometric-upper}
\max_{\rank(P)=2}\Tr(P\rho)=\frac12+\frac14=\frac34.
\end{equation}
So every isometric decoder satisfies
\begin{equation}\label{eq:isometric-decoder-upper}
F(\Enc,\Dec_V)\le \frac34.
\end{equation}
Moreover, every maximizing projector has the form
\begin{equation}\label{eq:Ptg}
P_{t,\gamma}=\proj{\Sigma}+\proj{\eta_{t,\gamma}},
\qquad
\ket{\eta_{t,\gamma}}:=\cos t\ket{00}+e^{i\gamma}\sin t\ket{11},
\quad
0\le t\le \frac{\pi}{2},
\quad
\gamma\in[0,2\pi).
\end{equation}

\subsection{Implications of equality in the upper bound}\label{subsec:equality-reduction}

Assume, for contradiction, that some encoder $\Enc$ and isometry $V$ satisfy
\[
F(\Enc,\Dec_V)=\frac34.
\]
Then equality must hold in both Eqs.~\eqref{eq:isometric-upper-start} and \eqref{eq:isometric-upper}. Consequently:
\begin{enumerate}[label=(\alph*),leftmargin=2.2em]
\item the projector $P=VV^\dagger$ must be one of the maximizers in Eq.~\eqref{eq:Ptg}, say $P=P_{t,\gamma}$.
\item for every $\phi$, the pointwise inequality
\[
\bra{\Psi_\phi}\rho_\phi\ket{\Psi_\phi}
\le
\bra{\Psi_\phi}P_{t,\gamma}\ket{\Psi_\phi}
\]
must actually be an equality.
\end{enumerate}
The second claim follows because the integrand is continuous and nonnegative. If its integral is $0$, then it vanishes identically.

Let
\[
s_\phi:=\bra{\Psi_\phi}P_{t,\gamma}\ket{\Psi_\phi}=\norm{P_{t,\gamma}\ket{\Psi_\phi}}^2.
\]
Since every $P_{t,\gamma}$ contains $\ket{\Sigma}$ and Eq.~\eqref{eq:Psi-expansion} has a nonzero $\ket{\Sigma}$ component, $s_\phi>0$ for all $\phi$. Define the normalized projections
\begin{equation}\label{eq:chi-def}
\ket{\chi^{(t,\gamma)}_\phi}:=\frac{P_{t,\gamma}\ket{\Psi_\phi}}{\sqrt{s_\phi}}.
\end{equation}
Because $\rho_\phi$ is supported on $\im(P_{t,\gamma})$,
\[
\bra{\Psi_\phi}\rho_\phi\ket{\Psi_\phi}
=
\bra{\Psi_\phi}P_{t,\gamma}\rho_\phi P_{t,\gamma}\ket{\Psi_\phi}
=
s_\phi\,\bra{\chi^{(t,\gamma)}_\phi}\rho_\phi\ket{\chi^{(t,\gamma)}_\phi}.
\]
Using Eq.~\eqref{eq:rho-phi-between-0-and-P}, the restriction of $\rho_\phi$ to $\im(P_{t,\gamma})$ is bounded above by the identity on that two-dimensional subspace. Hence
\[
\bra{\chi^{(t,\gamma)}_\phi}\rho_\phi\ket{\chi^{(t,\gamma)}_\phi}\le 1.
\]
Equality of the fidelities therefore forces
\[
\rho_\phi=\proj{\chi^{(t,\gamma)}_\phi}
\qquad
\text{for every }\phi.
\]
Thus equality in Eq.~\eqref{eq:isometric-decoder-upper} would imply the existence of a CPTP map
\[
\Phi:=\Dec_V\circ \Enc
\]
that sends each input pure state $\proj{\Psi_\phi}$ to the pure state $\proj{\chi^{(t,\gamma)}_\phi}$. We next show that no such channel exists.

\subsection{A pure-state transformation criterion}\label{subsec:pure-transform-criterion}

\begin{theorem}[\citeMy{chefles2004existence}, Theorem~2]\label{thm:exact-transform}
Let $\{\ket{\alpha_j}\}_{j=1}^m\subset H_A$ and $\{\ket{\beta_j}\}_{j=1}^m\subset H_B$ be finite families of unit vectors. There exists a CPTP map $\Phi:L(H_A)\to L(H_B)$ satisfying
\[
\Phi(\proj{\alpha_j})=\proj{\beta_j}
\qquad
(j=1,\dots,m)
\]
if and only if there exists a positive semidefinite matrix $M\in M_m(\C)$ with diagonal entries $M_{jj}=1$ such that
\begin{equation}\label{eq:gram-condition}
G_\alpha = M\circ G_\beta,
\end{equation}
where $G_\alpha$ and $G_\beta$ are the Gram matrices
\[
(G_\alpha)_{ij}=\braket{\alpha_i}{\alpha_j},
\qquad
(G_\beta)_{ij}=\braket{\beta_i}{\beta_j},
\]
and $\circ$ denotes Hadamard (entrywise) product.
\end{theorem}

\subsection{Reduction of $\gamma$ to $0$}\label{subsec:gamma-to-zero}

\begin{lemma}\label{lem:gamma-to-zero}
Let
\[
U_\beta:=\proj{0}+e^{i\beta}\proj{1},
\qquad
W_\beta:=U_\beta^{\otimes 2}.
\]
Then:
\begin{enumerate}[label=(\roman*),leftmargin=2.2em]
\item $W_\beta\ket{\Psi_\phi}=\ket{\Psi_{\phi+\beta}}$ for every $\phi,\beta$.
\item $W_{\gamma/2} P_{t,0} W_{\gamma/2}^\dagger = P_{t,\gamma}$ for every $t,\gamma$.
\item therefore
\begin{equation}\label{eq:chi-gamma-shift}
W_{\gamma/2}\ket{\chi^{(t,0)}_\phi} = \ket{\chi^{(t,\gamma)}_{\phi+\gamma/2}}
\qquad
\textnormal{for every }\phi.
\end{equation}
\end{enumerate}
Consequently, there exists a channel sending every $\proj{\Psi_\phi}$ to $\proj{\chi^{(t,\gamma)}_\phi}$ if and only if there exists a channel sending every $\proj{\Psi_\phi}$ to $\proj{\chi^{(t,0)}_\phi}$.
\end{lemma}

\begin{proof}
Part (i) is immediate from
\[
U_\beta\ket{\psi_\phi}=\frac{\ket{0}+e^{i(\phi+\beta)}\ket{1}}{\sqrt2}=\ket{\psi_{\phi+\beta}}.
\]
For part (ii), recall that
\[
P_{t,0}=\proj{\Sigma}+\proj{\eta_{t,0}},
\qquad
\ket{\eta_{t,0}}=\cos t\ket{00}+\sin t\ket{11}.
\]
Now
\[
W_{\gamma/2}\ket{\Sigma}=e^{i\gamma/2}\ket{\Sigma},
\qquad
W_{\gamma/2}\ket{\eta_{t,0}}=\cos t\ket{00}+e^{i\gamma}\sin t\ket{11}=\ket{\eta_{t,\gamma}},
\]
which proves part (ii). For part (iii), apply part (ii) and then part (i):
\[
W_{\gamma/2}P_{t,0}\ket{\Psi_\phi}
=
P_{t,\gamma}W_{\gamma/2}\ket{\Psi_\phi}
=
P_{t,\gamma}\ket{\Psi_{\phi+\gamma/2}}.
\]
Since $W_{\gamma/2}$ is unitary, the norms on both sides agree, proving Eq.~\eqref{eq:chi-gamma-shift}.

If a channel $\Phi_\gamma$ sends $\proj{\Psi_\phi}$ to $\proj{\chi^{(t,\gamma)}_\phi}$ for all $\phi$, then
\[
\Phi_0:=\Ad_{W_{\gamma/2}^\dagger}\circ \Phi_\gamma \circ \Ad_{W_{\gamma/2}}
\]
sends $\proj{\Psi_\phi}$ to $\proj{\chi^{(t,0)}_\phi}$ by Eq.~\eqref{eq:chi-gamma-shift}. The converse is obtained by reversing the conjugation.
\end{proof}

By Lemma~\ref{lem:gamma-to-zero}, it is enough to treat the case $\gamma=0$. Write
\[
P_t:=P_{t,0},
\qquad
\ket{\eta_t}:=\ket{\eta_{t,0}}=\cos t\ket{00}+\sin t\ket{11}.
\]

\subsection{The case $0\le \sin 2t<1$}\label{subsec:case-sin-less-than-1}

Set
\[
s:=\sin 2t,
\qquad
c:=\cos 2t.
\]
For the projector $P_t=\proj{\Sigma}+\proj{\eta_t}$, define the normalized projected outputs
\[
\ket{\chi^{(t)}_\phi}:=\frac{P_t\ket{\Psi_\phi}}{\norm{P_t\ket{\Psi_\phi}}}.
\]
Using Eq.~\eqref{eq:Psi-expansion},
\[
P_t\ket{\Psi_\phi}
=
\frac{\cos t+e^{2i\phi}\sin t}{2}\ket{\eta_t}+\frac{e^{i\phi}}{\sqrt2}\ket{\Sigma},
\]
so
\begin{equation}\label{eq:PtPsi-norm}
\norm{P_t\ket{\Psi_\phi}}^2 = \frac{3+s\cos(2\phi)}{4}.
\end{equation}
In particular,
\[
\ket{\chi^{(t)}_0} = \frac{(\cos t+\sin t)\ket{\eta_t}+\sqrt2\ket{\Sigma}}{\sqrt{3+s}},
\]
\[
\ket{\chi^{(t)}_{\pi/2}} = \frac{(\cos t-\sin t)\ket{\eta_t}+i\sqrt2\ket{\Sigma}}{\sqrt{3-s}},
\]
\[
\ket{\chi^{(t)}_{\pi}} = \frac{(\cos t+\sin t)\ket{\eta_t}-\sqrt2\ket{\Sigma}}{\sqrt{3+s}},
\]
\[
\ket{\chi^{(t)}_{3\pi/2}} = \frac{(\cos t-\sin t)\ket{\eta_t}-i\sqrt2\ket{\Sigma}}{\sqrt{3-s}}.
\]
From these formulas one computes
\begin{equation}\label{eq:chi-inner-products-case1}
\braket{\chi^{(t)}_0}{\chi^{(t)}_{\pi/2}} = \frac{c+2i}{\sqrt{9-s^2}},
\qquad
\braket{\chi^{(t)}_0}{\chi^{(t)}_{\pi}} = \frac{s-1}{s+3},
\qquad
\braket{\chi^{(t)}_0}{\chi^{(t)}_{3\pi/2}} = \frac{c-2i}{\sqrt{9-s^2}}.
\end{equation}
Now choose the four phases
\[
\phi_1=0,
\qquad
\phi_2=\frac{\pi}{2},
\qquad
\phi_3=\pi,
\qquad
\phi_4=\frac{3\pi}{2}.
\]
For the inputs,
\begin{equation}\label{eq:input-inner-products-case1}
\braket{\Psi_0}{\Psi_{\pi/2}}=\frac{i}{2},
\qquad
\braket{\Psi_0}{\Psi_{\pi}}=0,
\qquad
\braket{\Psi_0}{\Psi_{3\pi/2}}=-\frac{i}{2},
\end{equation}
and cyclic symmetry determines the rest of the $4\times 4$ Gram matrix. If a channel mapped the four input states $\proj{\Psi_{\phi_j}}$ to the four pure output states $\proj{\chi^{(t)}_{\phi_j}}$, then Theorem~\ref{thm:exact-transform} would force a positive semidefinite matrix $M_t$ with diagonal entries $1$ such that
\[
G_\Psi = M_t\circ G_{\chi^{(t)}}.
\]
For these four phases all off-diagonal entries are uniquely determined, and one obtains
\begin{equation}\label{eq:Mt-matrix}
M_t=
\begin{pmatrix}
1 & h & 0 & \overline h \\
\overline h & 1 & h & 0 \\
0 & \overline h & 1 & h \\
h & 0 & \overline h & 1
\end{pmatrix},
\qquad
h
=
\frac{\tfrac{i}{2}}{\tfrac{c+2i}{\sqrt{9-s^2}}}
=
\frac{\sqrt{9-s^2}(2+ic)}{2(5-s^2)}.
\end{equation}
Now apply $M_t$ to the vector
\[
v:=(1,-1,1,-1)^\top.
\]
A direct multiplication gives
\[
M_t v = (1-h-\overline h)v = \left(1-\frac{2\sqrt{9-s^2}}{5-s^2}\right)v.
\]
For $0\le s<1$ we have
\[
2\sqrt{9-s^2}\ge 2\sqrt8 > 5 \ge 5-s^2,
\]
so
\[
1-\frac{2\sqrt{9-s^2}}{5-s^2}<0.
\]
Hence $M_t$ has a negative eigenvalue and cannot be positive semidefinite, contradicting Theorem~\ref{thm:exact-transform}. We have proved the following.

\begin{proposition}\label{prop:case1-impossible}
If $0\le \sin 2t<1$, then no CPTP map sends every $\proj{\Psi_\phi}$ to $\proj{\chi^{(t,0)}_\phi}$.
\end{proposition}

\subsection{The case $\sin 2t=1$}\label{subsec:case-sin-equals-1}

Now suppose $\sin 2t=1$, i.e. $t=\pi/4$. Then
\[
\ket{\eta_{\pi/4}} = \frac{\ket{00}+\ket{11}}{\sqrt2}.
\]
A short calculation shows that
\[
P_{\pi/4}\ket{\Psi_\phi} = \frac{e^{i\phi}}{\sqrt2}\bigl(\cos\phi\,\ket{\eta_{\pi/4}}+\ket{\Sigma}\bigr),
\]
so the normalized projected outputs are
\begin{equation}\label{eq:chi-pi-over-4}
\ket{\chi^{(\pi/4)}_\phi} = e^{i\phi}\,\frac{\cos\phi\,\ket{\eta_{\pi/4}}+\ket{\Sigma}}{\sqrt{1+\cos^2\phi}}.
\end{equation}
Choose the three phases
\[
\phi_1=0,
\qquad
\phi_2=\frac{\pi}{6},
\qquad
\phi_3=\frac{5\pi}{6}.
\]
If a channel mapped the corresponding input states to the three pure outputs from Eq.~\eqref{eq:chi-pi-over-4}, then Theorem~\ref{thm:exact-transform} would again force a positive semidefinite matrix $M$ with unit diagonal satisfying
\[
G_\Psi = M\circ G_{\chi^{(\pi/4)}}.
\]
A direct computation of the forced entrywise quotients gives
\begin{equation}\label{eq:M-pi-over-4}
M=
\begin{pmatrix}
1 & \sqrt{14}/4 & \sqrt{14}/4 \\
\sqrt{14}/4 & 1 & 7/4 \\
\sqrt{14}/4 & 7/4 & 1
\end{pmatrix}.
\end{equation}
Applying $M$ to $(0,1,-1)^\top$ gives
\[
M(0,1,-1)^\top = -\frac34(0,1,-1)^\top,
\]
so $M$ has a negative eigenvalue and is not positive semidefinite. Therefore:

\begin{proposition}\label{prop:case2-impossible}
If $\sin 2t=1$, then no CPTP map sends every $\proj{\Psi_\phi}$ to $\proj{\chi^{(t,0)}_\phi}$.
\end{proposition}

\subsection{Conclusion of the counterexample}\label{subsec:counterexample-conclusion}

Combining Lemma~\ref{lem:gamma-to-zero} and Propositions~\ref{prop:case1-impossible} and \ref{prop:case2-impossible} gives the following impossibility statement.

\begin{corollary}\label{cor:isometric-upper-not-attained}
No encoder $\Enc$ and isometry $V$ can satisfy
\[
F(\Enc,\Dec_V)=\frac34.
\]
Equivalently, the upper bound in Eq.~\eqref{eq:isometric-decoder-upper} is never attained.
\end{corollary}

The set of encoders is compact in finite dimensions (for example via Choi matrices), and the set of isometries $V:\C^2\to (\C^2)^{\otimes 2}$ is compact as well. Since $(\Enc,V)\mapsto F(\Enc,\Dec_V)$ is continuous, the maximum over all encoder-isometric-decoder pairs is attained. By Corollary~\ref{cor:isometric-upper-not-attained}, that maximum is strictly smaller than $3/4$. Together with Eq.~\eqref{eq:general-sup-lower-bound}, this proves Proposition~\ref{prop:counterexample}.

\subsection{Detailed restatement and proof of Observation~\ref{obs:finite_eval}}
\label{sec:proof_finite_eval}

\myBlue
Let $\Phi_{U,V}:=\Dec_V\circ \Enc_U$,
\begin{equation}
\label{eq:finite_eval_setup}
f_{U,V}(\phi):=\bra{\Psi_\phi}\Phi_{U,V}(\proj{\Psi_\phi})\ket{\Psi_\phi},
\qquad
\tilde{F}_{\mu_{\mathrm{ph}}}(U,V):=\frac{1}{2\pi}\int_0^{2\pi} f_{U,V}(\phi)\,d\phi.
\end{equation}

\begin{observation}
\label{obs:finite_eval_formal}
Set $\phi_j:=2\pi j/5$ for $j=0,1,2,3,4$. Then
\begin{equation}
\label{eq:finite_eval_statement}
\tilde{F}_{\mu_{\mathrm{ph}}}(U,V)=\frac15\sum_{j=0}^4 f_{U,V}(\phi_j).
\end{equation}
In particular, the phase integral in Eq.~\eqref{eq:finite_eval_setup} can be replaced by five evaluations of the fidelity integrand.
\end{observation}
\color{black}

\begin{proof}
Write
\begin{equation}
\ket{\Sigma}:=\frac{\ket{01}+\ket{10}}{\sqrt{2}},
\qquad
\ket{v_0}:=\frac12\ket{00},\qquad
\ket{v_1}:=\frac{1}{\sqrt2}\ket{\Sigma},\qquad
\ket{v_2}:=\frac12\ket{11}.
\end{equation}
From Eq.~\eqref{eq:dec_counter_source} and the definition of $\ket{\Sigma}$,
\begin{equation}
\label{eq:finite_eval_phase_expansion}
\ket{\Psi_\phi}
=
\frac12\ket{00}
+\frac{e^{i\phi}}{\sqrt2}\ket{\Sigma}
+\frac{e^{2i\phi}}2\ket{11}
=
\sum_{r=0}^2 e^{ir\phi}\ket{v_r}.
\end{equation}
Therefore
\begin{equation}
\proj{\Psi_\phi}
=
\sum_{a,b=0}^2 e^{i(a-b)\phi}\ket{v_a}\!\bra{v_b}.
\end{equation}
By linearity of $\Phi_{U,V}$,
\begin{equation}
\label{eq:finite_eval_linearity}
\Phi_{U,V}(\proj{\Psi_\phi})
=
\sum_{a,b=0}^2 e^{i(a-b)\phi}\,
\Phi_{U,V}\!\left(\ket{v_a}\!\bra{v_b}\right).
\end{equation}
Also,
\begin{equation}
\label{eq:finite_eval_braket}
\bra{\Psi_\phi}=\sum_{c=0}^2 e^{-ic\phi}\bra{v_c},
\qquad
\ket{\Psi_\phi}=\sum_{d=0}^2 e^{id\phi}\ket{v_d}.
\end{equation}
Substituting Eqs.~\eqref{eq:finite_eval_linearity} and \eqref{eq:finite_eval_braket} into the definition of $f_{U,V}(\phi)$ gives
\begin{align}
f_{U,V}(\phi)
&=
\sum_{a,b,c,d=0}^2
e^{i(a-b-c+d)\phi}\,
\bra{v_c}\Phi_{U,V}\!\left(\ket{v_a}\!\bra{v_b}\right)\ket{v_d}.
\end{align}
Now $a,b,c,d\in\{0,1,2\}$, so the integer
\[
m:=a-b-c+d
\]
always lies in $\{-4,-3,\dots,4\}$. Hence $f_{U,V}$ is a trigonometric polynomial of degree at most $4$:
\[
f_{U,V}(\phi)=\sum_{m=-4}^4 c_m(U,V)e^{im\phi}
\]
for suitable coefficients $c_m(U,V)$.

Averaging over $[0,2\pi)$ keeps only the zero Fourier mode. Thus
\begin{align}
\tilde{F}_{\mu_{\mathrm{ph}}}(U,V)
&=c_0(U,V)
=
\sum_{a,b,c,d=0}^2
\delta_{a-b-c+d,\,0}\;
\bra{v_c}\Phi_{U,V}\!\left(\ket{v_a}\!\bra{v_b}\right)\ket{v_d}.
\end{align}

Now let $\omega:=e^{2\pi i/5}$ and $\phi_j:=2\pi j/5$. The root-of-unity identity gives
\[
\frac15\sum_{j=0}^4 \omega^{jm}
=
\begin{cases}
1,&m\equiv 0\pmod 5,\\[2mm]
0,&m\not\equiv 0\pmod 5.
\end{cases}
\]
Because $m\in\{-4,-3,\dots,4\}$, the congruence $m\equiv 0\pmod 5$ is equivalent here to $m=0$. Therefore
\begin{align}
\frac15\sum_{j=0}^4 f_{U,V}(\phi_j)
&=
\sum_{a,b,c,d=0}^2
\left(
\frac15\sum_{j=0}^4
e^{2\pi i j(a-b-c+d)/5}
\right)
\bra{v_c}\Phi_{U,V}\!\left(\ket{v_a}\!\bra{v_b}\right)\ket{v_d}
\nonumber\\
&=
\sum_{a,b,c,d=0}^2
\delta_{a-b-c+d,\,0}\;
\bra{v_c}\Phi_{U,V}\!\left(\ket{v_a}\!\bra{v_b}\right)\ket{v_d}
\nonumber\\
&=
\tilde{F}_{\mu_{\mathrm{ph}}}(U,V).
\end{align}
This proves Eq.~\eqref{eq:finite_eval_statement}.
\end{proof}

\begin{remark}
The proof uses only the phase expansion in Eq.~\eqref{eq:finite_eval_phase_expansion} and the linearity of $\Phi_{U,V}$. Hence the same argument works for any fixed linear map
$
\Phi:L\bigl((\mathbb C^2)^{\otimes 2}\bigr)\to L\bigl((\mathbb C^2)^{\otimes 2}\bigr).
$
\end{remark}

\subsection{Proof of Theorem~\ref{thm:decoder-k1-threshold}}
\label{subsec:k1-extension}

\myBlue
We recall and prove Theorem~\ref{thm:decoder-k1-threshold}.
\DecoderSharpKOne*
\color{black}

\begin{proof}
We first prove the lower bound on the unrestricted optimum. Let
$\Enc_{\mathrm{ph}} : L(H_4) \to L(H_2)$ and
$\Dec_{\mathrm{ph}} : L(H_2) \to L(H_4)$ be the encoder and decoder constructed in
Appendix~\ref{subsec:rank2-achiever} for the source $\mu_{\mathrm{ph}}$, so that
\[
F_{\mu_{\mathrm{ph}}}(\Enc_{\mathrm{ph}},\Dec_{\mathrm{ph}})=\frac{3}{4}.
\]
Define
\[
\Enc^{(n)} := \Enc_{\mathrm{ph}} \circ \operatorname{Tr}_R,
\qquad
\Dec^{(n)}(\rho) := \Dec_{\mathrm{ph}}(\rho) \otimes |0_R\rangle\langle 0_R|.
\]
Then $\Enc^{(n)} : L(H_{2^n}) \to L(H_2)$ and
$\Dec^{(n)} : L(H_2) \to L(H_{2^n})$ are CPTP, and for every $\phi$ we have
\[
(\Dec^{(n)} \circ \Enc^{(n)})(|\Omega_{\phi}\rangle\langle \Omega_{\phi}|)
=
(\Dec_{\mathrm{ph}} \circ \Enc_{\mathrm{ph}})(|\Psi_{\phi}\rangle\langle \Psi_{\phi}|)
\otimes |0_R\rangle\langle 0_R|.
\]
Therefore
\[
F_{\mu_{\mathrm{ph}}^{(n)}}(\Enc^{(n)},\Dec^{(n)})
=
F_{\mu_{\mathrm{ph}}}(\Enc_{\mathrm{ph}},\Dec_{\mathrm{ph}})
=
\frac{3}{4},
\]
which proves
\[
\sup_{\Enc,\Dec} F_{\mu_{\mathrm{ph}}^{(n)}}(\Enc,\Dec) \ge \frac{3}{4}.
\]

We next bound the performance of isometric decoders. Fix an encoder
$\Enc : L(H_{2^n}) \to L(H_2)$ and an isometry
$V : H_2 \to H_{2^n}$, and let $P := VV^{\dagger}$, so $P$ is a rank-two
projector on $H_{2^n}$. For each $\phi$, set
\[
\rho_{\phi} := V\,\Enc(|\Omega_{\phi}\rangle\langle \Omega_{\phi}|)V^{\dagger}.
\]
Since $\Enc(|\Omega_{\phi}\rangle\langle \Omega_{\phi}|)$ is a qubit density
operator, it is bounded above by $I_2$, and hence
\[
0 \le \rho_{\phi} \le V I_2 V^{\dagger} = P.
\]
Thus
\[
\langle \Omega_{\phi}|\rho_{\phi}|\Omega_{\phi}\rangle
\le
\langle \Omega_{\phi}|P|\Omega_{\phi}\rangle.
\]
Averaging over $\phi$ gives
\[
F_{\mu_{\mathrm{ph}}^{(n)}}(\Enc,\Dec_V)
\le
\operatorname{Tr}(P\,\overline{\rho}_{\mathrm{ph}}^{(n)}),
\qquad
\overline{\rho}_{\mathrm{ph}}^{(n)}
:=
\mathbb{E}_{\phi}[|\Omega_{\phi}\rangle\langle \Omega_{\phi}|].
\]
By the calculation already carried out in Eq.~\eqref{eq:rho-average},
\[
\overline{\rho}_{\mathrm{ph}}
:=
\mathbb{E}_{\phi}[|\Psi_{\phi}\rangle\langle \Psi_{\phi}|]
=
\frac{1}{4}|00\rangle\langle 00|
+
\frac{1}{2}|\Sigma\rangle\langle \Sigma|
+
\frac{1}{4}|11\rangle\langle 11|,
\]
so
\[
\overline{\rho}_{\mathrm{ph}}^{(n)}
=
\overline{\rho}_{\mathrm{ph}} \otimes |0_R\rangle\langle 0_R|.
\]
Hence the nonzero eigenvalues of $\overline{\rho}_{\mathrm{ph}}^{(n)}$ are
$1/2$, $1/4$, and $1/4$. Since $P$ has rank two,
\[
\operatorname{Tr}(P\,\overline{\rho}_{\mathrm{ph}}^{(n)}) \le \frac{1}{2}+\frac{1}{4}=\frac{3}{4}.
\]
Therefore
\[
\sup_{\Enc,\,V^{\dagger}V=I_2} F_{\mu_{\mathrm{ph}}^{(n)}}(\Enc,\Dec_V) \le \frac{3}{4}.
\]

To rule out equality, assume that some encoder $\Enc$ and isometry $V$ satisfy
\[
F_{\mu_{\mathrm{ph}}^{(n)}}(\Enc,\Dec_V)=\frac{3}{4}.
\]
Then the preceding inequalities are tight, so in particular
\[
\operatorname{Tr}(P\,\overline{\rho}_{\mathrm{ph}}^{(n)})=\frac{3}{4}.
\]
Because the only nonzero eigenvalues of $\overline{\rho}_{\mathrm{ph}}^{(n)}$ are
$1/2$, $1/4$, and $1/4$, equality implies that
$\operatorname{im}(P)$ is spanned by an eigenvector of eigenvalue $1/2$ together
with a unit vector from the $1/4$-eigenspace. In particular,
\[
\operatorname{im}(P) \subseteq \operatorname{supp}(\overline{\rho}_{\mathrm{ph}}^{(n)})
\subseteq
(\mathbb{C}^2)^{\otimes 2} \otimes \operatorname{span}\{|0_R\rangle\}.
\]
Hence there exists an isometry
$W : H_2 \to (\mathbb{C}^2)^{\otimes 2}$ such that
\[
V = W \otimes |0_R\rangle.
\]
Now define a two-qubit encoder
$\Enc^{\sharp} : L(H_4) \to L(H_2)$ by
\[
\Enc^{\sharp}(\rho) := \Enc(\rho \otimes |0_R\rangle\langle 0_R|).
\]
Then for every $\phi$,
\[
(\Dec_V \circ \Enc)(|\Omega_{\phi}\rangle\langle \Omega_{\phi}|)
=
(\Dec_W \circ \Enc^{\sharp})(|\Psi_{\phi}\rangle\langle \Psi_{\phi}|)
\otimes |0_R\rangle\langle 0_R|,
\]
and therefore
\[
F_{\mu_{\mathrm{ph}}}(\Enc^{\sharp},\Dec_W)
=
F_{\mu_{\mathrm{ph}}^{(n)}}(\Enc,\Dec_V)
=
\frac{3}{4}.
\]
This contradicts Corollary~\ref{cor:isometric-upper-not-attained}. Hence no encoder and isometry can attain
fidelity $3/4$ for $\mu_{\mathrm{ph}}^{(n)}$.

The sets of encoders $\Enc : L(H_{2^n}) \to L(H_2)$ and isometries
$V : H_2 \to H_{2^n}$ are compact in finite dimensions, and the map
$(\Enc,V) \mapsto F_{\mu_{\mathrm{ph}}^{(n)}}(\Enc,\Dec_V)$ is continuous. Therefore the
maximum over encoder-isometric-decoder pairs is attained. Since the value
$3/4$ is impossible, that maximum is strictly smaller than $3/4$:
\[
\sup_{\Enc,\,V^{\dagger}V=I_2} F_{\mu_{\mathrm{ph}}^{(n)}}(\Enc,\Dec_V) < \frac{3}{4}.
\]

Finally, consider an optimal $(n,1,n_B,n_E)$-QAE for the source
$\mu_{\mathrm{ph}}^{(n)}$. By Definition~\ref{def:qae} we must have $n_E \ge n-1$.
If $n_E < n$, then necessarily $n_E = n-1$. But then the discarded decoder
register $G$ has
$
1+n_E-n = 0
$
qubits, so the decoder channel is isometric. This is impossible for an
optimal QAE, because the globally achievable fidelity is at least $3/4$ while
all isometric decoders achieve fidelity strictly below $3/4$. Therefore every optimal
$(n,1,n_B,n_E)$-QAE for $\mu_{\mathrm{ph}}^{(n)}$ must satisfy $n_E \ge n$.
\end{proof}

\section{Non-isometric decoders can outperform isometric ones for a fixed suboptimal encoder}\label{sec:suboptimal-encoder-appendix}

The counterexample in Appendix~\ref{sec:counterexample-appendix} is about \emph{globally optimal} encoder--decoder pairs. A much simpler phenomenon already appears if the encoder is held fixed but suboptimal.

Let $n=2$ and $k=1$, and consider the source distribution consisting of
\[
\ket{\psi_0}=\ket{00},
\qquad
\ket{\psi_1}=\ket{11},
\]
sampled with probabilities $1-\eps$ and $\eps$, where $0<\eps<1/2$. With an optimal encoder one can achieve perfect reconstruction using, for example, the encoder whose Kraus operators are
\[
K_0=\ketbra{0}{00}+\ketbra{1}{11},
\qquad
K_1=\ketbra{0}{01}+\ketbra{1}{10},
\]
and the isometric decoder
\[
\Dec(\rho)=W\rho W^\dagger,
\qquad
W=\ketbra{00}{0}+\ketbra{11}{1}.
\]
Now instead fix the suboptimal encoder that discards the input and always outputs the maximally mixed qubit state $I_2/2$. If the decoder is allowed to be non-isometric, it can ignore its input and deterministically output $\ket{00}$, giving expected fidelity
\[
(1-\eps)\cdot 1 + \eps\cdot 0 = 1-\eps > \frac12.
\]
Suppose instead that the decoder is restricted to be an isometry $W$. Then the reconstructed output state is always
\[
W\frac{I_2}{2}W^\dagger = \frac12 P,
\qquad
P:=WW^\dagger,
\]
where $P$ is a rank-$2$ projector. For any pure state $\ket{\psi}$,
\[
\bra{\psi}\frac12 P\ket{\psi}\le \frac12.
\]
Hence the expected fidelity is at most $1/2$. So for this fixed suboptimal encoder, non-isometric decoders strictly outperform all isometric decoders. This is weaker than Proposition~\ref{prop:counterexample}, but it pinpoints the failure of a tempting proof strategy already noted in spirit by \citeMy{barnum1996general}.

\section{A source-dependent multiplicative guarantee for isometric decoders}\label{subsec:isometric-multiplicative}

Isometric decoders are not universally sufficient for optimality, but one can still prove a source-dependent approximation guarantee in the same blind single-copy setting.
Let $d:=2^n$ and $m:=2^k$.
For a source distribution $\mu$ over pure states $\ket{\psi}\in \C^d$, define
\[
F^\star(\mu)
:=
\sup_{\calE,\calD} F_\mu(\calE,\calD),
\qquad
F^\star_{\mathrm{iso}}(\mu)
:=
\sup_{\calE,\,V^\dagger V=I_m} F_\mu(\calE,\calD_V).
\]
Also let
\[
\bar\rho
:=
\mathbb{E}_{\ket{\psi}\sim\mu}[\proj{\psi}],
\qquad
s_m
:=
\sum_{i=1}^m \lambda_i(\bar\rho),
\qquad
\eta_m
:=
1-s_m,
\]
where all eigenvalues are listed in nonincreasing order and padded with zeros when necessary.
For each rank-$m$ orthogonal projector $P$ on $\C^d$, define
\[
p_P(\psi):=\bra{\psi}P\ket{\psi},
\qquad
Y_P
:=
\mathbb{E}_{\ket{\psi}\sim\mu}\Bigl[(1-p_P(\psi))P\proj{\psi}P\Bigr].
\]

\begin{definition}
For finite-dimensional Hilbert spaces $A$ and $B$, we say that a state
$\rho\in L(A\otimes B)$ is separable if it can be written as
\[
\rho=\sum_{x=1}^N p_x\,\sigma_x\otimes \tau_x,
\qquad
p_x\ge 0,
\qquad
\sum_{x=1}^N p_x=1,
\]
where each $\sigma_x$ is a density operator on $A$ and each $\tau_x$ is a density operator on $B$.
\end{definition}

Finally, define
\[
B_\mu
:=
\sup_{\substack{P=P^\dagger=P^2\\ \rank(P)=m}}
\left(
\mathbb{E}_{\ket{\psi}\sim\mu}\bigl[p_P(\psi)^2\bigr]
+
\lambda_{\max}(Y_P)
\right)
\]
and
\[
U_\mu
:=
\sup
\left\{
\sum_{i=1}^m \lambda_i(\Omega)
:
\begin{array}{l}
\Omega\in L(\C^d\otimes H_E),\ \Omega\ge 0,\ \Tr(\Omega)=1,\\[0.2em]
\Tr_E(\Omega)=\bar\rho,\ \Omega\text{ separable},\\[0.2em]
H_E\text{ finite dimensional}
\end{array}
\right\}.
\]

We now give a detailed restatement and proof of Theorem~\ref{thm:informal_multiplicative}.

\begin{theorem}
\phantomsection
\label{thm:isometric-multiplicative}
With the notation above,
\[
\frac{F^\star_{\mathrm{iso}}(\mu)}{F^\star(\mu)}
\ge
\frac{B_\mu}{U_\mu}.
\]
Moreover,
\[
U_\mu\le s_m,
\qquad
B_\mu\ge \left(1-\frac{1}{m}\right)s_m^2+\frac{s_m}{m},
\]
and therefore
\[
\frac{F^\star_{\mathrm{iso}}(\mu)}{F^\star(\mu)}
\ge
\frac{B_\mu}{s_m}
\ge
1-\left(1-\frac{1}{m}\right)\eta_m.
\]
\end{theorem}

\begin{proof}
We divide the proof into four steps.

\paragraph{Step 1: the unrestricted optimum is at most $U_\mu$.}
Fix any encoder $\calE:L(\C^d)\to L(\C^m)$ and decoder $\calD:L(\C^m)\to L(\C^d)$.
Choose a Stinespring isometry $W:\C^m\to \C^d\otimes H_E$ such that
\[
\calD(\sigma)=\Tr_E\bigl(W\sigma W^\dagger\bigr),
\]
where $H_E$ is finite dimensional, and set
\[
P:=WW^\dagger.
\]
Then $P$ is a rank-$m$ orthogonal projector on $\C^d\otimes H_E$.
For a pure input state $\proj{\psi}$ and a latent state $\sigma\in L(\C^m)$,
\[
\bra{\psi}\calD(\sigma)\ket{\psi}
=
\Tr\bigl[(\proj{\psi}\otimes I_E)W\sigma W^\dagger\bigr].
\]
As $\sigma$ ranges over all density operators on $\C^m$, the state $W\sigma W^\dagger$ ranges over all density operators supported on $\im(P)$.
Hence
\[
\sup_{\sigma\ge 0,\,\Tr\sigma=1}\bra{\psi}\calD(\sigma)\ket{\psi}
=
\lambda_{\max}\bigl(P(\proj{\psi}\otimes I_E)P\bigr).
\]
Now define
\[
R_\psi := (\bra{\psi}\otimes I_E)P.
\]
Then
\[
P(\proj{\psi}\otimes I_E)P = R_\psi^\dagger R_\psi,
\qquad
(\bra{\psi}\otimes I_E) P (\ket{\psi}\otimes I_E) = R_\psi R_\psi^\dagger.
\]
Therefore these two operators have the same nonzero eigenvalues, and so
\[
\lambda_{\max}\bigl(P(\proj{\psi}\otimes I_E)P\bigr)
=
\lambda_{\max}\bigl( (\bra{\psi}\otimes I_E) P (\ket{\psi}\otimes I_E) \bigr).
\]
Since $(\bra{\psi}\otimes I_E) P (\ket{\psi}\otimes I_E)$ is positive semidefinite on $H_E$,
\[
\lambda_{\max}\bigl( (\bra{\psi}\otimes I_E) P (\ket{\psi}\otimes I_E) \bigr)
=
\sup_{\tau\ge 0,\,\Tr\tau=1}
\Tr\bigl[P(\proj{\psi}\otimes \tau)\bigr].
\]
Consequently,
\[
\bra{\psi}(\calD\circ \calE)(\proj{\psi})\ket{\psi}
\le
\sup_{\tau\ge 0,\,\Tr\tau=1}
\Tr\bigl[P(\proj{\psi}\otimes \tau)\bigr].
\]
Averaging over $\mu$ and taking the supremum over all measurable state-valued maps $\ket{\psi}\mapsto \tau_\psi$ gives
\[
F_\mu(\calE,\calD)
\le
\sup_{\ket{\psi}\mapsto \tau_\psi}
\Tr(P\Omega),
\qquad
\Omega
:=
\mathbb{E}_{\ket{\psi}\sim\mu}[\proj{\psi}\otimes \tau_\psi].
\]
Every such $\Omega$ is separable and satisfies $\Tr_E(\Omega)=\bar\rho$.
Since $P$ has rank $m$,
\[
\Tr(P\Omega)
\le
\sum_{i=1}^m \lambda_i(\Omega)
\le
U_\mu.
\]
Thus $F_\mu(\calE,\calD)\le U_\mu$ for every CPTP pair $(\calE,\calD)$, and therefore
\[
F^\star(\mu)\le U_\mu.
\]

\paragraph{Step 2: the auxiliary upper bound $U_\mu$ is at most $s_m$.}
Fix a finite-dimensional auxiliary space $H_E$ and a separable state $\Omega\in L(\C^d\otimes H_E)$ with $\Tr_E(\Omega)=\bar\rho$.
Because the ambient spaces are finite dimensional, $\Omega$ admits a finite pure-product decomposition,
\[
\Omega
=
\sum_{x=1}^N p_x\,\proj{a_x}\otimes \proj{b_x},
\qquad
p_x\ge 0,
\quad
\sum_{x=1}^N p_x=1.
\]
Set
\[
u_x:=\sqrt{p_x}\ket{a_x}\in \C^d,
\qquad
w_x:=u_x\otimes \ket{b_x}\in \C^d\otimes H_E.
\]
Let $A$ be the $d\times N$ matrix whose $x$-th column is $u_x$, and let $B$ be the $(d\dim H_E)\times N$ matrix whose $x$-th column is $w_x$.
Then
\[
\bar\rho = AA^\dagger,
\qquad
\Omega = BB^\dagger.
\]
Define the Gram matrices
\[
H:=A^\dagger A,
\qquad
G:=B^\dagger B,
\qquad
K_{xy}:=\braket{b_x}{b_y}.
\]
The nonzero eigenvalues of $AA^\dagger$ and $A^\dagger A$ coincide, and likewise for $BB^\dagger$ and $B^\dagger B$.
Hence the nonzero eigenvalues of $\bar\rho$ are the eigenvalues of $H$, and the nonzero eigenvalues of $\Omega$ are the eigenvalues of $G$.
Moreover,
\[
G_{xy}
=
\braket{w_x}{w_y}
=
\braket{u_x}{u_y}\braket{b_x}{b_y}
=
H_{xy}K_{xy},
\]
so
\[
G=H\circ K=K\circ H,
\]
where $\circ$ denotes the entrywise product.
Because $K$ is a Gram matrix of unit vectors, it is positive semidefinite and has diagonal entries equal to $1$.
Thus the Schur multiplier
\[
\Phi_K(X):=K\circ X
\]
is completely positive by the Schur product theorem, and it is both unital and trace preserving because $K_{xx}=1$ for all $x$.
Its Hilbert--Schmidt adjoint is
\[
\Phi_K^\dagger(X)=\overline K\circ X,
\]
and $\overline K$ is again a correlation matrix, so $\Phi_K^\dagger$ is also positive, unital, and trace preserving.
Now let $X$ be any Hermitian matrix.
By the Ky Fan variational principle,
\[
\sum_{i=1}^m \lambda_i(X)
=
\max_{\substack{0\le Q\le I\\ \Tr(Q)=m}} \Tr(QX).
\]
Applying this to $G=\Phi_K(H)$ yields
\[
\begin{aligned}
\sum_{i=1}^m \lambda_i(G)
&=
\max_{\substack{0\le Q\le I\\ \Tr(Q)=m}} \Tr\bigl(Q\,\Phi_K(H)\bigr) \\
&=
\max_{\substack{0\le Q\le I\\ \Tr(Q)=m}} \Tr\bigl(\Phi_K^\dagger(Q)\,H\bigr) \\
&\le
\max_{\substack{0\le Y\le I\\ \Tr(Y)=m}} \Tr(YH)
=
\sum_{i=1}^m \lambda_i(H).
\end{aligned}
\]
Because the eigenvalues of $H$ are the nonzero eigenvalues of $\bar\rho$, the right-hand side equals $s_m$.
Therefore
\[
\sum_{i=1}^m \lambda_i(\Omega)
=
\sum_{i=1}^m \lambda_i(G)
\le s_m.
\]
As this holds for every feasible $\Omega$, we conclude that
\[
U_\mu\le s_m.
\]

\paragraph{Step 3: the isometric optimum is at least $B_\mu$.}
Fix a rank-$m$ orthogonal projector $P$ on $\C^d$ and a unit vector $\ket{\xi}\in \im(P)$.
Choose an isometry $V_P:\C^m\to \C^d$ such that
\[
V_PV_P^\dagger=P,
\qquad
V_P\ket{0}=\ket{\xi}.
\]
Define an encoder and decoder by
\[
\calE_{P,\xi}(\rho)
:=
V_P^\dagger P\rho P V_P + \Tr\bigl((I-P)\rho\bigr)\proj{0},
\qquad
\calD_{V_P}(\sigma):=V_P\sigma V_P^\dagger.
\]
The map $E_{P,\xi}$ is CPTP because it is the sum of two completely positive maps and
\[
\Tr\bigl(\calE_{P,\xi}(\rho)\bigr)
=
\Tr(P\rho P)+\Tr\bigl((I-P)\rho\bigr)
=
\Tr(\rho).
\]
For a pure input state $\proj{\psi}$, write $p:=p_P(\psi)=\bra{\psi}P\ket{\psi}$.
Using $V_PV_P^\dagger=P$ and $V_P\proj{0}V_P^\dagger=\proj{\xi}$, we obtain
\[
(\calD_{V_P}\circ \calE_{P,\xi})(\proj{\psi})
=
P\proj{\psi}P + (1-p)\proj{\xi}.
\]
Therefore
\[
\bra{\psi}(\calD_{V_P}\circ \calE_{P,\xi})(\proj{\psi})\ket{\psi}
=
p^2 + (1-p)\abs{\braket{\xi}{\psi}}^2.
\]
Averaging over $\mu$ gives
\[
F_\mu(\calE_{P,\xi},\calD_{V_P})
=
\mathbb{E}_{\ket{\psi}\sim\mu}[p_P(\psi)^2]
+
\bra{\xi}Y_P\ket{\xi}.
\]
Since $Y_P$ is supported on $\im(P)$,
\[
\sup_{\substack{\ket{\xi}\in \im(P)\\ \norm{\xi}=1}}
F_\mu(\calE_{P,\xi},\calD_{V_P})
=
\mathbb{E}_{\ket{\psi}\sim\mu}[p_P(\psi)^2]
+
\lambda_{\max}(Y_P).
\]
Taking the supremum over all rank-$m$ projectors $P$ proves that
\[
F^\star_{\mathrm{iso}}(\mu)\ge B_\mu.
\]

\paragraph{Step 4: a lower bound on $B_\mu$.}
Let $P_\star$ be the orthogonal projector onto the span of the eigenvectors of $\bar\rho$ corresponding to its $m$ largest eigenvalues.
Then
\[
\mathbb{E}_{\ket{\psi}\sim\mu}[p_{P_\star}(\psi)]
=
\Tr(P_\star\bar\rho)
=
s_m.
\]
Because $Y_{P_\star}\ge 0$ and $Y_{P_\star}$ is supported on the $m$-dimensional space $\im(P_\star)$,
\[
\lambda_{\max}(Y_{P_\star})
\ge
\frac{1}{m}\Tr(Y_{P_\star})
=
\frac{1}{m}
\mathbb{E}_{\ket{\psi}\sim\mu}\bigl[p_{P_\star}(\psi)-p_{P_\star}(\psi)^2\bigr].
\]
Hence
\begin{align}
\label{eq:B_mu_bound_almost}
B_\mu
&\ge
\mathbb{E}_{\ket{\psi}\sim\mu}\bigl[p_{P_\star}(\psi)^2\bigr]
+
\lambda_{\max}(Y_{P_\star}) \nonumber \\
&\ge
\left(1-\frac{1}{m}\right)
\mathbb{E}_{\ket{\psi}\sim\mu}\bigl[p_{P_\star}(\psi)^2\bigr]
+
\frac{1}{m}
\mathbb{E}_{\ket{\psi}\sim\mu}[p_{P_\star}(\psi)].
\end{align}
By Jensen's inequality,
\[
\mathbb{E}_{\ket{\psi}\sim\mu}\bigl[p_{P_\star}(\psi)^2\bigr]
\ge
\left(
\mathbb{E}_{\ket{\psi}\sim\mu}[p_{P_\star}(\psi)]
\right)^2
=
s_m^2.
\]
Substituting this and $\mathbb{E}[p_{P_\star}(\psi)]=s_m$ into Eq.~\eqref{eq:B_mu_bound_almost} yields
\[
B_\mu
\ge
\left(1-\frac{1}{m}\right)s_m^2+\frac{s_m}{m}.
\]

Combining Steps~1--4, we obtain
\[
\frac{F^\star_{\mathrm{iso}}(\mu)}{F^\star(\mu)}
\ge
\frac{B_\mu}{U_\mu}
\ge
\frac{B_\mu}{s_m}
\ge
\frac{\left(1-\frac{1}{m}\right)s_m^2+\frac{s_m}{m}}{s_m}
=
\left(1-\frac{1}{m}\right)s_m+\frac{1}{m}
=
1-\left(1-\frac{1}{m}\right)\eta_m.
\]
\end{proof}

\begin{remark}\label{rem:sm-optimal-average-state}
The converse $F^\star(\mu)\le s_m$ is optimal as a function of $\bar\rho$ alone.
Indeed, if
\[
\bar\rho = \sum_{i=1}^d \lambda_i(\bar\rho)\proj{i}
\]
in an orthonormal eigenbasis and $\mu$ is the orthogonal eigenensemble that draws $\ket{i}$ with probability $\lambda_i(\bar\rho)$, then a blind encoder can transmit the states $\ket{1},\ldots,\ket{m}$ perfectly through the $m$-dimensional bottleneck and map every $\ket{i}$ with $i>m$ to a fixed latent basis state.
With the obvious decoder, the resulting average fidelity is $\sum_{i=1}^m \lambda_i(\bar\rho)=s_m$.
Thus any further improvement of the bound
\[
\frac{F^\star_{\mathrm{iso}}(\mu)}{F^\star(\mu)}
\ge
1-\left(1-\frac{1}{m}\right)\eta_m
\]
must use information beyond the spectrum of $\bar\rho$ or must strengthen the lower bound on $F^\star_{\mathrm{iso}}(\mu)$.
\end{remark}

\section{MNIST data preparation}
\label{sec:appendix_mnist_data_prep}

\paragraph{Objects and dimensions.}
Each (rescaled) MNIST image is represented as a real matrix
\[
x \in [0,1]^{28\times 28}.
\]
Fix integers $n>k\ge 1$, and let $d=2^n<784$ and $m=2^k$. We map each image to a pure $n$-qubit state represented by a complex vector
\[
\psi \in \mathbb{C}^{d},\qquad \|\psi\|_2=1.
\]
Our goal is to concentrate most of the probability mass $\sum_{j=1}^{d}|\psi_j|^2$ into the first $m$ entries.
This data preparation scheme is designed to facilitate near-perfect reconstruction via a QAE with a bottleneck dimension of $2^k$. By concentrating the signal power into the low-frequency subspace, we ensure that the reconstruction infidelity is small. This approach enhances the sensitivity of our benchmarks, as small variations in model performance become more discernible when the total error magnitude is low.

\paragraph{Unitary-scaled 2D FFT feature map.}
We compute a 2D discrete Fourier transform of the image with scaling \citeMy{cooley1965algorithm}:
\[
F := \frac{1}{\sqrt{784}}\texttt{FFT2}(x)\ \in \mathbb{C}^{28\times 28}.
\]
We then apply an FFT shift
\[
\tilde F := \texttt{fftshift}(F),
\]
which permutes frequency bins so that the DC (zero-frequency) component is moved to the center of the $28\times 28$ grid.

\paragraph{Low-frequency-first ordering.}
Let the center index be $(c_y,c_x)=(14,14)$. For each frequency bin $(y,x)\in\{0,\dots,27\}^2$, define the squared radius from the center
\[
d(y,x)^2 := (y-c_y)^2 + (x-c_x)^2.
\]
We define an ordering (permutation) $\pi$ of the $784$ grid points by sorting bins in increasing order of $d(y,x)^2$ (with stable tie-breaking). Flattening $\tilde F$ into a length-$784$ vector and reordering by $\pi$ yields a sequence
\[
(a_1,a_2,\dots,a_{784}) \in \mathbb{C}^{784},
\]
where $a_1$ is the DC component, and increasing indices correspond to increasing spatial frequency. This low-frequency-first rule reflects the empirical fact that low spatial frequencies often capture most of the reconstruction-relevant content of images.
We retain only the first $d$ coefficients $(a_1,\dots,a_d)$. This discards higher-frequency coefficients (lossy truncation).

\paragraph{Head/tail placement and energy allocation.}
Split the retained coefficients into a ``head'' and ``tail'':
\[
u := (a_1,\dots,a_m)\in\mathbb{C}^m,\qquad
v := (a_{m+1},\dots,a_d)\in\mathbb{C}^{d-m}.
\]
We build an unnormalized state $\widehat\psi\in\mathbb{C}^d$ by placing the head coefficients in the first $m$ amplitudes and the tail coefficients next (optionally rescaled):
\[
\widehat\psi_j =
\begin{cases}
u_{j}, & 1\le j \le m,\\
\alpha\, v_{j-m}, & m+1 \le j \le d.
\end{cases}
\]
Given a target head energy fraction $e\in(0,1)$, we choose $\alpha\ge 0$ so that, after normalization, the fraction of probability mass in the head equals $e$. Specifically, we choose $e=0.9$ and set
\[
\alpha = \sqrt{\frac{1-e}{e}} \sqrt{\frac{\|u\|^2}{\|v\|^2}}.
\]

\paragraph{Normalization and output state.}
Finally we normalize to obtain a valid quantum state
\[
\psi := \frac{\widehat\psi}{\|\widehat\psi\|_2}.
\]
In non-degenerate cases, the construction ensures
$
\sum_{j=1}^{m} |\psi_j|^2 = e,
$
so a projective measurement onto the head subspace (spanned by the first $2^k$ computational basis vectors) succeeds with probability $e$. Moreover, because the head entries store the lowest-frequency Fourier coefficients, truncating to the first $2^k$ amplitudes corresponds to keeping a low-frequency approximation of the image, which is a simple proxy for preserving reconstruction-relevant information.

\clearpage
\section{Training time}
\label{sec:training-time}

\begin{table}[h]
\caption{Training time (in seconds) across different configurations. For $\mu_{1,0.1}$, cfg1--cfg3 denote $(n,k,0,n-k)$, $(n,k,0,n)$, and $(n,k,k,n)$-QAEs, respectively. For $\mu_2$, they denote $(n,k,0,n-k)$, $(n,k,k,n-k)$, and $(n,k,k,n-k+1)$-QAEs.}
\vspace{3mm}
\label{tab:training-times-simple}
\centering
\begin{tabular}{lcc|ccc}
\toprule
Dataset & $n$ & $k$ & cfg1 & cfg2 & cfg3 \\
\midrule
$\mu_{1,0.1}$ & 2 & 1 & 129 & 102 & 115\\
      & 3 & 1 & 108 & 121 & 126\\
      & 3 & 2 & 105 & 127 & 145\\
      & 4 & 1 & 121 & 134 & 146\\
      & 4 & 2 & 121 & 152 & 182\\
      & 4 & 3 & 121 & 198 & 276\\
\midrule
$\mu_2$ & 2 & 1 & 154 & 119 & 122\\
        & 3 & 1 & 121 & 126 & 128\\
        & 3 & 2 & 120 & 134 & 136\\
        & 4 & 1 & 126 & 137 & 148\\
        & 4 & 2 & 126 & 156 & 166\\
        & 4 & 3 & 126 & 194 & 204\\
\bottomrule
\end{tabular}
\end{table}

\section{Additional numerical experiments}
\label{app:additional-experiments}

We report the additional experiments in Section~\ref{sec:additional-experiments}. Each table compares $(n,k,0,n-k)$-, $(n,k,k,n-k)$-, and $(n,k,k,n-k+1)$-QAEs for $n\in\{2,3,4,5\}$ and $1\le k<n$. The middle architecture has an isometric decoder, and the final architecture permits a decoder of Kraus rank at most two.

For an evaluation ensemble $\{\ket{\psi_j}\}_{j=1}^{N_{\mathrm{ev}}}$ and a reconstructed state $\rho_j^{\mathrm{rec}}$, define the empirical average infidelity by
\begin{equation}
\label{eq:additional-empirical-infidelity}
\widehat L_{\mathrm{inf}}
:=\frac{1}{N_{\mathrm{ev}}}\sum_{j=1}^{N_{\mathrm{ev}}}
\left(1-\bra{\psi_j}\rho_j^{\mathrm{rec}}\ket{\psi_j}\right).
\end{equation}
The first three tables report this metric. The fourth reports average trace distance.

\subsection{Increased-tail MNIST}
\label{app:larger-tail}

The preparation in Appendix~\ref{sec:appendix_mnist_data_prep} allocates a prescribed mass to a fixed coordinate subspace. Write $d=2^n$, $m=2^k$, and $\tau:=1-e$. The original preparation uses $\tau=0.1$. Increasing this allocation to $\tau=0.2$ corresponds to $e=0.8$.
Table~\ref{tab:larger-tail} tests sensitivity to the larger-tail preparation.

\begin{table}[htbp]
\centering
\small
\setlength{\tabcolsep}{7pt}
\caption{Final average infidelity for the increased-tail MNIST setting.}
\label{tab:larger-tail}
\vspace{10pt}
\begin{tabular}{@{}ccccc@{}}
\toprule
 & & \multicolumn{3}{c}{Ancilla counts $(n_B,n_E)$} \\
\cmidrule(l){3-5}
$n$ & $k$ & $(0,n-k)$ & $(k,n-k)$ & $(k,n-k+1)$ \\
\midrule
2 & 1 & 0.002000 & 0.001903 & 0.001904 \\
\addlinespace[2pt]
3 & 1 & 0.025293 & 0.022532 & 0.022548 \\
3 & 2 & 0.022232 & 0.016264 & 0.016277 \\
\addlinespace[2pt]
4 & 1 & 0.057153 & 0.054477 & 0.054491 \\
4 & 2 & 0.053883 & 0.044734 & 0.044748 \\
4 & 3 & 0.022449 & 0.015791 & 0.015802 \\
\addlinespace[2pt]
5 & 1 & 0.090231 & 0.087808 & 0.087839 \\
5 & 2 & 0.098216 & 0.084887 & 0.084977 \\
5 & 3 & 0.077549 & 0.060488 & 0.060556 \\
5 & 4 & 0.027539 & 0.021319 & 0.021364 \\
\bottomrule
\end{tabular}
\end{table}

\subsection{Coherently perturbed GHZ states}
\label{app:noisy-ghz}

Let
\begin{equation}
\label{eq:ghz-source}
\ket{\mathrm{GHZ}_n}:=
\frac{\ket{0}^{\otimes n}+\ket{1}^{\otimes n}}{\sqrt2},
\qquad
\ket{G_{n,\xi}}:=U_{n,\xi}\ket{\mathrm{GHZ}_n},
\qquad U_{n,\xi}\sim\nu_{\mathrm{pert},n},
\end{equation}
where $\nu_{\mathrm{pert},n}$ denotes the distribution of coherent perturbation unitaries used to prepare the samples. The resulting source distribution is the pushforward of $\nu_{\mathrm{pert},n}$ under $U\mapsto U\ket{\mathrm{GHZ}_n}$. Specifically, we set

\begin{gather*}
U_{n,\xi} = \bigotimes_{i=1}^n U_i, 
\quad 
U_i = \exp\!\left[-\frac{i\theta_i}{2}(r_{i,x}X + r_{i,y}Y + r_{i,z}Z)\right],\\
\mathbf{r}_i = \frac{\mathbf{g}_i}{\|\mathbf{g}_i\|_2},
\quad
\mathbf{g}_i \sim \mathcal{N}(\mathbf{0}, I_3),
\quad
\theta_i \sim \mathcal{N}(0, 0.3^2).
\end{gather*}
Table~\ref{tab:noisy-ghz} provides the results.

\begin{table}[htbp]
\centering
\small
\setlength{\tabcolsep}{7pt}
\caption{Final average infidelity for the coherently perturbed GHZ source of Eq.~\eqref{eq:ghz-source}.}
\label{tab:noisy-ghz}
\vspace{10pt}
\begin{tabular}{@{}ccccc@{}}
\toprule
 & & \multicolumn{3}{c}{Ancilla counts $(n_B,n_E)$} \\
\cmidrule(l){3-5}
$n$ & $k$ & $(0,n-k)$ & $(k,n-k)$ & $(k,n-k+1)$ \\
\midrule
2 & 1 & 0.005028 & 0.003407 & 0.003337 \\
\addlinespace[2pt]
3 & 1 & 0.007157 & 0.004795 & 0.004803 \\
3 & 2 & 0.005545 & 0.003195 & 0.003214 \\
\addlinespace[2pt]
4 & 1 & 0.007575 & 0.006745 & 0.006759 \\
4 & 2 & 0.007603 & 0.005081 & 0.005084 \\
4 & 3 & 0.002508 & 0.001718 & 0.001728 \\
\addlinespace[2pt]
5 & 1 & 0.008179 & 0.008176 & 0.008196 \\
5 & 2 & 0.007452 & 0.006587 & 0.006605 \\
5 & 3 & 0.004200 & 0.003324 & 0.003374 \\
5 & 4 & 0.000046 & 0.000036 & 0.000037 \\
\bottomrule
\end{tabular}
\end{table}

\subsection{Output depolarization}
\label{app:depolarization}

Let $\mu_{\mathrm{dep}}=\mu_2$ denote the MNIST source used in this experiment. For
$\rho_\psi^{\mathrm{rec}}:=(\Dec\circ\Enc)(\proj{\psi})$, output noise acts as
\begin{equation}
\label{eq:local-output-noise}
\rho_{\psi,p}^{\mathrm{out}}
:=\mathcal N_p^{\otimes n}(\rho_\psi^{\mathrm{rec}}),
\qquad p=0.01,
\end{equation}
where
\begin{equation}
\label{eq:depolarizing-replacement-convention}
\mathcal N_p(X)
=(1-p)X+p\Tr(X)\frac{I_2}{2}
\end{equation}
for a single-qubit operator $X$.
The evaluation loss is
\begin{equation}
\label{eq:noisy-output-infidelity}
L_{\mathrm{dep}}(\Enc,\Dec)
:=1-\mathbb E_{\psi\sim\mu_{\mathrm{dep}}}
\bra{\psi}\rho_{\psi,p}^{\mathrm{out}}\ket{\psi}.
\end{equation}
The target remains the original pure source state.
Table~\ref{tab:depolarization} measures reconstruction after this output channel.

\begin{table}[htbp]
\centering
\small
\setlength{\tabcolsep}{7pt}
\caption{Final average infidelity after independent single-qubit output depolarization.}
\label{tab:depolarization}
\vspace{10pt}
\begin{tabular}{@{}ccccc@{}}
\toprule
 & & \multicolumn{3}{c}{Ancilla counts $(n_B,n_E)$} \\
\cmidrule(l){3-5}
$n$ & $k$ & $(0,n-k)$ & $(k,n-k)$ & $(k,n-k+1)$ \\
\midrule
2 & 1 & 0.012446 & 0.012387 & 0.012388 \\
\addlinespace[2pt]
3 & 1 & 0.029223 & 0.027760 & 0.027769 \\
3 & 2 & 0.035319 & 0.030458 & 0.030459 \\
\addlinespace[2pt]
4 & 1 & 0.049466 & 0.047943 & 0.047964 \\
4 & 2 & 0.054067 & 0.048149 & 0.048167 \\
4 & 3 & 0.041732 & 0.036295 & 0.036307 \\
\addlinespace[2pt]
5 & 1 & 0.070251 & 0.068832 & 0.068845 \\
5 & 2 & 0.080682 & 0.072489 & 0.072538 \\
5 & 3 & 0.073797 & 0.062757 & 0.062785 \\
5 & 4 & 0.050734 & 0.046040 & 0.046065 \\
\bottomrule
\end{tabular}
\end{table}

\subsection{Trace distance}
\label{app:trace-experiment}

Let $\mu_{\mathrm{tr}}=\mu_2$ denote the MNIST source used in this experiment. Define
\begin{equation}
\label{eq:trace-experiment-loss}
L_{\mathrm{tr},\mu_{\mathrm{tr}}}(\Enc,\Dec)
:=\mathbb E_{\psi\sim\mu_{\mathrm{tr}}}
\frac12\left\|(\Dec\circ\Enc)(\proj{\psi})-\proj{\psi}\right\|_1,
\end{equation}
where $\|X\|_1:=\Tr\sqrt{X^\dagger X}$. The empirical evaluation is the average of the individual trace distances:
\begin{equation}
\label{eq:trace-experiment-empirical}
\widehat L_{\mathrm{tr}}
:=\frac{1}{N_{\mathrm{ev}}}\sum_{j=1}^{N_{\mathrm{ev}}}
\frac12\|\rho_j^{\mathrm{rec}}-\proj{\psi_j}\|_1.
\end{equation}

Table~\ref{tab:trace-distance} examines the architecture comparison under this alternative metric. The relationship to infidelity is explained in Appendix~\ref{app:metric-transfer}.

\begin{table}[t]
\centering
\small
\setlength{\tabcolsep}{7pt}
\caption{Final average trace distance using the normalization $\tfrac12\|\rho-\sigma\|_1$.}
\label{tab:trace-distance}
\vspace{10pt}
\begin{tabular}{@{}ccccc@{}}
\toprule
 & & \multicolumn{3}{c}{Ancilla counts $(n_B,n_E)$} \\
\cmidrule(l){3-5}
$n$ & $k$ & $(0,n-k)$ & $(k,n-k)$ & $(k,n-k+1)$ \\
\midrule
2 & 1 & 0.024574 & 0.024565 & 0.024562 \\
\addlinespace[2pt]
3 & 1 & 0.100797 & 0.100062 & 0.100103 \\
3 & 2 & 0.100240 & 0.097757 & 0.097708 \\
\addlinespace[2pt]
4 & 1 & 0.159591 & 0.158924 & 0.158906 \\
4 & 2 & 0.150254 & 0.147490 & 0.147421 \\
4 & 3 & 0.097383 & 0.094570 & 0.094630 \\
\addlinespace[2pt]
5 & 1 & 0.206730 & 0.206119 & 0.206158 \\
5 & 2 & 0.207809 & 0.204181 & 0.204241 \\
5 & 3 & 0.176198 & 0.170996 & 0.171051 \\
5 & 4 & 0.102333 & 0.099928 & 0.100112 \\
\bottomrule
\end{tabular}
\end{table}

\section{Quantitative trace-distance transfer}
\label{app:metric-transfer}

\subsection{Pointwise and average trace-distance bounds}

Set
\[
\rho_\psi:=(\Dec\circ\Enc)(\proj\psi),
\qquad
\ell_{\mathrm{inf}}(\psi):=1-\bra\psi\rho_\psi\ket\psi,
\qquad
T(\rho,\sigma):=\frac12\|\rho-\sigma\|_1.
\]
The upper bound below is the Fuchs--van de Graaf bound \citeMy{fuchs1999cryptographic}, and the lower bound follows directly from the variational characterization of the trace distance \citeMy{nielsen2010quantum}:
\begin{equation}
\label{eq:pure-target-trace-bounds}
\ell_{\mathrm{inf}}(\psi)
\le T(\rho_\psi,\proj\psi)
\le\sqrt{\ell_{\mathrm{inf}}(\psi)}.
\end{equation}

Let $L_{\mathrm{inf},\mu}:=L_\mu$ and
$L_{\mathrm{tr},\mu}:=\mathbb E_{\psi\sim\mu}T(\rho_\psi,\proj\psi)$.
Averaging Eq.~\eqref{eq:pure-target-trace-bounds} and applying Jensen's inequality yields
\begin{equation}
\label{eq:average-trace-transfer}
L_{\mathrm{inf},\mu}(\Enc,\Dec)
\le L_{\mathrm{tr},\mu}(\Enc,\Dec)
\le\sqrt{L_{\mathrm{inf},\mu}(\Enc,\Dec)}.
\end{equation}
Equation~\eqref{eq:average-trace-transfer} also holds for empirical averages evaluated on the same samples and the same reconstructed states.

\subsection{Transfer to the width-constrained architecture}

Let $\mathfrak C$ be the set of all CPTP encoder--decoder pairs through the bottleneck, and let $\mathfrak G\subseteq\mathfrak C$ be the pairs realizable by $(n,k,k,n)$-QAEs. For $a\in\{\mathrm{inf},\mathrm{tr}\}$ and a class $\mathfrak A$, write
\[
L_{a,\mathfrak A}^{\star}
:=\min_{(\Enc,\Dec)\in\mathfrak A}L_{a,\mu}(\Enc,\Dec).
\]
Theorem~\ref{thm:sufficiency} supplies a pair
$(\Enc_F,\Dec_F)\in\mathfrak G$ satisfying
$L_{\mathrm{inf},\mu}(\Enc_F,\Dec_F)=L_{\mathrm{inf},\mathfrak C}^{\star}$.
Consequently,
\begin{equation}
\label{eq:goldilocks-trace-transfer}
\begin{aligned}
L_{\mathrm{tr},\mathfrak C}^{\star}
&\le L_{\mathrm{tr},\mathfrak G}^{\star}
\le L_{\mathrm{tr},\mu}(\Enc_F,\Dec_F)\\
&\le\sqrt{L_{\mathrm{inf},\mathfrak C}^{\star}}
\le\sqrt{L_{\mathrm{tr},\mathfrak C}^{\star}}.
\end{aligned}
\end{equation}
Thus the infidelity-optimal architecture provides a quantitative trace-distance guarantee, but the inequalities do not establish
$L_{\mathrm{tr},\mathfrak G}^{\star}=L_{\mathrm{tr},\mathfrak C}^{\star}$.
An approximately infidelity-optimal trained pair satisfying
$$
L_{\mathrm{inf},\mu}(\Enc,\Dec)\le L_{\mathrm{inf},\mathfrak C}^{\star}+\delta
$$
also obeys
$$
L_{\mathrm{tr},\mu}(\Enc,\Dec)\le\sqrt{L_{\mathrm{inf},\mathfrak C}^{\star}+\delta}\le\sqrt{L_{\mathrm{tr},\mathfrak C}^{\star}+\delta}.
$$

\subsection{A source-dependent bound with an isometric decoder}

Let $\mathfrak G_{\mathrm{iso}}$ denote the $(n,k,k,n-k)$-QAE class. The construction used in Theorem~\ref{thm:isometric-multiplicative} belongs to this class and achieves at least
\[
b_m=\left(1-\frac1m\right)(1-\eta_m)^2+\frac{1-\eta_m}{m},
\qquad m=2^k.
\]
Applying Eq.~\eqref{eq:average-trace-transfer} to that construction gives
\begin{equation}
\label{eq:isometric-trace-tail}
L_{\mathrm{tr},\mathfrak G_{\mathrm{iso}}}^{\star}
\le\sqrt{1-b_m}
=\sqrt{\left(2-\frac1m\right)\eta_m
-\left(1-\frac1m\right)\eta_m^2}.
\end{equation}

% ------------------------

% \newpage
\clearpage
\section*{NeurIPS Paper Checklist}

\begin{enumerate}

\item {\bf Claims}
    \item[] Question: Do the main claims made in the abstract and introduction accurately reflect the paper's contributions and scope?
    \item[] Answer: \answerYes{} % Replace by \answerYes{}, \answerNo{}, or \answerNA{}.
    \item[] Justification: This work studies blind single-copy compression of pure \(n\)-qubit states under average infidelity, proves universal sufficiency of the \((n,k,k,n)\)-QAE architecture, establishes a sharp encoder-side ancilla threshold, and presents an explicit counterexample showing that isometric decoders are not universally sufficient. The claimed scope is also described appropriately, because the paper explicitly limits the theory to the infidelity objective and notes that a general sharp decoder threshold for arbitrary \(k\) remains open.

\item {\bf Limitations}
    \item[] Question: Does the paper discuss the limitations of the work performed by the authors?
    \item[] Answer: \answerYes{} % Replace by \answerYes{}, \answerNo{}, or \answerNA{}.
    \item[] Justification: The paper explicitly discusses its limitations in the Discussion section by stating that the analysis is restricted to blind single-copy compression of pure \(n\)-qubit states under the average infidelity objective, and that it does not address additional objectives such as latent-space regularization.

\item {\bf Theory assumptions and proofs}
    \item[] Question: For each theoretical result, does the paper provide the full set of assumptions and a complete (and correct) proof?
    \item[] Answer: \answerYes{} % Replace by \answerYes{}, \answerNo{}, or \answerNA{}.
    \item[] Justification: The Main results section specifies the source model, bottleneck size, ancilla counts, and channel classes. We provide short proof sketches in the main text together with full proofs in the appendices.

    \item {\bf Experimental result reproducibility}
    \item[] Question: Does the paper fully disclose all the information needed to reproduce the main experimental results of the paper to the extent that it affects the main claims and/or conclusions of the paper (regardless of whether the code and data are provided or not)?
    \item[] Answer: \answerYes{} % Replace by \answerYes{}, \answerNo{}, or \answerNA{}.
    \item[] Justification: The paper specifies the experimental setups that support its main claims, including which QAE architectures are compared, the source distributions used, the training/test sample sizes, the number of epochs, etc. The paper also gives the MNIST state-preparation pipeline in the appendix.
    \item[] Guidelines:
    \begin{itemize}
        \item The answer \answerNA{} means that the paper does not include experiments.
        \item If the paper includes experiments, a \answerNo{} answer to this question will not be perceived well by the reviewers: Making the paper reproducible is important, regardless of whether the code and data are provided or not.
        \item If the contribution is a dataset and\slash or model, the authors should describe the steps taken to make their results reproducible or verifiable. 
        \item Depending on the contribution, reproducibility can be accomplished in various ways. For example, if the contribution is a novel architecture, describing the architecture fully might suffice, or if the contribution is a specific model and empirical evaluation, it may be necessary to either make it possible for others to replicate the model with the same dataset, or provide access to the model. In general, releasing code and data is often one good way to accomplish this, but reproducibility can also be provided via detailed instructions for how to replicate the results, access to a hosted model (e.g., in the case of a large language model), releasing of a model checkpoint, or other means that are appropriate to the research performed.
        \item While NeurIPS does not require releasing code, the conference does require all submissions to provide some reasonable avenue for reproducibility, which may depend on the nature of the contribution. For example
        \begin{enumerate}
            \item If the contribution is primarily a new algorithm, the paper should make it clear how to reproduce that algorithm.
            \item If the contribution is primarily a new model architecture, the paper should describe the architecture clearly and fully.
            \item If the contribution is a new model (e.g., a large language model), then there should either be a way to access this model for reproducing the results or a way to reproduce the model (e.g., with an open-source dataset or instructions for how to construct the dataset).
            \item We recognize that reproducibility may be tricky in some cases, in which case authors are welcome to describe the particular way they provide for reproducibility. In the case of closed-source models, it may be that access to the model is limited in some way (e.g., to registered users), but it should be possible for other researchers to have some path to reproducing or verifying the results.
        \end{enumerate}
    \end{itemize}

\item {\bf Open access to data and code}
    \item[] Question: Does the paper provide open access to the data and code, with sufficient instructions to faithfully reproduce the main experimental results, as described in supplemental material?
    \item[] Answer: \answerYes{} % Replace by \answerYes{}, \answerNo{}, or \answerNA{}.
    \item[] Justification: The code for model training and evaluation is included in the supplemental material (qae\_submit.zip).

\item {\bf Experimental setting/details}
    \item[] Question: Does the paper specify all the training and test details (e.g., data splits, hyperparameters, how they were chosen, type of optimizer) necessary to understand the results?
    \item[] Answer: \answerYes{} % Replace by \answerYes{}, \answerNo{}, or \answerNA{}.
    \item[] Justification: The Experiments section states the concrete experimental details needed to understand the results, including the train/test split, number of epochs, optimizer, learning rate, batch size, and hardware.

\item {\bf Experiment statistical significance}
    \item[] Question: Does the paper report error bars suitably and correctly defined or other appropriate information about the statistical significance of the experiments?
    \item[] Answer: \answerYes{} % Replace by \answerYes{}, \answerNo{}, or \answerNA{}.
    \item[] Justification: The main experimental figures report uncertainty bands explicitly as 2-sigma confidence intervals.

\item {\bf Experiments compute resources}
    \item[] Question: For each experiment, does the paper provide sufficient information on the computer resources (type of compute workers, memory, time of execution) needed to reproduce the experiments?
    \item[] Answer: \answerYes{} % Replace by \answerYes{}, \answerNo{}, or \answerNA{}.
    \item[] Justification: The Setup section reports the hardware used for the experiments, and the Training time section in the appendix reports training times across the reported dataset/configuration settings. This provides the type of compute workers, memory, and execution-time information.
    
\item {\bf Code of ethics}
    \item[] Question: Does the research conducted in the paper conform, in every respect, with the NeurIPS Code of Ethics \url{https://neurips.cc/public/EthicsGuidelines}?
    \item[] Answer: \answerYes{} % Replace by \answerYes{}, \answerNo{}, or \answerNA{}.
    \item[] Justification: The paper reports standard theoretical and numerical research on quantum-state compression, preserves submission anonymity, and does not involve human subjects, crowdsourcing, or high-risk model/data release.

\item {\bf Broader impacts}
    \item[] Question: Does the paper discuss both potential positive societal impacts and negative societal impacts of the work performed?
    \item[] Answer: \answerNA{} % Replace by \answerYes{}, \answerNo{}, or \answerNA{}.
    \item[] Justification: There is no societal impact of the work performed.
    
\item {\bf Safeguards}
    \item[] Question: Does the paper describe safeguards that have been put in place for responsible release of data or models that have a high risk for misuse (e.g., pre-trained language models, image generators, or scraped datasets)?
    \item[] Answer: \answerNA{} % Replace by \answerYes{}, \answerNo{}, or \answerNA{}.
    \item[] Justification: The paper poses no such risks.

\item {\bf Licenses for existing assets}
    \item[] Question: Are the creators or original owners of assets (e.g., code, data, models), used in the paper, properly credited and are the license and terms of use explicitly mentioned and properly respected?
    \item[] Answer: \answerYes{} % Replace by \answerYes{}, \answerNo{}, or \answerNA{}.
    \item[] Justification: We use the MNIST handwritten digit dataset and credit its creators in the paper. We also state that MNIST (version 3.0.1) is available under the CC BY-SA 3.0 license and use the standard dataset for benchmarking.

\item {\bf New assets}
    \item[] Question: Are new assets introduced in the paper well documented and is the documentation provided alongside the assets?
    \item[] Answer: \answerNA{} % Replace by \answerYes{}, \answerNo{}, or \answerNA{}.
    \item[] Justification: The paper does not release new assets.

\item {\bf Crowdsourcing and research with human subjects}
    \item[] Question: For crowdsourcing experiments and research with human subjects, does the paper include the full text of instructions given to participants and screenshots, if applicable, as well as details about compensation (if any)? 
    \item[] Answer: \answerNA{} % Replace by \answerYes{}, \answerNo{}, or \answerNA{}.
    \item[] Justification: The paper does not involve crowdsourcing nor research with human subjects.

\item {\bf Institutional review board (IRB) approvals or equivalent for research with human subjects}
    \item[] Question: Does the paper describe potential risks incurred by study participants, whether such risks were disclosed to the subjects, and whether Institutional Review Board (IRB) approvals (or an equivalent approval/review based on the requirements of your country or institution) were obtained?
    \item[] Answer: \answerNA{} % Replace by \answerYes{}, \answerNo{}, or \answerNA{}.
    \item[] Justification: The paper does not involve crowdsourcing nor research with human subjects.

\item {\bf Declaration of LLM usage}
    \item[] Question: Does the paper describe the usage of LLMs if it is an important, original, or non-standard component of the core methods in this research? Note that if the LLM is used only for writing, editing, or formatting purposes and does \emph{not} impact the core methodology, scientific rigor, or originality of the research, declaration is not required.
    %this research? 
    \item[] Answer: \answerNA{} % Replace by \answerYes{}, \answerNo{}, or \answerNA{}.
    \item[] Justification: The core method development in this research does not involve LLMs as any important, original, or non-standard components.

\end{enumerate}

\end{document}